\documentclass[a4paper,onecolumn,11pt,amsmath,amssymb,longbibliography,floatfix,nofootinbib,reprint,superscriptaddress, accepted=2026-01-15]{quantumarticle}
\pdfoutput=1
\usepackage[utf8]{inputenc}
\usepackage[english]{babel}
\usepackage[T1]{fontenc}

\usepackage{ulem}

\usepackage{tabularray}

\usepackage{graphicx}
\usepackage{xcolor}
\usepackage[colorlinks=true,citecolor=blue,linkcolor=magenta]{hyperref}
\usepackage{bm}
\usepackage{tikz}
\usepackage{mathtools}
\usepackage{comment}
\usepackage{caption}
\usepackage{subcaption}
\usepackage{footmisc}
\usepackage{amsthm}
 
\DeclareMathOperator{\gme}{GME}
\DeclareMathOperator{\ghz}{GHZ}

\newcommand{\Lt}{\mathcal{L}^{\tau}}
\newcommand{\Lgme}{\mathcal{L}^{\mathrm{C}}}
\newcommand{\Lce}{\mathcal{L}^{\mathrm{CE}}}
\newcommand{\Lsqrtce}{\mathcal{L}^{\sqrt{\mathrm{CE}}}}

\newcommand{\AME}{\mathrm{AME}}
\DeclareMathOperator{\tr}{tr}
\DeclareMathOperator{\rk}{rank}
\DeclareMathOperator{\supp}{supp}
\DeclareMathOperator{\Span}{span}
\DeclareMathOperator{\Haar}{Haar}
\newcommand{\sumi}{\sum\nolimits}
\DeclareMathOperator*{\EV}{\mathbb{E}}

\newtheorem{theorem}{Theorem}
\newtheorem{lemma}[theorem]{Lemma}
\newtheorem{corollary}[theorem]{Corollary}

\newcommand{\CE}[2]{C\!\left(#1; #2\right)}

\allowdisplaybreaks

\begin{document}

\title{Localizing multipartite entanglement with local and global measurements}
\author{Christopher Vairogs}
\email{christopher.vairogs@gmail.com}
\affiliation{Department of Physics, University of Illinois Urbana-Champaign, Urbana, Illinois 61801, USA}
\author{Samihr Hermes}
\affiliation{Department of Physics, University of Illinois Urbana-Champaign, Urbana, Illinois 61801, USA}
\author{Felix Leditzky}
\affiliation{Department of Mathematics, University of Illinois Urbana-Champaign, Urbana, Illinois 61801, USA}
\affiliation{IQUIST, University of Illinois Urbana-Champaign, Urbana, Illinois 61801, USA}


\begin{abstract}
We study the task of localizing multipartite entanglement in pure quantum states onto a subsystem by measuring the remaining systems. To this end, we fix a multipartite entanglement measure and consider two quantities: the multipartite entanglement of assistance (MEA), defined as the entanglement measure averaged over the post-measurement states and maximized over arbitrary measurements; and the localizable multipartite entanglement (LME), defined in the same way but restricted to only local single-system measurements.
Both quantities generalize previously considered bipartite entanglement localization measures.
In our work we choose the $n$-tangle, the genuine multipartite entanglement concurrence and the concentratable entanglement (CE) as the underlying seed measure, and discuss the resulting MEA and LME quantities.
First, we prove easily computable upper and lower bounds on MEA and LME and establish Lipschitz-continuity for the $n$-tangle and CE-based LME and MEA.
Using these bounds we investigate the typical behavior of entanglement localization by deriving concentration inequalities for the MEA evaluated on Haar-random states and performing numerical studies for small tractable system sizes. We then turn our attention to protocols that transform graph states. We give a simple criterion based on a matrix equation to decided whether states with a specified $n$-tangle value can be obtained from a given graph state, providing no-go theorems for a broad class of such graph state transformations beyond the usual ``local Clifford plus local Pauli measurement'' framework. This analysis is generalized to weighted graph states, which provide a realistic error model in current experiments preparing graph state. Our entanglement localization framework certifies the near-optimality of recently discussed local-measurement protocols to transform uniformly weighted line graph states into GHZ states, even when considering arbitrary entangled measurements. Finally, we demonstrate how our MEA and LME quantities can be used to detect critical phenomena such as phase transitions in transversal field Ising models. Since entanglement localization is operationally relevant throughout quantum networking and measurement-based quantum computation, our framework of results based on the MEA and LME has the potential for broad applications in these fields.
\end{abstract}

\maketitle 

\newpage
\tableofcontents

\section{Introduction}
Entanglement is the quintessential feature distinguishing quantum information theory from its classical counterpart, enabling various quantum information-processing tasks such as quantum teleportation \cite{bennett1993teleporting,ishizaka2008asymptotic,ishizaka2009quantum}, quantum communication over long distances \cite{duer1999repeaters} or secret sharing protocols \cite{hillery1999secret,cleve1999secret}.
These protocols all require certain types of entanglement resources, but preparing high-fidelity multipartite entangled states remains a major experimental challenge \cite{duan2004scalable,nysteen2017limitations,pichler2017universal,Firstenberg2013, Firstenberg2016, Tiarks2016, Thompson2017, Sagona-Stophel2020, Tiarks2018}.

We are thus presented with the task of transforming multipartite states produced by some entanglement source into a suitable form for further use in quantum information-processing protocols.
There are different approaches to achieving this task.
In entanglement distillation (also called purification) one aims to convert many independent and identically distributed copies of a noisy source state into a smaller number of pure target states \cite{bennett1996purification,bennett1996concentrating,bennett1996mixed}.
In a different, ``single-copy'' approach, measurements are performed on some subset of the quantum systems in order to prepare the target state on the remaining systems, thus increasing or \emph{localizing} entanglement on a smaller system \cite{Divincenzo1998,raussendorf2001oneway,raussendorf2003measurement,Lausten2003,Verstraete-LE,Popp-LE}.
Recently, this measurement-based entanglement localization approach has been employed to study protocols that extract highly entangled multipartite states such as GHZ states from less entangled source states \cite{paunkovic2002entanglement,hwang2007practical,xiong2011schemes,deng2012optimal,Zhou2013,Choudhury2013}.
For example, the recent work \cite{Frantzeskakis2023} describes a protocol to extract GHZ states from larger multipartite states with a linear correlation structure by suitably measuring out auxiliary qubits.

In our work we focus on entanglement localization and aim to devise benchmarks for such protocols using tools from entanglement theory.
To this end, we define multipartite versions of the entanglement of assistance \cite{Divincenzo1998, Lausten2003,Smolin2005,Gour2006} and the localizable entanglement \cite{Verstraete-LE, Popp-LE}.
Both concepts quantify the maximal amount of entanglement, with respect to a chosen multipartite entanglement measure, that can be localized on a subset of qubits via measurements on the remaining systems. We consider two possibilities for these  measurements: local measurements and global measurements. Naturally, local measurements, in which single qudits are measured individually, form a proper subset of global measurements, in which multiple parties may be measured jointly.
For the entanglement of assistance we allow global measurements, whereas for the localizable entanglement we restrict to local measurements.
In this paper, we set out to characterize these entanglement localization quantities in terms of the $n$-tangle~\cite{Wong2001, Jaeger2003Invariance, Jaeger2003Mixedness, Teodorescu2003}, the genuine multipartite entanglement (GME) concurrence~\cite{Ma2011, Dai2020}, and the concentratable entanglement~\cite{Beckey2021}. Each of these functions are bona fide entanglement monotones in the sense that \textit{(i)} they vanish on product states and \textit{(ii)} they do not increase, on average, under local operations and classical communication (LOCC)~\cite{vidal2000monotones,Horodecki2009}, which makes them suitable seed measures for our entanglement localization measures.
Furthermore, the favorable properties of these entanglement measures allow us to gain new insights into multipartite entanglement localization that significantly extend previous studies \cite{sadhukhan2017multipartite,banerjee2021localizing,harikrishnan2023localizing,amaro2018estimating,amaro2020scalable,banerjee2020uniform,banerjee2022hierarchies,krishnan2023controlling}. Many of our results center around the entanglement localization measures defined in terms of the $n$-tangle due to its many exploitable features, and our results for the GME-concurrence and CE are provided for further context.
We give an overview of our main findings in the following paragraph.

\paragraph*{Main results and structure of this paper.}
We show in this work that our multipartite entanglement localization measures give useful tools to benchmark multipartite state transformations and detect critical phenomena in quantum spin systems.
In Sec.~\ref{sec:entanglement-localization-measures} we define the two main quantities of this work, the multipartite entanglement of assistance (MEA) and the localizable multipartite entanglement (LME) in terms of three entanglement measures: the $n$-tangle \eqref{eq:n-tangle-def}, the genuine multipartite entanglement concurrence \eqref{eq:gme-conc-def}, and the concentratable entanglement \eqref{eq:CE}. 
In Sec.~\ref{sec:bounds} we derive easily computable upper (Thms.~\ref{Th:n-tangle-upper-bound}, \ref{Th:cgme-upper-bound}, \ref{Th:CE-upper-bound}) and lower bounds (Thm.~\ref{Th:corr-func-lower-bound}) as technical tools.
In Sec.~\ref{sec:continuity-bounds} we establish Lipschitz-continuity for these quantities (Thm.~\ref{Th:LE-continuity-bound} and Cor.~\ref{cor:LE-local-bounds}).
In Sec.~\ref{sec:concentration} we discuss generic aspects of localizing multipartite entanglement (Thms.~\ref{Th:EoA-concentration} and \ref{Th:EoA-concentration-2}), and we provide a numerical study of Haar-random states in Sec.~\ref{sec:numerics}.
We then focus on protocols transforming (weighted) graph states in Sec.~\ref{sec:graph-states}, proving efficient criteria for determining what state transformations are possible (Thm.~\ref{Th:main-graph-thm}) and showing that the aforementioned GHZ-extraction protocol of~\cite{Frantzeskakis2023} is close to optimal (Fig.~\ref{fig:mTangle-optimization}), even when considering global measurements.
Finally, we apply our methods to the study of critical phenomena of a spin-half transverse field Ising model in Sec.~\ref{sec:spin-models}, and conclude with a summary and open problems in Sec.~\ref{sec:discussion}.
The proofs of our main results are given in the appendices.

\section{Entanglement localization measures}
\label{sec:entanglement-localization-measures}

\subsection{Notation and basic definitions}

Given a multipartite pure state $|\Psi\rangle\in\mathcal{H}_1 \otimes \dots \otimes \mathcal{H}_N$ over $N$ systems, we consider measurements performed on some proper subset $A\subset [N] \coloneqq \{1, \dots, N\}$ of the $N$ systems. 
We make the following two restrictions in our work: First, we only consider projective measurements so that post-measurement states are well-defined.
Second, in order to ensure that the post-measurement states on the subset $A\subset [N]$ are pure, we further restrict to rank-1 measurements.
This allows us to directly evaluate pure-state entanglement measurements on the post-measurement states, thereby avoiding the issue of evaluating mixed-state entanglement measures (e.g., through convex-roof constructions \cite{Horodecki2009}) altogether, which would make our analysis infeasible.
More succinctly, we only consider von Neumann measurements on the subset $A\subset [N]$.

We denote the complement of $A$ in $[N]$ by $B$, and we write $N_A\coloneqq |A|$ and $N_B\coloneqq |B|$. Let $\mathcal{H}_A$ and $\mathcal{H}_B$ denote the Hilbert spaces of the $A$ and $B$ subsystems, respectively, and write $d_A \coloneqq \dim \mathcal{H}_A$ and $d_B \coloneqq \dim \mathcal{H}_B$. A rank-one projective measurement over the systems in $A$ produces an ensemble $\{(p_i, |\phi_i\rangle)\}_i$ of states over the unmeasured systems in $B$. We stress that the measurement bases in this definition are completely arbitrary and may be entangled over the subsystems in $A$. Given a pure state entanglement measure $E$, we may define the \textit{multipartite entanglement of assistance} (MEA) of $|\Psi\rangle$ with respect to $E$ to be\footnote{Despite the resemblance of $L^E(|\Psi\rangle)$ to a convex roof extension, we emphasize that this is a different quantity defined by a \textit{maximization} over post-measurement ensembles, i.e., decompositions of the reduced density matrix $\tr_A[|\Psi\rangle\langle \Psi|]$.
In contrast, a convex roof measure extends an entanglement measure on pure states to mixed states using a \textit{minimization} over pure-state decompositions of the given mixed state.}
\begin{equation}\label{eq:eoa-def}
    \mathcal{A}^E(|\Psi\rangle) \coloneqq \max\left\lbrace \sum\nolimits_{i} p_i E(|v_i\rangle): \{(p_i, |v_i\rangle)\}_i~\text{results from a global measurement}\right\rbrace.
\end{equation} 
That is, the MEA of $|\Psi\rangle$ is the maximum amount of entanglement, on average, that can be localized onto the subsystems in $B$ by performing rank-one projective measurements on the subsystems in $A$ that are \textit{global} in the sense that their measurement operators may be entangled across multiple subsystems in $A$. This definition is a direct generalization of the bipartite entanglement of assistance~\cite{Divincenzo1998, Lausten2003, Smolin2005, Gour2006}.

Since global measurements are operationally burdensome, it is practically relevant to restrict to \textit{local} projective measurement on the $A$ subsystems. For this reason, we define the \textit{localizable multipartite entanglement} (LME) of $|\Psi\rangle$ by
\begin{align}\label{eq:localizable-multipartite-entanglement}
    \mathcal{L}^E(|\Psi\rangle) \coloneqq \max\left\lbrace \sum\nolimits_{i} p_i E(|v_i\rangle): \{(p_i, |v_i\rangle)\}_i~\text{results from a local measurement}\right\rbrace.
\end{align} 
Here, we restrict the optimization in \eqref{eq:eoa-def} to rank-one projective that are \textit{local} in the sense that their measurement operators necessarily assume a tensor product form with respect to the subsystems in $A$.
As with the MEA, this definition is a direct generalization of the bipartite localizable entanglement~\cite{Verstraete-LE, Popp-LE}.
Naturally, the LME and MEA are highly dependent on the choice of entanglement measure $E$ and the set of measured-out systems $A$. 
To avoid over-encumbering our notation, we will make the choice of subsystems $A$ clear from the context in what follows.

\subsection{Choice of entanglement measures}
\begin{table}
\centering
\begin{tblr}{|c|c|c|c|}
    \hline
    & {$n$-tangle} & {GME-concurrence} & {concentratable entanglement} \\
    \hline
    {local (LME)} & {$\mathcal{L}^\tau$} & {$\mathcal{L}^C$} & {$\mathcal{L}^{\mathrm{CE}}$} \\
    \hline
    {global (MEA)} & {$\mathcal{A}^\tau$} & {$\mathcal{A}^C$} & {$\mathcal{A}^{\mathrm{CE}}$} \\
    \hline
\end{tblr}
\caption{Given a partition $A|B$ of $N$ subsystems, our entanglement localization quantities have two essential ingredients: the class of allowed measurements on $A$ and the entanglement measure used to quantify the entanglement of states over $B$ in the optimal post-measurement ensemble.}
\label{table:entanglement-table}
\end{table}

The question of how much entanglement, as quantified by the concurrence, can be localized on average onto two-qubit subsystems of a multi-qubit system via local measurements has previously been considered~\cite{Verstraete-LE, Popp-LE}.  However, when it comes to studying the entanglement localizable onto multiple qubits, the problem requires the choice of an appropriate entanglement measure. Of course, we may consider a bipartition of the $B$ subsystems and use an entanglement measure that describes the entanglement shared between the two parties in this bipartition, e.g., the entropy of entanglement or the generalized bipartite concurrence~\cite{Rungta2001, Mintert2005}. However, in this case, we do not account for entanglement that is shared between all subsystems in $B$, rather than across a specific bipartition. Several entanglement measures that capture a truly multipartite correlation have been proposed, although the appropriate choice of a specific measure for computing the entanglement localizable onto the $B$ systems can be nebulous because these measures are often highly coarse, difficult to compute, or tailored to specific state properties~\cite{Eisert2001, Hein2004, Osterloh2006}.

In our work, we set out to characterize localizable multipartite entanglement and entanglement of assistance in terms of the $n$-tangle~\cite{Wong2001, Jaeger2003Invariance, Jaeger2003Mixedness, Teodorescu2003}, the genuine multipartite entanglement (GME) concurrence~\cite{Ma2011, Dai2020}, and the concentratable entanglement~\cite{Beckey2021}, which are all established entanglement monotones. Our aim is to derive easily computable bounds for them whose utility we will demonstrate with a variety of examples.

\textit{$n$-tangle.---}The $n$-tangle~\cite{Wong2001, Jaeger2003Invariance, Jaeger2003Mixedness, Teodorescu2003} is defined using the so-called ``Wootter's tilde'' of an $n$-qubit state $|\psi\rangle$ to be the state $|\tilde{\psi}\rangle \coloneqq \sigma_y^{\otimes n} |\psi^\ast\rangle$, where $|\psi^\ast \rangle$ is the complex conjugate of $|\psi\rangle$ with respect to the computational basis. Given an arbitrary $n$-qubit density matrix $\rho$, we give its Wooter's tilde by $\Tilde{\rho} \coloneqq \sigma_y^{\otimes n}\rho^\ast \sigma_y^{\otimes n}$, where $\rho^\ast$ is the complex conjugate of $\rho$ with respect to the computational basis. 
The $n$-tangle $\tau_n$ of an $n$-qubit state $|\psi\rangle$ is defined as
\begin{equation}\label{eq:n-tangle-def}
    \tau_n(|\psi\rangle) \coloneqq |\langle \psi|\Tilde{\psi}\rangle|.
\end{equation}
The $n$-tangle is generically zero when $n$ is odd, so we mainly use the $n$-tangle when the number $n$ of qubits is even. 
The localizable multipartite entanglement defined via \eqref{eq:localizable-multipartite-entanglement} in terms of the $n$-tangle will be denoted as $\Lt$.
Note that even though our definition of the $n$-tangle in \eqref{eq:n-tangle-def} is the square root of the definition commonly used in the literature (e.g., \cite{Wong2001,Jaeger2003Mixedness}), it still defines an entanglement monotone because the square root function is concave.\footnote{
To see this, note that a function $E$ is called an \emph{entanglement monotone} \cite{vidal2000monotones} if it is non-increasing on average under local operations and classical communication (LOCC): If a quantum state $\rho$ is mapped to the state $\rho_i$ with probability $p_i$ by some LOCC map, then $E(\rho) \geq \sum\nolimits_{i} p_i E(\rho_i)$.
By monotonicity and concavity of the square root we then also have $\sqrt{E}(\rho) \geq \sqrt{\sum_i p_i E(\rho_i)} \geq \sum_i p_i \sqrt{E}(\rho_i),$ and hence $\sqrt{E}$ is a valid entanglement monotone as well.
\label{fn:square-root-monotone}}
Our definition of the $n$-tangle is motivated by the fact that the expression~\eqref{eq:n-tangle-def} coincides with the concurrence on two-qubit pure states, allowing us to view~\eqref{eq:n-tangle-def} as a generalization of the concurrence. We find that the convenient analytical form~\eqref{eq:n-tangle-def} of the $n$-tangle, which requires no optimizations and may be interpeted as the fidelity of $|\psi\rangle$ and $|\tilde{\psi}\rangle$, makes it preferable for an analytical treatment of entanglement localization.

Additionally, the $n$-tangle is a particularly interesting entanglement measure from the perspective of the characterization of states equivalent under stochastic local operations and classical communication (SLOCC). We say that two states are SLOCC-equivalent if they may be interconverted with nonzero probability using LOCC. Two states are known to be SLOCC-equivalent if and only if, up to scalar multiples, one is obtained from the other via local $\mathrm{SL}(\mathbb{C}^2)$ (determinant-one) operators~\cite{Dur2000}. The $n$-tangle is an invariant under $\mathrm{SL}(\mathbb{C}^2)^{\otimes n}$~\cite{Jaeger2003Invariance}, and hence yields insight into the SLOCC equivalence classes of $(\mathbb{C}^2)^{\otimes n}$. As an example, if $\mathcal{L}^\tau(|\Psi\rangle) \approx 1$ for some $|\Psi\rangle \in \mathcal{H}_A \otimes \mathcal{H}_B$, with the $B$ system consisting of an even number of qubits, then we may only extract states in the SLOCC class of the W-state with low probability because $\tau(|\rm{W}\rangle) = 0$.

\textit{GME-concurrence.---}The GME-concurrence~\cite{Ma2011, Dai2020} of an $n$-partite state $|\psi\rangle$ is the quantity 
\begin{equation}\label{eq:gme-conc-def}
    C_{\gme}(|\psi\rangle) \coloneqq \min_{\varnothing \subsetneq \gamma \subsetneq \{1, \dots, n\}} \sqrt{2(1 - \tr(\psi_\gamma^2))},
\end{equation}
where $\psi_\gamma$ denotes the reduced state of $|\psi\rangle\langle \psi|$ over the systems in $\gamma$. Since the purity $\tr(\psi_\gamma^2)$ indicates the mixedness of the reduced state $\psi_\gamma$ and, hence, the correlations in $|\psi\rangle$ across the bipartition $\gamma|\tilde{\gamma}$, the RHS of~\eqref{eq:gme-conc-def} may be viewed intuitively as a reflection of the smallest degree of correlation across all bipartitions.
The corresponding localizable multipartite entanglement defined via \eqref{eq:localizable-multipartite-entanglement} in terms of $C_{\gme}$ will be denoted as $\Lgme$.

Both the $n$-tangle and GME-concurrence are desirable quantifiers of entanglement from the viewpoint of transforming entanglement into GHZ-type entanglement. In this scenario, we wish to extract from $|\Psi\rangle$ a state over the $B$ subsystems that approximates a Greenberger-Horne-Zeilinger (GHZ) state~\cite{Greenberger1989}, defined for an $n$-qubit system as $|\ghz_n\rangle = (|0\rangle^{\otimes n} + |1\rangle^{\otimes n})/\sqrt{2}$. For multi-qubit systems, both the $n$-tangle and the GME-concurrence attain their maximal value of 1 on GHZ states and, by their monontonicity, on any state that is locally unitarily equivalent to a GHZ state. In this way, $\mathcal{L}^\tau$ ($\mathcal{A}^\tau$) and $\mathcal{L}^C$ ($\mathcal{A}^C$) must be maximal for any state from which it is possible to extract a GHZ state via local (global) measurements on the $A$ subsystem and local unitary corrections on the $B$ subsystem. Therefore, designing GHZ-extraction protocols by singling out resource states $|\Psi\rangle$ for which either $\mathcal{L}^\tau(|\Psi\rangle)$ or $\mathcal{L}^C(|\Psi\rangle)$ is high is an appealing approach, as finding measurements to achieve an exact state transformation can be a daunting challenge~\cite{Dahlberg2020}. Moreover, many other quantifiers of ``GHZ-likeness'' of the post-measurement ensemble, such as average fidelity with the GHZ state, discriminate against states that are locally unitarily equivalent to the GHZ state even though such states could be converted into the GHZ state via inexpensive single-qudit unitary rotations.

Furthermore, the GME-concurrence has the special property that it is zero for any biseparable state, i.e., any state that can be written as a tensor product of two states over some subsystems. This is useful for localizing entanglement, since maximal  $\mathcal{L}^C(|\Psi\rangle)$ precludes the possibility of extracting biseparable states, which do not have a GHZ entanglement structure but can still maximize the $n$-tangle.
Thus, choosing $|\Psi\rangle$ to maximize $\Lt(|\Psi\rangle)$ defined via the $n$-tangle may well yield such non-GHZ states over the $B$ subsystems, while it is unlikely that we extract such states when $\mathcal{L}^C$ is maximized.

\textit{Concentratable entanglement.---}Finally, the concentratable entanglement $\CE{|\psi\rangle}{s}$~\cite{Beckey2021} of an $n$-partite state $|\psi\rangle$ with respect to a nonempty subset $s \subset [n]$ of the $n$ systems is defined to be 
\begin{equation}\label{eq:CE}
    \CE{|\psi\rangle}{s} \coloneqq 1 - \frac{1}{2^{|s|}}\sum_{\gamma\subseteq s} \tr(\psi_\gamma^2).
\end{equation}
In other words, the concentratable entanglement (CE) depends on some choice $s$ of the subsystems of $|\psi\rangle$. Note that for a given $s$ it is equal to the average of the entanglement, as measured by the linear subsystem entropy, between every bipartition of $M$ with one part contained in $s$. 
The corresponding LME and MEA defined via \eqref{eq:localizable-multipartite-entanglement} and \eqref{eq:eoa-def} will be denoted as $\Lce$ and $\mathcal{A}^{CE}$, respectively (see Table~\ref{table:entanglement-table}). We require that when the concentratable entanglement is used to define the localizable entanglement, the set $s$ is chosen so that $s \subseteq B$ in all measurement branches. If $s$ were chosen so that $s\not\subseteq B$, then the localizable entanglement would not be well-defined.
In our discussion, we will sometimes also use the square root of the concentratable entanglement, which is still a valid entanglement measure due to the concavity of the square root function (see Footnote~\ref{fn:square-root-monotone}). 

It is shown in~\cite{Beckey2021} that several well-studied entanglement measures may be recovered from the concentratable entanglement alone. 
Thus, it stands to reason that by appropriately choosing the subset $s$ in \eqref{eq:CE} to be some proper subset of $B$, we may get a more detailed picture of how entanglement is shared across the $N_B$ systems than the $n$-tangle alone can offer.

\subsection{Relation to prior work}

Sadhukhan et al.~\cite{sadhukhan2017multipartite} previously discussed a multipartite version of localizable entanglement based on the geometric measure of entanglement.
They proved analytical results for generalized GHZ states, generalized W-states and Dicke states, carried out a numerical study for Haar-random states on a few qubits, and analyzed critical phenomena in quantum spin models.
Banerjee et al.~\cite{banerjee2021localizing} introduced a version of LME in the continuous-variables setting, and Harikrishnan and Pal \cite{harikrishnan2023localizing} studied the LME based on the genuine multiparty concurrence for stabilizer states under various noise models.
Other studies have also considered the special case of localizing entanglement on bipartite systems, using (mixed-state) entanglement measures such as logarithmic negativity to define LME and MEA quantities \cite{amaro2018estimating,amaro2020scalable,banerjee2020uniform,banerjee2022hierarchies,krishnan2023controlling}. 

In contrast, our approach allows us to provide extensive analytical results for the localization of pure multipartite entanglement on \textit{arbitrary subsets}, and starting from \textit{arbitrary multipartite states}. As a result, our study goes beyond prior work that deals with localizing entanglement onto multipartite subsystems in the finite-dimensional setting mainly by either taking a numerical approach for arbitrary states or providing analytical results for restricted families of states.
Our choices of entanglement measures mentioned above yield computable upper and lower bounds on both LME and MEA, which we analyze for Haar-random states, state transformations of (weighted) graph states, and ground states of certain spin models.

\section{Easily computable upper and lower bounds}\label{sec:bounds}

Due to the variational approach to defining the localizable entanglement, it is useful to derive upper and lower bounds in terms of the state itself. In this way, one can estimate the localizable entanglement of a state more easily given its experimentally accessible reduced density matrices. To this end, we prove computable upper and lower bounds on the localizable entanglement.

\subsection{Upper bounds}

We first derive an upper bound on the $n$-tangle assuming $|\Psi\rangle\in \mathcal{H}_1 \otimes \dots \otimes \mathcal{H}_N$ is a multi-qubit pure state. From here on, we will use the abbreviation $\varphi\coloneqq |\varphi\rangle\langle\varphi|$ for the density matrix corresponding to a pure state $|\varphi\rangle$. For any pure state $|\varphi\rangle$ over some subset $s$ of the subsystems in $A\cup B$ and for any subset $\gamma \subseteq s$, we denote by $\varphi_\gamma \coloneqq \tr_{\overline{\gamma}}(\varphi)$ the reduced density operator of $\varphi$ on the subsystems in $\gamma$.  
\begin{theorem}\label{Th:n-tangle-upper-bound}
    Suppose that $|\Psi\rangle$ is a multi-qubit pure state and $N_B$ is even. When determined by the $N_B$-tangle, the localizable entanglement and entanglement of assistance of $|\Psi\rangle$ obey 
    \begin{equation}
        \Lt(|\Psi\rangle) \leq \mathcal{A}^\tau(|\Psi\rangle) = F(\Psi_{B}, \Tilde{\Psi}_{B}),
    \end{equation}
    where $F(\rho, \sigma) \coloneqq \textnormal{Tr}(\sqrt{\sqrt{\rho}\sigma\sqrt{\rho}})$ is the fidelity of quantum states $\rho$ and $\sigma$ and $\tilde{\Psi}_B \coloneqq \sigma_y^{\otimes N_B} \Psi_B^\star \sigma_y^{\otimes N_B}$ with $\Psi_B^\star$ denoting the complex conjugate of $\Psi_B$ with respect to the computational basis.
\end{theorem}
For a complete proof, see Appendix~\ref{app:thm-1}. The basic approach of the proof is to apply the proof idea of~\cite{Divincenzo1998} to the multiqubit setting. More precisely, we use the fact that $\tau(|\psi\rangle) = F(\psi, \tilde{\psi})$ for any state $|\psi\rangle \in \mathcal{H}_B$ and the joint concavity of the fidelity functional to show that the average $n$-tangle of the optimal post-measurement ensemble for $|\Psi\rangle_{AB}$ is upper bounded by $F(\Psi_B, \tilde{\Psi}_B)$. To demonstrate that equality holds when global measurements are allowed, we exhibit an explicit measurement basis that yields an average $n$-tangle equal to $F(\Psi_B, \tilde{\Psi}_B)$. 

On a similar note, we are also able to derive an upper bound on the localizable entanglement for the GME-concurrence and concentratable entanglement in terms of the reduced density matrices of $|\Psi\rangle$. 
\begin{theorem}\label{Th:cgme-upper-bound}
    When determined by the GME-concurrence, the localizable entanglement and entanglement of assistance of $|\Psi\rangle$ obey
    \begin{equation}\label{eq:cgme-ub}
        \Lgme(|\Psi\rangle) \leq \mathcal{A}^C(|\Psi\rangle) 
        \leq \min_{\varnothing  \subsetneq \gamma \subsetneq B} \sqrt{2(1 - \textnormal{Tr}(\Psi_\gamma^2))}.
    \end{equation}
\end{theorem}
\begin{theorem}\label{Th:CE-upper-bound}
    When determined by the concentratable entanglement $\CE{|\varphi\rangle}{s}$ for some $\varnothing \subsetneq s \subseteq B$, the localizable entanglement and entanglement of assistance of $|\Psi\rangle$ obey 
    \begin{equation}\label{eq:LCE-bound}
        \Lce(|\Psi\rangle) \leq \mathcal{A}^{CE}(|\Psi\rangle) \leq \CE{|\Psi\rangle}{s}.
    \end{equation}
\end{theorem}
For full proofs of Theorems~\ref{Th:cgme-upper-bound} and~\ref{Th:CE-upper-bound}, see Appendix~\ref{app:ub-proofs}. We prove these theorems essentially by noting that average GME-concurrence and concentratable entanglement, respectively, of an optimal post-measurement ensemble may be viewed as an average of the values of maps derived from the functional $\rho \mapsto 1 - \tr(\rho^2)$, and then using the concavity of this functional. In doing so, we bring convex combinations of states in the post-measurement ensemble, which are just the reduced density matrices of $|\Psi\rangle$, into this functional. Since the reduced density matrices are independent of the optimization over  ensembles, so are the RHS of~\eqref{eq:cgme-ub} and~\eqref{eq:LCE-bound}.

It might be tempting to use the monotonicity of entanglement measures to get an upper bound of the form $L^E(|\Psi\rangle) \leq E(|\Psi\rangle)$, where $E$ is either the GME concurrence, $N_B$-tangle, or the concentratable entanglement and $L^E(|\Psi\rangle)$ is determined by the entanglement measure $E$. However, the $N_B$-tangle and GME-concurrence have only been shown to be entanglement monotones in the following sense: if an LOCC protocol produces from a state $\rho$ an ensemble of states $\{(p_i, \sigma_i)\}$ in a Hilbert space of the \textit{same dimension} as that of $\rho$, then $\sum_i p_i E(\sigma_i) \leq E(\rho)$, where $E$ denotes the convex roof extension~\cite{Horodecki2009} of either $C_{\gme}$ or $\tau_M$. A projective measurement of the sort we consider takes $|\Psi\rangle$ in $\mathcal{H}_A\otimes \mathcal{H}_B$ to a state $|\varphi_i\rangle$ in $\mathcal{H}_B$, which of course has a smaller dimension. As a result, we cannot use monotonicity to bound $L(|\Psi\rangle)$. Towards the contrary,~\cite{sadhukhan2017multipartite} exhibits several examples of states whose entanglement, when measured by the so-called $k$-separability based geometric measure, is exceeded by their localizable entanglement with respect to some subsystem. 

We now turn our attention to lower-bounding the localizable entanglement. We note that while known lower bounds on the GME-concurrence, namely Theorem 2 of~\cite{Dai2020} and the lower bound from Eq. 17 of~\cite{Ma2011}, may be converted into lower bounds on $\Lgme$ by using convexity arguments, we have yet to numerically find an example of a state for which these bounds give a positive value. Hence, we will not pursue this direction here and leave open the possibility of exploring these bounds further via numerics. 

In contrast, there is a useful lower bound on the localizable entanglement defined in terms of the concentratable entanglement in terms of two-point spin correlation functions.
This approach is inspired by \cite{Lausten2003}.
To define this lower bound, let $\vec{a}, \vec{b}\in \mathbb{R}^3$ be arbitrary vectors of unit norm, and define for any two $i\neq j\in B$ the two-point spin correlation function 
\begin{equation}
		Q_{\vec{a}, \vec{b}}^{ij}(|\Psi\rangle) = \langle \Psi|
		(\vec{a}\cdot\vec{\sigma}^{(i)})\otimes (\vec{b}\cdot \vec{\sigma}^{(j)})|\Psi\rangle - \langle \Psi|(\vec{a}\cdot\vec{\sigma}^{(i)})|\Psi\rangle\langle\Psi|(\vec{b}\cdot\vec{\sigma}^{(j)})|\Psi\rangle, \label{eq:Q-function}
\end{equation}
where $\vec{\sigma}^{(i)} = \sum_{k=1}^3 \sigma_k^{(i)} \hat{e_i}$, the set $\{\hat{e_i}\}_{i=1}^3$ is the standard basis for $\mathbb{R}^3$, and $\sigma_k^{(i)}$ is the $k$-th Pauli matrix acting on system $i$. We then have the following bound:

\begin{theorem}\label{Th:corr-func-lower-bound}
    The localizable entanglement defined by the square root concentratable entanglement $\sqrt{\CE{|\varphi\rangle}{s}}$ for $\varnothing \subsetneq s \subseteq B$ satisfies     
    \begin{align}\label{eq:corr-func-lb}
       \sqrt{\Lce(|\Psi\rangle)} \geq L^{\sqrt{\mathrm{CE}}}(|\Psi\rangle) \geq 
       \frac{1}{2} \max_{\substack{i\neq j \in s\\\hat{a}, \hat{b} \in S^2}}  \left|Q_{\hat{a}, \hat{b}}^{ij}(|\Psi\rangle)\right|,
    \end{align}
	where $Q_{\hat{a}, \hat{b}}^{ij}$ is the function defined in \eqref{eq:Q-function} and $S^2$ is the unit sphere in $\mathbb{R}^3$. Furthermore, the right-hand side of~\eqref{eq:corr-func-lb} is equal to the maximum singular value of the $3\times 3$ matrix $\mathbf{Q}^{ij}(|\Psi\rangle)$ defined by $(\mathbf{Q}^{ij}(|\Psi\rangle))_{pq} = Q^{ij}_{\hat{e}_p, \hat{e}_q}(|\Psi\rangle)$ for the optimal $i\neq j \in s$.
\end{theorem}
Theorem~\ref{Th:corr-func-lower-bound} shows that the presence of classical correlations serves as an indicator for the ability to create quantum correlations. To prove it, we use the known fact~\cite{Verstraete-LE} that for a target system consisting of two qubits, the localizable entanglement defined in terms of the two-qubit concurrence is lower bounded by absolute values of all spin correlation functions across the target qubit pair. By writing the two-qubit concurrence in terms of the CE of a qubit pair and employing concavity and monotonicity arguments, we show that the localizable entanglement defined by the concurrence lower bounds $\Lsqrtce$ for any qubit pair. Hence, the magnitudes of the spin correlation functions between any pair of qubits in $B$ lower bound $\Lsqrtce$ and $\sqrt{\Lce}$ as in~\eqref{eq:corr-func-lb}. For the full proof, see Appendix~\ref{app:LCE-corr-func-LB}. Note that the square root CE is indeed an entanglement measure by the concavity of the square root function. We introduce the auxiliary quantity $L^{\sqrt{\mathrm{CE}}}$ to obtain a tighter lower bound on the localizable entanglement. 

In summary, we have presented a series of general bounds on the entanglement one may localize onto part of a quantum state via projective measurements on its subsystems. These bounds hold for arbitrary multipartite pure states, and their main significance is that they are computable purely in terms of the experimentally accessible reduced state of $|\Psi\rangle$ over the $B$ systems. We thereby avoid complicated bounds in terms of optimization problems over the state space.
Finding easily computable and non-trivial lower bounds for the LME based on the $n$-tangle and the GME-concurrence remain interesting open questions.

\section{Continuity bounds}
\label{sec:continuity-bounds}

We now turn to analyzing the continuity of the localizable entanglement and entanglement of assistance. It is useful to produce uniform continuity bounds on these functionals for both physical and mathematical reasons. Firstly, continuity places functions of quantum states on firmer physical footing. One would not expect a physical quantity to change drastically when the underlying system is perturbed by a small amount of noise. Secondly, if the value of the localizable entanglement is known or may be estimated at a particular state, then a uniform continuity bound automatically gives us upper and lower bounds for our figures of merit that are nontrivial locally in a neighborhood of the given state, thereby fulfilling one of our goals of bounding the localizable entanglement and entanglement of assistance. Finally, the \textit{Lipschitz} continuity of these functionals will form a key tool in our analysis of the concentration phenomena that they display for large Hilbert spaces (Section~\ref{sec:concentration}). We add that while it is certainly true that the average post-measurement entanglement for a fixed measurement is continuous by virtue of it being a composition of continuous functions, this does not automatically imply a uniform continuity bound on the localizable entanglement and entanglement of assistance. 

Let $E\colon \mathcal{H}_B \to [0,1]$ be an entanglement measure, and for a linear operator $L$ denote by $\|L\|_1 \coloneqq \tr(\sqrt{L^\dagger L})$ its trace norm. Suppose $E$ has the property that
there exists a concave, monotonically increasing function $f\colon \mathbb{R}_{\geq 0} \to \mathbb{R}_{\geq 0 }$ such that 
\begin{align}
    |E(|\psi\rangle) - E(|\psi'\rangle)| &\leq f(\|\psi - \psi' \|_1) \label{eq:modified-lipschitz}
\intertext{and}
    E(|\psi\rangle) &\leq f(\|\psi\|_1) \label{ineq:entanglement-bound}
\end{align}
for all states $|\psi\rangle, |\psi'\rangle \in \mathcal{H}_B$. Our next theorem then allows us to bound the variations in $L^E$ and, if $f(x) \to 0$ as $x \to 0^+$, implies that $L^E$ is uniformly continuous.
\begin{theorem}\label{Th:LE-continuity-bound}
     For any $|\Psi\rangle, |\Psi'\rangle \in \mathcal{H}_A \otimes \mathcal{H}_B$, we have
    \begin{align}\label{ineq:main-result}
        |\mathcal{L}^E(|\Psi\rangle) - \mathcal{L}^E(|\Psi'\rangle)| &\leq f(2\|\Psi - \Psi'\|_1) + \|\Psi - \Psi'\|_1 \\ |\mathcal{A}^E(|\Psi\rangle) - \mathcal{A}^E(|\Psi'\rangle)| &\leq f(2\|\Psi - \Psi'\|_1) + \|\Psi - \Psi'\|_1
    \end{align}
\end{theorem} 
While the full proofs of Theorem~\ref{Th:LE-continuity-bound} and other results in this section are given in Appendix~\ref{app:continuity-results}, we provide some intuition for these proofs here. We prove Theorem~\ref{Th:LE-continuity-bound} by deriving analogous continuity bounds on the average post-measurement entanglement for a fixed measurement basis, as opposed to $L^E$ or $\mathcal{A}^E$ directly. The two terms on the RHS of the continuity bounds of Theorem~\ref{Th:LE-continuity-bound} derive from the triangle inequality, with the first term describing variations in the entanglement values of states in the post-measurement ensembles of $|\Psi\rangle, |\Psi'\rangle$ and the latter deriving from the variations in their respective measurement probabilities. For the first term, we use~\eqref{eq:modified-lipschitz} to upper bound the deviations in the entanglement value $E$ of states in the post-measurement ensemble by deviations in their trace distance, which we relate to $\|\Psi - \Psi'\|$ via the convexity of $f$. Simultaneously, we use~\eqref{ineq:entanglement-bound} to bound the latter term by the total variation distance of the distributions of measurement outcomes for $|\Psi\rangle, |\Psi'\rangle$, which we bound by $\|\Psi - \Psi'\|_1$ via an argument invoking the monotonicity of the trace distance under CPTP maps.

The following lemma guarantees that the $N_B$-tangle, the GME-concurrence, and the concentratable entanglement all satisfy condition~\eqref{eq:modified-lipschitz}. 

\begin{lemma}\label{lemma:3-continuity-bounds}
    We have 
    \begin{align}
        |\tau_{N_B}(|\psi\rangle) - \tau_{N_B}(|\psi'\rangle)| &\leq\label{ineq:n-tangle-continuity-bound} \sqrt{2}\| \psi - \psi' \|_1 \\
        |C_{\gme}(|\psi\rangle) - C_{\gme}(|\psi'\rangle)| &\leq\label{ineq:gme-concurrence-continuity-bound} 2^{3/4} \sqrt{\|\psi - \psi'\|_1} \\
        |C(|\psi\rangle; s) - C(|\psi'\rangle; s)| &\leq\label{ineq:ce-continuity-bound} \sqrt{2}\|\psi - \psi'\|_1
    \end{align}
    for all quantum states $|\psi\rangle, |\psi'\rangle \in \mathcal{H}_B$.
\end{lemma}
The continuity bound for the concentratable entanglement in Lemma~\ref{lemma:3-continuity-bounds} is taken from \cite{schatzki2021entangled}. The continuity bound for the $n$-tangle may be derived via a straightfoward argument involving the Cacuhy-Schwartz inequality. Finally, our derivation of the bound for the GME-concurrence resembles the derivation of the continuity bound for the CE from \cite{schatzki2021entangled} in that it relates the deviations in subsystem purities to the Hilbert-Schmidt distance of reduced states, at which point norm inequalities and the monotonicity of the trace distance under CPTP maps may be used to complete the proof.

For the choices of $f$ implied by Lemma~\ref{lemma:3-continuity-bounds}, we see that condition~\eqref{ineq:entanglement-bound} is also satisfied because our entanglement measures of choice all assume values in $[0,1]$, while $f(\|\psi\|_1) = f(1) > 1$ for all $|\psi\rangle \in \mathcal{H}_B$. Consequently, we receive from Theorem~\ref{Th:LE-continuity-bound} continuity bounds on the localizable entanglement defined by each of the three entanglement measures.
\begin{corollary}\label{cor:LE-local-bounds}
    Let $\mathcal{L}^\tau$, $\Lgme$, and $\Lce$ be the localizable entanglement when defined by the $n$-tangle, GME-concurrence, and the concentratable entanglement, respectively. Then for any quantum states $|\Psi\rangle, |\Psi'\rangle\in\mathcal{H}_A \otimes \mathcal{H}_B$, we have 
    \begin{align}
        |L^{\tau}(|\Psi\rangle) - L^{\tau}(|\Psi'\rangle)| &\leq (2\sqrt{2} + 1)\|\Psi - \Psi'\|_1 \\
        |\Lgme(|\Psi\rangle) - \Lgme(|\Psi'\rangle)| &\leq 2^{5/4}\sqrt{\|\Psi - \Psi'\|_1} + \|\Psi - \Psi'\|_1 \\
        |\Lce(|\Psi\rangle) - \Lce(|\Psi'\rangle)| &\leq (2\sqrt{2} + 1)\|\Psi - \Psi'\|_1
    \end{align}
    The same bounds hold for the entanglements of assistance $\mathcal{A}^\tau, \mathcal{A}^C,$ and $\mathcal{A}^{CE}$ as well.
\end{corollary}
We note that the numerical constants appearing on the RHS of the inequalities in Lemma~\ref{lemma:3-continuity-bounds} and Corollary~\ref{cor:LE-local-bounds} are merely artifacts of our proof technique and have no particular conceptual or technical significance here.

Corollary~\ref{cor:LE-local-bounds} shows that the localizable entanglement and entanglement of assistance defined by the $N_B$-tangle and concentratable entanglement are \emph{Lipschitz}-continuous, a stronger form of continuity. As a result, coherent state preparation errors that perturb a state by a small amount $\varepsilon$ in trace distance will only shift these figures of merit by $\mathcal{O}(\varepsilon)$. Also, in contrast to continuity bounds on various quantum entropies, which often depend on the underlying dimension, the Lipschitz constant here is independent of dimension. This dimension independence is particularly important in the context of scalability.

\section{Bounds via measure concentration}\label{sec:concentration}

We now analyze the multipartite localizable entanglement and entanglement of assistance for Haar-random states. Due to the phenomenon of measure concentration, the values of these quantities tend to concentrate around their mean values for large Hilbert spaces. This allows us to make a few interesting remarks about the behavior of multipartite entanglement in large systems. Furthermore, measure concentration implies that we can bound the deviations of the multipartite localizable entanglement and entanglement of assistance from their mean value, at least probabilistically. In this way, we may fulfill our stated goal of deriving bounds on these figures of merit.  

In our discussion, we will assume that $|\Psi\rangle$ is distributed according to the Haar measure on $\mathcal{H}_A \otimes \mathcal{H}_B$. The following theorem, which we prove in detail in Appendix~\ref{app:concentration-results}, tells us that the MEA defined by the $N_B$-tangle will be near maximal for almost all quantum states if the size of the measured subsystem is much larger than the size of the surviving subsystem.  

\begin{theorem}\label{Th:EoA-concentration}
    Assume $N_B$ is even. For any $\varepsilon > 0$, the probability that the multipartite entanglement of assistance $\mathcal{A}^\tau$ is bounded away from one satisfies
    \begin{equation}
    \Pr\bigg(\mathcal{A}^\tau(|\Psi\rangle_{AB}) \leq 1 - \sqrt{\frac{d_B}{2d_A}} - \varepsilon \bigg) \leq 2\exp\left(-\frac{2d_A d_B \varepsilon^2}{9\pi^3 (4\sqrt{2} + 2)^2} \right) \label{eq:MEA-concentration}
    \end{equation}
\end{theorem}
Theorem~\ref{Th:n-tangle-upper-bound} states that $\mathcal{A}^\tau(|\Psi\rangle) = F(\Psi_B, \tilde{\Psi}_B)$, so by deriving a lower bound for the expectation value of $F(\Psi_B, \tilde{\Psi}_B)$ for Haar-random $|\Psi\rangle$, we obtain a lower bound for the expectation value of $\mathcal{A}^\tau(|\Psi\rangle)$. Since $\mathcal{A}^\tau$ is Lipschitz-continuous by Corollary~\ref{cor:LE-local-bounds}, we may invoke Levy's Lemma (see Appendix~\ref{app:concentration-results}) to derive the concentration inequality of Theorem~\ref{Th:EoA-concentration}.

We may glean from this theorem an important observation about the nature of multipartite entanglement. The entanglement of assistance of a state basically describes the maximal average value of the entanglement of states in an ensemble produced by a projective measurement on part of a multipartite system. On the other hand, one may consider the entanglement of the density matrix obtained by averaging all states produced by measuring out a state in part, i.e., the \textit{reduced density matrix}. The difference between these two approaches of quantifying the entanglement between the constituents \textit{within} a given subsystem essentially lies in the order in which the averaging and the computation of the entanglement measure are carried out. 

Naturally, we would expect the former approach to be more sensitive to detecting entanglement within a subsystem because it accounts for the correlations within each \textit{individual} state of the projected ensemble. However, the gap between the two approaches is most compellingly illustrated by the characteristic example of graph states.
The two-qubit marginals of graph states corresponding to connected graphs are separable, meaning that they only show classical correlations. On the other hand, the localizable entanglement between any two qubits of a connected graph state is maximal~\cite{Hein2006}, which implies the same for the entanglement of assistance. For this reason, graph states have served as an important motivator for the study of localizable entanglement (see the introductory discussions of~\cite{Popp-LE, krishnan2023controlling, harikrishnan2023localizing}).

Theorem~\ref{Th:EoA-concentration} shows that the maximal gap between the entanglement-of-average and average-of-entanglement approaches illustrated by graph states holds not only for special examples, but for \textit{almost all states} in a large Hilbert space. To see this, note that the bipartite entanglement entropies of quantum states cluster tightly around their maximal value~\cite[Thm.~III.3]{Hayden2006}. More directly, we have for arbitrary $\varepsilon, \delta > 0$ that 
\begin{equation}
    \Pr_{|\Psi\rangle \sim \Haar(d)}\left(S(\Psi_B) < \ln d_B - \varepsilon \right) \leq \delta,
\end{equation}
if $d_A$ is sufficiently larger than $d_B$. 
Since typical states in $\mathcal{H}_A\otimes \mathcal{H}_B$ have near-maximal entanglement entropy with respect to the partition $A|B$, it follows that their reduced states are nearly maximally mixed over $\mathcal{H}_B$. Since the maximally mixed state over $\mathcal{H}_B$ exhibits no entanglement, we conclude that when the system size of $A$ is large, the entanglement-of-average approach detects only near-minimal entanglement while the average-of-entanglement approach simultaneously detects near-maximal entanglement within $B$ for nearly all states.  

Of course, since the MEA may be defined through various multipartite entanglement measures that behave in distinct ways, Theorem~\ref{Th:EoA-concentration} does not automatically guarantee that the MEA defined by other measures will behave similarly.  However, our next result, which we also prove in Appendix~\ref{app:concentration-results}, guarantees that the MEA also concentrates near its maximal value for the GME-concurrence when $B$ is a multi-qubit system and for the concentratable entanglement when $s=B$, provided that $d_B$ is sufficiently large. To formulate the statement, let $\overline{C_{\rm GME, \beta}
}(|\Psi\rangle)$ and $\overline{C}_\beta(|\Psi\rangle;s)$ respectively denote the value of the GME-concurrence and CE averaged over the ensemble produced by a projective measurement $\beta$ performed on $|\Psi\rangle$ over the $A$ subsystem. 

\begin{theorem}\label{Th:EoA-concentration-2}    
    Fix an orthonormal basis $\beta$ of $\mathcal{H}_A$. 
    \begin{enumerate}
    \item For any subset $\gamma \subset B$, let $d_\gamma$ denote the Hilbert space dimension of the subsystem associated with $\gamma$. For all $\varepsilon > 0$, the probability that the average post-measurement GME-concurrence is bounded away from one satisfies
    \begin{multline}
        \Pr_{|\Psi\rangle \sim \Haar(d)}\left(\overline{C_{\rm{GME}, \beta}}(|\Psi\rangle) \leq 2 - 2f(B) - 4d_B\exp\left(-\frac{\sqrt{d_B}}{72\pi^3}\right) - 2d_B^{-1/4} - \varepsilon \right)\\
        \leq 2\exp\left( - \frac{2d_Ad_B\varepsilon^2}{9\pi^3(2^{7/2}+2)^2} \right)
        \label{eq:Cgme-concentration}
    \end{multline}
    where $f(B)\coloneqq \max_{\varnothing\subsetneq \gamma \subsetneq B} (d_\gamma + d_{B\setminus\gamma})/(d_B + 1) = 1/2$ when $B$ is a multi-qubit system. 
    \item For all $\varepsilon>0$, the probability that the average post-measurement CE is bounded away from one satisfies    
    \begin{equation}
        \Pr_{|\Psi\rangle \sim \Haar(d)}\bigg(\overline{C}_\beta(|\Psi\rangle; s) \leq 1 - \frac{3^{|s|}(2^{N_B - |s|} + 1)}{2^{|s|}(2^{N_B} + 1)} - \varepsilon \bigg) 
        \leq 2\exp\left(-\frac{2d_A d_B \varepsilon^2}{9\pi^3 (4\sqrt{2} + 2)^2} \right).
        \label{eq:CE-concentration}
    \end{equation}
    \end{enumerate}
\end{theorem}

The proof of Theorem~\ref{Th:EoA-concentration-2} proceeds in a way similar to that of Theorem~\ref{Th:EoA-concentration}.
Note that since $\overline{C_{\rm{GME, \beta}}}(|\Psi\rangle) \leq \mathcal{L}^C(|\Psi\rangle) \leq \mathcal{A}^C(|\Psi\rangle)$ and $\overline{C}_\beta(|\Psi\rangle;s) \leq \Lce(|\Psi\rangle) \leq \mathcal{A}^{\rm CE}(|\Psi\rangle)$ for any orthonormal basis $\beta$, the statement of the theorem automatically holds when $\overline{C_{\rm{GME}}^{(\beta)}}$ is swapped by $\mathcal{L}^C$ or $\mathcal{A}^C$ and when $\overline{C}_\beta$ is swapped by $\mathcal{L}^{CE}$ or $\mathcal{A}^{CE}$. Thus, the maximal gap between the average-of-entanglement and entanglement-of-average approaches also holds for typical states when the GME-concurrence (when $B$ consists of $N_B$ qubits) or CE (when $s=B$) are used to quantify the entanglement of states in the projected ensemble for high-dimensional Hilbert spaces. 

It is worth noting how the concentration of the MEA for the $N_B$-tangle is dependent on the optimal measurement basis. For any orthonormal basis $\beta$ of $\mathcal{H}_A$, let $\overline{\tau}_\beta(|\Psi\rangle)$ denote the value of the $N_B$-tangle averaged over the ensemble produced by a projective measurement of $\beta$ performed on $|\Psi\rangle$ over $A$. For a fixed measurement basis $\beta$, the average entanglement $\overline{\tau}_\beta$ concentrates near zero for large system sizes. In particular, for any $\varepsilon, \delta >0$, we have
\begin{equation}\label{eq:tau-beta-concentration}
    \Pr_{|\Psi\rangle\sim \Haar(d)}(\overline{\tau}_\beta(|\Psi\rangle) \geq \varepsilon) \leq \delta.
\end{equation}
if $d_A$ and $d_B$ are sufficiently large (for details, see end of appendix~\ref{app:concentration-results}). Thus, the fact that the values of $\mathcal{A}^\tau$ cluster near one when $d_A \gg d_B$ relies on an understanding of how the average value of the $N_B$-tangle behaves for the optimal measurement basis. For this purpose, we used Theorem~\ref{Th:n-tangle-upper-bound} in our proof of Theorem~\ref{Th:EoA-concentration}. 

At the same time, the fact that the values of $\overline{C_{\rm GME, \beta}}$ and $\overline{C}_\beta$ concentrate near one for multi-qubit systems and $s = B$, respectively, may be seen as a result of the fact that the values of the GME-concurrence and concentratable entanglement themselves cluster near one for typical states in $\mathcal{H}_B$ under the same conditions~\cite{Schatzki2024}.
In other words, if almost every quantum state in $\mathcal{H}_B$ has near maximal entanglement, then intuitively the projected ensemble produced by almost every state in $\mathcal{H}_A \otimes \mathcal{H}_B$ will have states with near maximal entanglement occuring with high probability.

Furthermore, our continuity bounds on the \textit{localizable entanglement} $\mathcal{L}^\tau$ imply that the values of $\mathcal{L}^\tau$ should concentrate around their mean value. To see this, note that Theorem~\ref{Th:LE-continuity-bound} shows that $\mathcal{L}^\tau$ is Lipschitz-continuous.
This allows us to use Levy's Lemma, which describes how the values of a function concentrate exponentially tightly around its mean value on Haar-random states. 
We are thus guaranteed that the values of $\mathcal{L}^\tau$ will concentrate around its mean for large system sizes. In fact, we used Levy's Lemma in Appendix~\ref{app:concentration-results} to derive the other concentration inequalities of this section in such a way. Unfortunately, computing the mean value of $\mathcal{L}^\tau$ seems to be a more challenging task than computing the mean of $\mathcal{A}^\tau$ and $\mathcal{L}^{CE}$ due to the delicate dependence of $\overline{\tau}_\beta$ on the optimal measurement basis and the optimization over specifically \textit{local} measurement bases. However, in Fig.~\ref{fig:l-tau-decay}, we present numerical evidence that the values of $\mathcal{L}^\tau$ concentrate away from one for large system sizes, reinforcing the significance of the optimal measurement in the study of $\mathcal{L}^\tau$ and $\mathcal{A}^\tau$.

\begin{figure}[h!]
    \centering
    \includegraphics[width = 0.55\columnwidth]{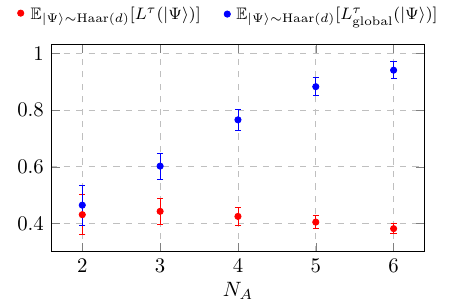}
    \caption{Each red point indicates the localizable entanglement $\mathcal{L}^\tau$ averaged across $10^4$ numerically sampled Haar-random states for $N_A = 2,3$ and across $10^3$ states for $N_A = 4,5,6$, all with $N_B = 4$. The blue points indicate the entanglement of assistance $\mathcal{A}^\tau$ averaged across the same samples. Error bars represent one standard deviation for each sample.}\label{fig:l-tau-decay}
\end{figure}

We add that the importance of the optimal global measurement basis in the asymptotic concentration of the average post-measurement value of the $N_B$-tangle around its maximal value parallels the situation of our graph state example closely. Indeed, the sensitivity to the optimal measurement seems to capture the situation of the example more than our concentration result for the concentratable entanglement does. For a connected graph state, there exist many measurements such that if two qubits remain after the measurement, their entanglement is zero or near minimal. Only for select choices of measurements is the entanglement on the surviving qubit pair maximal. Similarly, the values of $\overline{\tau}_\beta$ concentrate towards zero when $d_A \gg d_B$; it is only the value of the average post-measurement entanglement for the optimal global measurement that we can guarantee concentrates towards one. 

Finally, we note that while previous works~\cite{krishnan2023controlling, sadhukhan2017multipartite, banerjee2020uniform} have studied measure concentration phenomena in various formulations of the localizable entanglement, our work differs from theirs in a few key respects. Our results are primarily analytical, whereas prior results are mainly numerical. Our formulation also uses seed entanglement measures that are fundamentally multipartite. Finally, our work results in an important conceptual observation about the problem of quantifying internal correlations within a subsystem. 

\section{Numerics for Haar-random states}\label{sec:numerics}

To visualize how much entanglement can typically be  localized with local measurements, we plot the values of $\mathcal{L}^\tau$, $\Lgme$, and $\Lce$ for states numerically sampled uniformly at random from the Haar measure over a multi-qubit Hilbert space consisting of $N_A$ qubits in the measured-out system $A$ and $N_B$ qubits in the surviving system $B$. Thus, computing the LME requires optimizing over $2N_{A}$ real parameters. In order to perform this task, we made use of particle swarm optimization (PSO) implemented via the package \textit{pyswarms}. In Fig.~\ref{fig:comprehensive-Haar-numerics}, the $y$-coordinate of each point corresponds to the localizable entanglement of a Haar-random state while the $x$-coordinate corresponds to its value of the associated upper bound from Theorems~\ref{Th:n-tangle-upper-bound},~\ref{Th:cgme-upper-bound}, and~\ref{Th:CE-upper-bound}.

\begin{figure}[h!]
    \includegraphics[scale=0.9]{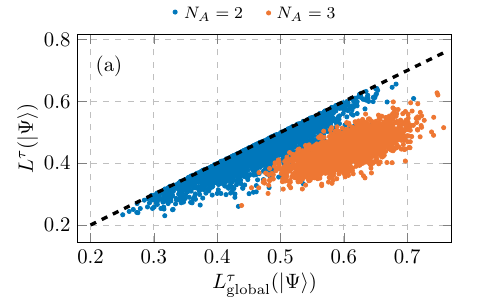}
    \includegraphics[scale=0.9]{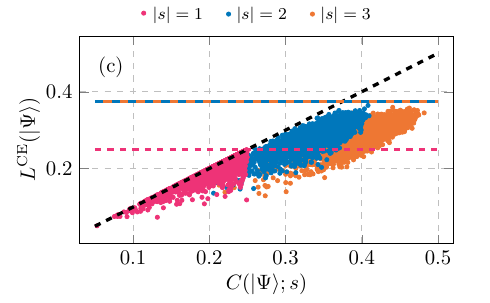}    
    \includegraphics[scale=0.9]{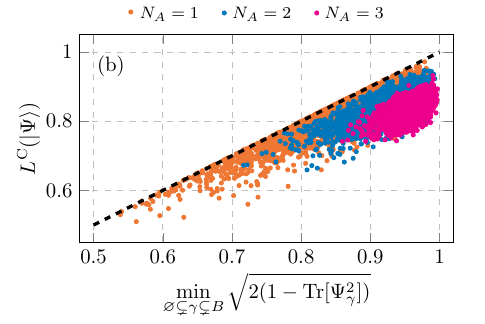}
    \includegraphics[scale=0.9]{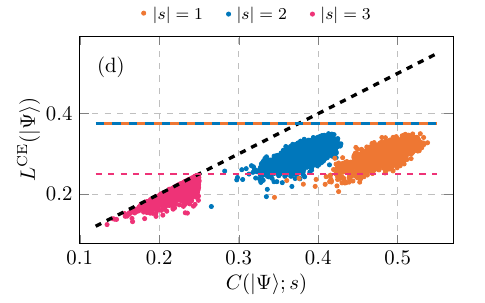}
    \caption{We compare values of the LME to the values of our upper bounds from Section~\ref{sec:bounds} for Haar-random states. Each color corresponds to a sample size of 2500 states. The dashed black lines are plots of the line $y=x$. Tightness of the bounds from Theorems~\ref{Th:n-tangle-upper-bound},~\ref{Th:cgme-upper-bound}, and~\ref{Th:CE-upper-bound} is judged by proximity of the sampled data to this line. (a) Data obtained for the LME defined by the $n$-tangle with $N_B = 4$ held fixed. Values of $\mathcal{A}^\tau(|\Psi\rangle)$ were obtained by using the fact that $\mathcal{A}^\tau(|\Psi\rangle) = F(\Psi_B, \tilde{\Psi}_B)$ from Theorem~\ref{Th:n-tangle-upper-bound}. (b) A visualization of the LME defined by the GME-concurrence with $N_B = 3$ held fixed. (c) Values of the LME defined by the CE for varying sizes of $s$ and $N_A = 1, N_B = 3$ held constant. Colored dashed lines represent simple upper bounds on $\Lce(|\Psi\rangle)$ obtained by assuming that every state in the post-measurement ensemble is \textit{absolutely maximally entangled} - that is, each marginal obtained by tracing out at least half of the subsystems in $B$ is maximally mixed. Note that these bounds coincide for $|s| = 2$ and $|s| = 3$, as shown by the overlapping orange and blue lines. (d) The situation is essentially the same as in (c), with the only difference being that $N_A = 2$ is held fixed. 
 }\label{fig:comprehensive-Haar-numerics}
\end{figure}

We find that when the localizable entanglement is defined by the $N_B$-tangle, our upper bound on $\mathcal{L}^\tau$ is relatively tight and even saturated by many states when $N_A = 2, 3$ (see Fig.~\ref{fig:comprehensive-Haar-numerics}a). However, as the number of measured qubits grows, the separation between $\mathcal{L}^\tau$ and its upper bound becomes pronounced (Fig.~\ref{fig:l-tau-decay}). This implies that as the number of measured qubits grows, global measurements gain a distinct advantage. Measure concentration effects are also displayed by the tighter clustering of points for higher $N_A$ in Fig.~\ref{fig:comprehensive-Haar-numerics}a and the decay in error bar size in Fig.~\ref{fig:l-tau-decay}. 

When defined by the GME-concurrence, the LE also seems to saturate the upper bound often for $N_A = 1$ (see Fig.~\ref{fig:comprehensive-Haar-numerics}b). Unlike $\mathcal{L}^\tau$, the separation between the values of $\Lgme$ and its upper bound seems to decay more gently as $N_A$ grows. Since the upper bound on $\Lgme$ from Theorem~\ref{Th:cgme-upper-bound} also holds for global measurements, this implies that global measurements gain an advantage more gradually in localizing GME-concurrence. Despite our lack of a formal concentration inequality for $\Lgme$ in Section~\ref{sec:concentration}, the tighter clustering of points for higher $N_A$ is suggestive of measure concentration. 

In Fig.~\ref{fig:comprehensive-Haar-numerics}c, we plot the values of $\Lce$ for random states with $N_A = 1, N_B = 3$ for varying sizes of $s\subseteq B$. The upper bound of Theorem~\ref{Th:CE-upper-bound} seems to remain fairly tight as $|s|$ grows. A simple upper bound on the values of the CE is obtained by assuming that each state in the post-measurement ensemble has the property that tracing out at least half of the subsystems in $B$ yields a maximally mixed state. For example, when $|s| = 2$, this yields $1 - \frac{1}{2^2}(1 + \frac{1}{2} + \frac{1}{2} + \frac{1}{2}) = 0.375$ since maximally mixed states over single qubits have a purity of 1/2 and the purity of the marginal over the entirety of $s$ is equal to the purity of the marginal over the single qubit in $B \setminus s$. These upper bounding values are represented as horizontal dashed lines for varying sizes of $|s|$ and are relatively tight for $|s| = 1,2,3$.

However, when we increase the number of measured qubits to $N_A = 2$ (see Fig.~\ref{fig:comprehensive-Haar-numerics}d), we find that the upper bound of Theorem~\ref{Th:CE-upper-bound} becomes markedly less tight. The performance of the upper bound is particularly diminished when $|s|$ is maximized. This can be inferred from the fact that our proof of the upper bound relies on the convexity of the function $\rho \mapsto \tr(\rho^2)$, which provides a  bound whose tightness decreases as the number of terms involved increases. This suggests that to characterize the amount of entanglement that can be localized with local measurements in terms of the CE, it is best to consider smaller sizes of the qubit label $s$.

\section{Applications to graph states}\label{sec:graph-states}

\subsection{Overview}

Graph states are a theoretically rich class of multi-qubit states that serve as a paradigm to study properties of multipartite entanglement and nonlocality. 
Furthermore, many quantum information-processing tasks such as measurement-based quantum computing \cite{raussendorf2001oneway,raussendorf2003measurement} rely on them as a resource.

Given a simple graph $G=(V,E)$ with vertex set $V$ labeling a collection of qubits and edge set $E\subseteq \{\{a,b\}: a\neq b\in V\}$, we define the graph state 
\begin{equation}\label{eq:ord-graph-state}
    |G\rangle \coloneqq \prod_{\{a,b\} \in E} \text{CZ}_{ab} |+\rangle^{\otimes N},
\end{equation}
where $|+\rangle \coloneqq (|0\rangle + |1\rangle)/\sqrt{2}$ and $\text{CZ}_{ab}$ is a controlled-Z gate acting between qubits $a$ and $b$. Note that the controlled-Z operator between two fixed qubits does not depend on the choice of control and target qubit, preventing ambiguity in the above definition. 

A problem of considerable interest is characterizing the possible transformations from one graph state to another.
However, characterizing one-way transformations allowing an arbitrary sequence of local measurements and local unitaries is a rather challenging task. On the other hand, local Pauli Measurements (LPM) and Local Clifford (LC) operations are convenient to analyze because these operations transform the underlying graph structure in a simple way when they are applied to graph states. 
Nevertheless, it was shown in~\cite{Dahlberg2020} that deciding whether a source graph state $|G_1\rangle$ may be transformed into a target graph state $|G_2\rangle$ on some subsystem using even the well-behaved set of local Clifford operations, local Pauli measurements, and classical communication is an NP-complete problem.
A similar result was shown for the problem of transforming an arbitrary graph state into a tensor product of Bell pairs~\cite{Dahlberg2020-2}. Thus, deciding the allowable graph state transformations even under a set of relatively simple operations is a challenging problem.

Our work sheds light on this problem by utilizing the framework from Section~\ref{sec:bounds}. This analysis allows insights into what graph state transformations are possible using projective measurements and general local unitary operations, going beyond the restrictions to LC + LPM from prior studies. Our work also provides a simple criterion for deciding whether a certain subset of graph state transformations are possible. This criterion is phrased in terms of a matrix equation, which is solvable in polynomial time using, for example, Gaussian elimination. In this way, our work circumvents the bottleneck imposed by the NP-completeness results of~\cite{Dahlberg2020, Dahlberg2020-2}. Our study extends the previous work of~\cite{harikrishnan2023localizing}, which highlights various properties of the post-measurement ensemble produced by LC + LPM on graph states with Pauli noise, and, which~\cite{amaro2018estimating} provides estimates for the localizable entanglement defined by the bipartite logarithmic negativity on noisy graph states.

\subsection{Graph state post-measurement ensembles}\label{sec:post-meas-ensembles}
We now establish some notation for the remainder of Section~\ref{sec:graph-states}. Let $G = (V, E)$ be a graph whose vertex set $V$ is partitioned into two nonempty sets $A$ and $B$. When discussing the LME or MEA, we assume as before that measurements are performed on the subsystem $A$, which is then subsequently discarded. We again take $N \coloneqq |V|$, $N_A \coloneqq |A|$, and $N_B \coloneqq |B|$. Let $G - A$ denote the subgraph of $G$ induced by $B$. That is, $G-A$ is the graph having $B$ as a vertex set and whose edge set is the collection of all edges of $G$ connecting only vertices in $B$. Let $\mathbf{\Gamma}_{BA}$ denote the $N_B\times N_A$ block of the adjacency matrix of $G$ describing the connectivity between $A$ and $B$. That is, for vertices $b\in B$ and $a\in A$, the matrix element $(\mathbf{\Gamma})_{ba}$ is one if there is an edge connecting $a$ and $b$ and zero otherwise. We denote by $\mathbb{F}_2^A$ and $\mathbb{F}_2^B$ the fields of binary vectors with $N_A$ and $N_B$ components, respectively.
Let $\sigma_i^{\mathbf{z}}$ for $\mathbf{z} \in \mathbb{F}_2^B$ denote the Pauli string consisting of the Pauli operator $\sigma_i$ acting on all systems in $B$ whose component in $\mathbf{z}$ is one and the identity operator acting on all other systems. Finally, let $\mathbf{D} \in \mathbb{F}_2^B$ denote the binary vector whose $b$-th component is one if vertex $b\in B$ has even degree in the subgraph $G-A$ and zero otherwise. We will refer to $\mathbf{D}$ as the \textit{degree vector} of $G - A$. 

First, we will apply our analysis of the entanglement of assistance from prior sections to graph states. To facilitate our discussion, we will refer to a quantum state $|\psi\rangle \in \mathcal{H}_B$ as a $\tau = 1$ state if $\tau_{N_B}(|\psi\rangle) = 1$ and as a $\tau = 0$ state if $\tau_{N_B}(|\psi\rangle) = 0$. Technical details and proofs are given in Appendix~\ref{app:unweighted-graph-results}.
The following lemma may be derived by using standard tools from the stabilizer formalism~\cite{Hein2004}.

\begin{lemma}\label{lemma:fidelity-graph-lemma}
    Suppose $N_B$ is even. We have $F(G_B, \tilde{G}_B) = 1$ if the matrix equation $\mathbf{\Gamma}_{BA}\mathbf{x} = \mathbf{D}$ has a solution $\mathbf{x}\in \mathbb{F}_2^A$ and $F(G_B, \tilde{G}_B) = 0$ otherwise. Here, the matrix multiplication is understood to be modulo 2. 
\end{lemma}

Remarkably, this simple lemma gives us a powerful criterion for characterizing the entanglement structure of states resulting from projective measurements on graph states. To see this, recall that Theorem~\ref{Th:n-tangle-upper-bound} implies that $\mathcal{A}^\tau(|G\rangle)$ is equal to $F(G_B, \tilde{G}_B)$. Hence, if $\mathbf{\Gamma}_{BA}\mathbf{x} = \mathbf{D}$ does not have a solution, Lemma~\ref{lemma:fidelity-graph-lemma} implies that $\mathcal{A}^\tau(|G\rangle) = 0$, which means that every possible projective measurement on the $A$ subsystems will result in only $\tau = 0$ states without regard to whether the measurement is local or global. Furthermore, Lemma~\ref{lemma:fidelity-graph-lemma} shows that if $\mathbf{\Gamma}_{BA} \mathbf{x} = \mathbf{D}$ \textit{does} have a solution, then there exists a measurement on $A$ that \textit{only} produces $\tau = 1$ states. We summarize these conclusions in the following theorem.   

\begin{theorem}\label{Th:main-graph-thm}~
    \begin{enumerate}
        \item If $\mathbf{\Gamma}_{BA} \mathbf{x} = \mathbf{D}$ does not have a solution in $\mathbb{F}_2^A$, then it is impossible to extract any kind of state other than a $\tau = 0$ state via projective measurements, either local or global, on the $A$ subsystem of $|G\rangle$.
        
        \item Suppose  $\mathbf{\Gamma}_{BA} \mathbf{x} = \mathbf{D}$ has a solution $\mathbf{x} \in \mathbb{F}_2^A$. Then there exists a projective measurement on the $A$ subsystem of $|G\rangle$ such that every possible post-measurement state is a $\tau = 1$ state. 
    \end{enumerate}
\end{theorem}

While the matrix $\mathbf{\Gamma}_{BA}$ is a relatively transparent object to work with, it is possible to phrase implications of Theorem \ref{Th:main-graph-thm} directly in terms of graph properties. A particularly simple example of this is presented in the following corollary.

\begin{corollary}\label{cor:tau-and-vertex-degrees}
    If there exists a vertex in $B$ having no neighbors in $A$ and even degree in $G - A$, then it is impossible to extract any kind of state other than a $\tau = 0$ state via projective measurements on $A$, either local or global. 
\end{corollary}

To see that this corollary holds, suppose there is a vertex in $B$ that has even degree in $G-A$, so that its associated component in $\mathbf{D}$ is one. If said vertex also has no neighbors in $A$, then every matrix element corresponding to its row in $\mathbf{\Gamma}_{BA}$ is zero. Hence, $\mathbf{D}$ cannot lie in the column space of $\mathbf{\Gamma}_{BA}$, meaning that the equation $\mathbf{\Gamma}_{BA}\mathbf{x} = \mathbf{D}$ has no solution. The conclusion of the corollary then follows from Theorem~\ref{Th:main-graph-thm}.

The results of this subsection extend previously-known observations regarding the $M$-tangle in an important way. It was shown in~\cite{schatzki2023} that for a graph state $|G'\rangle$ on an even number $M$ of qubits, we have $\tau_M(|G'\rangle) = 1$ if $G'$ contains a vertex of even degree and $\tau_M(|G'\rangle) = 0$ otherwise. It follows from the well-known Pauli measurement rules for graph states~\cite{Hein2004, Hein2006} that performing a $\sigma_z$ measurement on each qubit in $A$ of the graph state $|G\rangle$ results in the graph state $|G - A\rangle$. Hence, if $G - A$ has an even number of vertices all of odd degree, then $\mathcal{L}^\tau(|G\rangle) = 1$. However, if $G - A$ contains a vertex of even degree, the observation of~\cite{schatzki2021entangled} yields no conclusion about the value of $\mathcal{L}^\tau$. The results of this section provide a useful tool for analyzing such situations. 

\begin{figure}[h!]
    \centering
    \includegraphics[width=0.8\columnwidth]{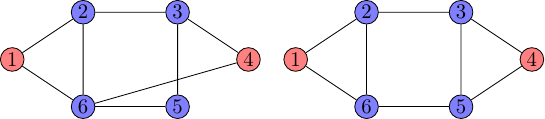}
    \caption{Two graphs on six vertices divided into qubits $A$ (red) and $B$ (blue).}\label{fig:even-deg-graphs}
\end{figure}

For example, for the graphs in Fig.~\ref{fig:even-deg-graphs}, every vertex in $G - A$ has even degree. Thus, the aforementioned analysis is inconclusive here. However, note that vertex 5 of the left graph is not connected to any vertex of $A$. Hence, by Corollary~\ref{cor:tau-and-vertex-degrees} every projective measurement on the $A$ subsystem of the associated graph state may only produce $\tau = 0$ states. By a straightforward computation, the right graph produces a solvable matrix equation $\mathbf{\Gamma}_{BA} \mathbf{x} = \mathbf{D}$, meaning that we can measure $A$ for the corresponding graph state in such a way that every state in the post-measurement ensemble maximizes the $N_B$-tangle. 

Our results also agree with and strengthen known examples in the literature. For instance, the authors of~\cite{deJong2024} considered the problem of extracting GHZ states from linear cluster states. As an example of their conclusions, they take $N_B = 4$ and consider each of the 35 possible configurations of $A$ and $B$ for a 7-qubit graph state in Fig. 1 of their paper. They identify in each case whether or not it is possible to extract a GHZ state over $B$ via LC on all qubits and LPM on $B$. It is not immediately clear whether allowing local unitaries and general projective measurements on $A$ might allow one to extract a GHZ state in the 20 configurations for which this is impossible under LC + LPM. However, Theorem~\ref{Th:main-graph-thm} shows that this is not the case for any configuration. In fact, a simple visual inspection of the graphs together with Corollary~\ref{cor:tau-and-vertex-degrees} shows that this is impossible in 16 of the configurations. Moreover, our results imply that it is impossible to extract \textit{any} $\tau = 1$ state, not just GHZ states.

\begin{figure}[h!]
    \centering
    \includegraphics[width=0.5\columnwidth]{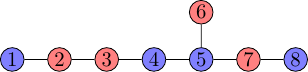}
    \caption{A graph on eight vertices divided into qubits $A$ (red) and $B$ (blue).}\label{fig:counterexample}
\end{figure}

One might wonder whether it is possible to generate an ensemble of all $\tau = 1$ states using \textit{local} measurements whenever $\mathbf{\Gamma}_{BA}\mathbf{x} = \mathbf{D}$ has a solution. If this were true, we could  strengthen the second part of Theorem~\ref{Th:main-graph-thm}. However, the graph of Fig.~\ref{fig:counterexample} provides a likely counterexample to such a conjecture. Numerically, we observed that $\mathcal{L}^\tau$ assumes a value of zero on the corresponding graph state ($<2\times 10^{-16}$ in the numerical optimization). On the other hand, the second part of Theorem~\ref{Th:main-graph-thm} implies that $\mathcal{A}^\tau$ assumes a value of one on this graph state. Thus, there are situations in which global measurements may have an advantage in extracting entanglement from graph states.

In summary, the chief contribution of our study of $\mathcal{L}^\tau$ on graph states is a characterization of achievable graph state transformations in terms of a simple matrix equation. This formalism offers two appealing advantages. First, while deciding whether an arbitrary graph state transformation using LC + LPM is achievable is an NP-complete problem, Theorem~\ref{Th:main-graph-thm} offers a polynomial-time test to determine whether \textit{specific} transformations are possible.  
The second advantage is that, when our test tells us that it is impossible to extract a certain kind of graph state, it does so for \textit{any} kind of local unitary operations and projective measurements, as opposed to the traditional LC + LPM framework.

\subsection{Insights from $\Lce$ and $\Lgme$ on graph states}

If the entanglement of a graph state $|G'\rangle$ over $B$ exceeds the localizable entanglement of the state $|G\rangle$, then it is impossible to extract $|G'\rangle$ via LPM on $A$ and LC operations because LC + LPM on graph states yield a post-measurement ensemble of LC-equivalent states~\cite{Hein2004,harikrishnan2023localizing}. We combine this observation with our bounds on $\Lce$ and $\Lgme$ in this section to identify scenarios in which it is impossible to extract states of interest from $|G\rangle$.

Our upper bound~\eqref{eq:LCE-bound} on $\Lgme$ assumes a value of zero if and only if there exists a nonempty, proper subset of vertices in $B$ that is disconnected from the rest of $G$. Thus, the upper bound on $\Lgme$ just confirms the fact that we cannot extract a state over $B$ that shows genuine multipartite entanglement if the initial state is already biseparable with one part contained in $B$. 
Note that this does not automatically mean $\Lgme$ is zero in all cases where $G$ is disconnected, or even when $G-A$ is disconnected. 

\begin{figure}[h!]
    \centering
    \includegraphics[width=0.53\columnwidth]{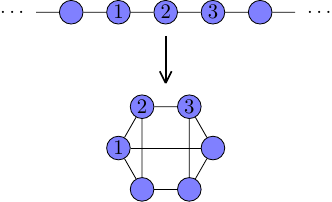}
    \caption{An example where AME state extraction via LC + LPM is not feasible. We consider the subsystem $A$ to be some subset of the unlabeled qubits of the cluster state (above). }\label{fig:CE-analysis}
\end{figure}

To understand the implications of our bounds on $\Lce$, it is useful to consider a few examples. Absolutely maximally entangled (AME) states are multipartite states such that all of the reduced states obtained by tracing out at least half of the subsystems are maximally mixed~\cite{Helwig2012}. It is known that AME graph states over five and six qubits are unique up to local unitaries~\cite{Scott2004, helwig2013}. The only other multi-qubit AME graph states that exist are locally unitarily equivalent to three-qubit GHZ states and two-qubit Bell pairs, respectively, and the methods of extracting GHZ states from linear cluster states and Bell pairs from arbitrary graphs via LC + LPM have been characterized in~\cite{deJong2024, raussendorf2001oneway}. In Fig.~\ref{fig:CE-analysis}, we consider the task of transforming a linear cluster state $|L_N\rangle$ over $N \geq 6$ qubits into the pictured 6-qubit AME graph state $|\AME_6\rangle$ via LPM on $A$ together with LC operations. We take $s$ to be the set of qubits labeled 1, 2, and 3, none of which are on the boundary of the line graph. If this state transformation task is possible, we must have $C(|\AME_6\rangle; s) \leq \Lce(|L_N\rangle; s) \leq C(|L_N\rangle; s)$ by Theorem~\ref{Th:CE-upper-bound}. However, the proposed inequality is violated because $C(|\AME_6\rangle; s) = 0.578125$ while $C(|L_N\rangle; s) = 0.5625$. Hence, it is not possible to extract a 6-qubit AME graph from a linear cluster state in such a way that the connectivity of $s$ with its neighbors changes as shown in Fig.~\ref{fig:CE-analysis}. Similarly, we see that the five-qubit AME state extraction proposed in Fig.~\ref{fig:CE-analysis-2} is impossible by considering values of the concentratable entanglement.

\begin{figure}[h!]
    \centering
    \includegraphics[width=0.3\columnwidth]{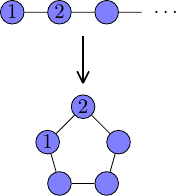}
    \caption{Another impossible AME state extraction using LC + LPM. We again take the measured subsystem $A$ to be some subset of the unlabeled qubits of the cluster state (above).}\label{fig:CE-analysis-2}
\end{figure}

A significant aspect of these examples comes from the fact that the concentratable entanglement of a graph state only depends on the local connectivity of $s$. It follows from a standard computation~\cite{Hein2004} of the subsystem purities for graph states that the concentratable entanglement of graph state $|G\rangle$ is equal to
\begin{equation}\label{eq:CE-adj-matrix}
    C(|G\rangle; s) = 1 - \frac{1}{2^{|s|}} \sum_{\gamma\subseteq s} \exp_2(- \rk(\mathbf{\Gamma}_{\gamma \overline{\gamma}})),
\end{equation}
where $\exp_2(x)=2^x$ and $\mathbf{\Gamma}_{\gamma \overline{\gamma}}$ is the block of the adjacencey matrix $\mathbf{\Gamma}$ of $G$ that describes the connectivity between $\gamma$ and its complement $\overline{\gamma} \coloneqq V\setminus \gamma$. Note that $\mathbf{\Gamma}_{\gamma \overline{\gamma}}$ only depends on the edges between the members of $\gamma$ and their neighbors in $\overline{\gamma}$. Hence, for a graph $G'$ containing $s$ among its vertices, we have $C(|G\rangle; s) = C(|G'\rangle; s)$ when the graphs obtained by taking the subgraphs in $G$ and $G'$ induced by $s$ and connecting their vertices to their respective neighbors within $G$ and $G'$ are the same. 

\begin{figure}[h!]
    \centering
    \includegraphics[width=0.75\columnwidth]{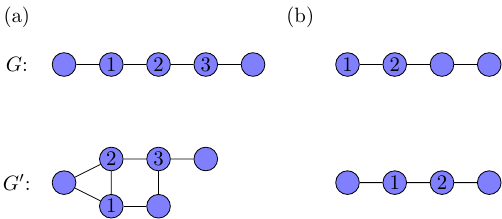}
    \caption{It is impossible to extract $|G'\rangle$ from $|G\rangle$ if locally, $G$ and $G'$ have the edge structure pictured. 
    }\label{fig:subgraphs}
\end{figure}

Hence, our examples do not only rule out the extraction of 5-qubit and 6-qubit AME states from linear cluster states in the manner shown in Figs.~\ref{fig:CE-analysis} and~\ref{fig:CE-analysis-2}. They eliminate the possibility of a much broader class of transformations. In particular, our prior analysis shows that if \textit{any} two graphs $G$ and $G'$ containing vertices 1, 2, and 3 are such that the edges between these vertices and their neighbors are as shown in Fig.~\ref{fig:subgraphs}, then it is impossible to extract $|G'\rangle$ from $|G\rangle$. Note that such a transformation is forbidden regardless of the connections between the neighbors of $s$ among themselves. In future work, it would be interesting to identify other classes of prohibited graph state transformations by analyzing the values of the CE for certain choices of $s$.

A common approach to studying the multipartite entanglement of graph states is to use the Schmidt measure~\cite{Hein2004}. One could in principle use the monotonicity of the Schmidt measure to study prohibited graph state extractions. However, the Schmidt measure can sometimes be hard to compute, even for graph states~\cite{Hein2004, schatzki2023}. A pleasing aspect of the approach of this subsection is that the concentratable entanglement of a graph state may be readily computed in terms of the adjacency matrix of the graph via~\eqref{eq:CE-adj-matrix}. 
With that said, for the purpose of studying allowable graph state transformations, it would be quite reasonable to simply consider the subsystem entropy vector as well. Finally, we add that in our context, computations involving the concentratable entanglement may also be approached through the stabilizer formalism framework of~\cite{Cullen2022}.

\subsection{Generalization to weighted graph states}

If the CZ gates in the graph state generating circuit~\eqref{eq:ord-graph-state} are replaced by controlled-phase (CP) gates, i.e., gates that act on a pair of qubits $a$ and $b$ as
\begin{equation}
    \mathrm{CP}_{ab}(\varphi_{ab}) = (|0\rangle\langle 0|)_a \otimes I_b + (|1\rangle\langle 1|)_a \otimes S(\varphi_{ab})_b 
\end{equation}
with $S(\varphi_{ab})_b = \mathrm{diag}(1, e^{i\varphi})$, then
the resulting state is known as a \textit{weighted graph state} \cite{Hein2006}. Experimentally, a promising approach for generating photon-photon controlled-Z gates comes from the use of cavity-QED systems \cite{duan2004scalable,nysteen2017limitations,pichler2017universal}. A drawback to this approach is that we do not yet have precise control over the phase that such a setup induces on the system, resulting in gates better modeled as CP gates~\cite{Firstenberg2013, Firstenberg2016, Tiarks2016, Thompson2017, Sagona-Stophel2020, Tiarks2018}. 
Hence, the coherent error produced by cavity-QED experiments may be represented by \textit{weighted} graph states. 
We may think of the phase $\varphi_{ab}$ of a CP operators as weighting the edge between vertices $a$ and $b$ in the underlying graph. The faultier the graph state generation process is, the further we would expect $\varphi_{ab}$ to land from $\pi$. 

In this section, we consider a protocol developed by Frantzeskakis et al.~in~\cite{Frantzeskakis2023} for extracting perfect GHZ states from weighted graph states via local measurements and local unitary rotations. We assume that all edge weights are uniform, i.e., $\varphi_{ab} = \varphi\in [0,2\pi)$ for all vertices $a$ and $b$. It was shown by Frantzeskakis et al.~that by starting with a weighted graph state corresponding to a linear graph with vertices labeled from 1 to $2N+1$ from left to right (for example, see Fig.~\ref{fig:Frantzeskakis}), one may obtain with a nonzero probability a perfect GHZ state by measuring out the qubits corresponding to evenly-labeled vertices in the eigenbasis of the rotated-$\sigma_x$ operator
\begin{equation}\label{eq:rot-x-meas}
    M_x(\varphi) = e^{-i\varphi \sigma_z/2}\sigma_x e^{i\varphi \sigma_z/2}
\end{equation}
and performing a local unitary rotation on the resulting state. The probability of obtaining an exact GHZ state with such a protocol is
\begin{equation}
    p_{\rm GHZ}(\varphi) = \frac{1}{2^N} \left|\sin\left(\frac{\varphi}{2}\right)\right|^{N},
\end{equation}
which exceeds that of traditional linear optical methods~\cite{Frantzeskakis2023}. 
Thus, even if an experimental configuration produces faulty linear cluster states as described above, there exists an operationally inexpensive probabilistic protocol for converting them into perfect GHZ states.

\begin{figure}[h!]
    \centering
    \includegraphics[width=0.72\columnwidth]{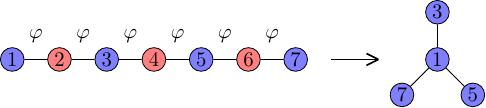}
    \caption{For a linear uniformly weighted graph state on an odd number of qubits, the protocol of Frantzeskakis et al.~prescribes a rotated $\sigma_x$ measurement on the evenly-labeled qubits (red), which produces with nonzero probability a perfect GHZ state over the remaining qubits (blue).}\label{fig:Frantzeskakis}
\end{figure}

We have observed numerically that the average post-measurement value of the $N_B$-tangle produced by this protocol seems to coincide with the value of $\mathcal{L}^\tau(|\Psi\rangle)$, the latter being very close to the value of $\mathcal{A}^\tau$ across the entire range of values for $\varphi \in [0,2\pi)$. For instance, see the top pane of Fig.~\ref{fig:mTangle-optimization}. In this sense, the protocol of Frantzeskakis \textit{et al.} is optimal and, moreover, there is only a slight advantage offered by using global instead of local measurements. Furthermore, we have observed that a similar situation occurs in other weighted graph states when we select a subset of qubits on which to perform the rotated measurement~\eqref{eq:rot-x-meas}. Though we do not offer a systematic analysis of this behavior, some examples may be observed in Fig.~\ref{fig:mTangle-optimization}.
\begin{figure}[h!]
    \includegraphics[scale=1.0]{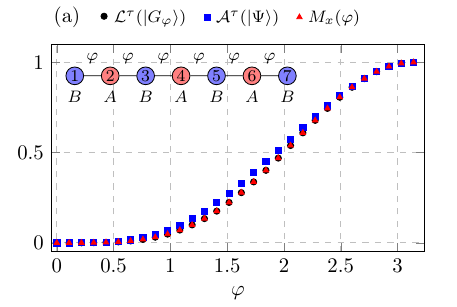}
    \includegraphics[scale=1.0]{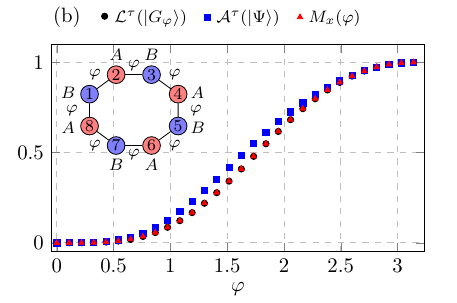}
    \begin{center}
        \includegraphics[scale=1.0]{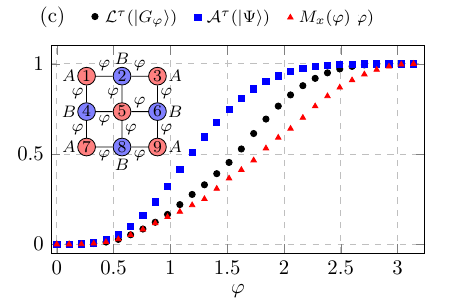}
    \end{center}
    \caption{Values of $\mathcal{L}^\tau(|\Psi\rangle)$ and $\mathcal{A}^\tau(|\Psi\rangle)$ for three weighted graph states $|G_{\varphi}\rangle$ with uniform edge weights. In red is the average post-measurement $N_B$-tangle produced by the measurement $M_x(\varphi)$ of~\eqref{eq:rot-x-meas} on all measured-out qubits. We present the results (a) for 7-qubit linear graph with qubits of even label measured out (see also Fig.~\ref{fig:Frantzeskakis}), (b) an 8-qubit ring graph with alternating qubits measured out, and (c) a $3\times 3$ lattice with the corner qubits and center qubit measured out.}
    \label{fig:mTangle-optimization}
\end{figure}

This leads to the question from what other noisy graph states it might be possible to extract a desired perfect graph state. In general, this is a challenging problem due to the lack of a stabilizer formalism for dealing with weighted graph states, as has been noted in~\cite{Hein2004}. However, if we only require that the probability of extracting the perfect GHZ state from a weighted graph state $|G_{\varphi}\rangle$ with \textit{uniform} edge weights $\varphi$ approaches a nonzero value as the error is suppressed, i.e., $\varphi \to \pi$, then Theorem~\ref{Th:main-graph-thm} implies that only resource states whose underlying graphs satisfy the matrix equation $\mathbf{\Gamma}_{BA} \mathbf{x} = \mathbf{D}$ will suffice. Furthermore, such is the case whenever the target perfect graph state is a $\tau = 1$ state, i.e., it has no vertices of even degree. Conversely, suppose the underlying graph $G$ of the noisy source state has no solution to the matrix equation. Note that the maximum probabilities $p_{\ghz}(\varphi)$ and $p_{\tau = 1}(\varphi)$ of extracting a GHZ state and a general $\tau = 1$ state, respectively, from $|G_{\varphi}\rangle$ are upper bounded by $\mathcal{A}^\tau(G_{\varphi})$ due to the construction of $\mathcal{A}^\tau$. Theorem~\ref{Th:main-graph-thm} implies that $\mathcal{A}^\tau(|G\rangle) = 0$.  Using the continuity bound of Corollary~\ref{cor:LE-local-bounds}, it follows that
\begin{equation}\label{eq:p-ghz-bound}
    p_{\ghz}(\varphi) \leq p_{\tau = 1}(\varphi) \leq \mathcal{L}^\tau(G_{\varphi}) \leq (2\sqrt{2} + 1)\|G_{\varphi} - G\|_1 = \mathcal{O}(|\varphi - \pi|^{1/2}). 
\end{equation}
 In Appendix~\ref{app:weighted-graph-results}, we show how the term $\|G_{\varphi} - G\|_1$ in~\eqref{eq:p-ghz-bound} may be exactly calculated entirely in terms of $\varphi$ and the full adjacency matrix $\mathbf{\Gamma}$ of $G$. Therein, we provide an upper bound of $\|G_\varphi - G\|_1$ that holds for small $\varphi$ and decays as $\mathcal{O}(|\varphi - \pi|^{1/2})$. We conclude that the decay rate of $p_{\ghz}(\varphi), p_{\tau = 1}(\varphi)$ to zero in the limit of vanishing noise parameter $\varphi$ is controlled by $\|G_{\varphi} - G\|_1$, which decays as $\mathcal{O}(|\varphi - \pi|^{1/2})$.

One might be inclined to use $\mathcal{A}^\tau$ as a predictor of our ability to extract GHZ states in this scenario. However, it is not immediately guaranteed that the ability to extract a perfect GHZ state from $|G_\varphi\rangle$ in a \textit{single} measurement branch automatically implies that $\mathcal{A}^\tau$ is high because the associated probability of this outcome may be low. Nevertheless, if there exists a protocol for which the probability that we can extract a perfect GHZ state from $|G_\varphi\rangle$ approaches some nonzero value as $\varphi \to \pi$, then Theorem~\ref{Th:main-graph-thm} and Corollary~\ref{cor:LE-local-bounds} guarantee that
\begin{equation}
    \mathcal{A}^\tau(|G_{\varphi}\rangle) \geq 1 - \mathcal{O}(|\varphi - \pi|^{1/2}).
\end{equation}
That is, for $\varphi$ sufficiently close to $\pi$, we can find a sequence of measurements on $A$ such that states in the post-measurement ensemble of $|G_{\varphi}\rangle$ will assume a high value of the $N_B$-tangle on average.


\section{Spin-half transverse field Ising model}
\label{sec:spin-models}

Quantum statistical mechanics is concerned with correlations in multipartite systems, and it is thus natural to consider the entanglement-theoretic properties of relevant spin systems. To this end, an understanding of how the localizable entanglement characterizes interesting behavior in these systems is helpful, and various formulations of the localizable entanglement have been successfully used to probe for phase transitions in spin-half systems~\cite{Verstraete-LE,Popp-LE,sadhukhan2017multipartite,krishnan2023controlling}. In this section, we demonstrate the usefulness of our formulation of the LME for this purpose.

We will consider the Transverse Field Ising model (TFIM) with periodic boundary conditions,
\begin{equation}\label{eq:tfim-H}
    H=-J\sum_{i=1}^{N}\sigma_{i}^{x}\sigma_{i+1}^{x}-h\sum_{i=1}^{N}\sigma_{i}^{z},
\end{equation}
where $\sigma^{x}_{N+1}\equiv\sigma^{x}_{1}$. The parameter $J$ represents the nearest-neighbor interaction strength and $h$ an external magnetic field. This model is known to possess a phase transition at $J=h$ that can be probed by properties of the ground state, denoted by $\left| \text{GS}_{N} \right>$. When $J>h$, the system is in a ferromagnetic phase. In this regime, neighboring spins tend to align with each other in order to minimize the ground state energy of the Hamiltonian. In the absence of any symmetry breaking term, none of the two possible alignment directions is favored ($\left| - \right>^{\otimes N}$ for the $-x$ direction and $\left| + \right>^{\otimes N}$ for the $+x$ direction) and the ground state of the Hamiltonian is entangled and may be taken to be $\left|\ghz_{N} \right> \approx \frac{1}{\sqrt{2}}( \left| + \right>^{\otimes N} + \left| - \right>^{\otimes N})$. When $J<h$ the system will be in the paramagnetic phase, and the spins will favor alignment in the $z$-direction in order to minimize the ground state energy. Since there is no $z$ interaction (only $x$), the ground state will approximate $\left| 0 \right>^{\otimes N}$.

\begin{figure}[h!]
    \includegraphics[scale=0.9]{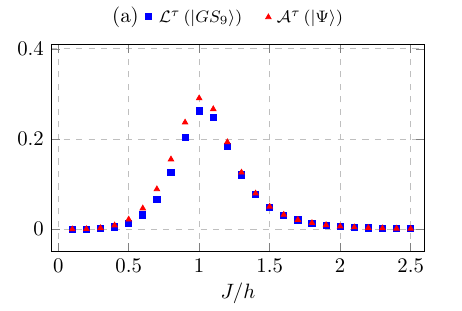}
    \includegraphics[scale=0.9]{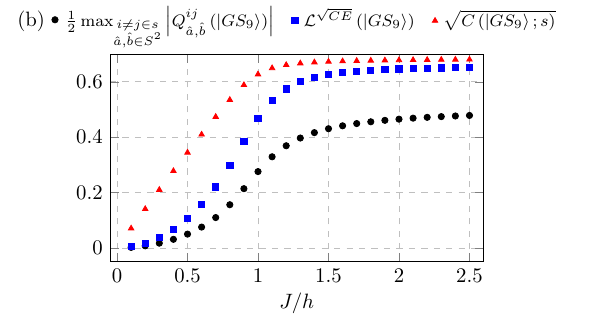}
    \caption{(a) Probing the phase transition using LME with the $N_B$-tangle as the seed measure. We calculate the transverse field Ising model with $N=9$. The measured out systems are $A=\left\{ 1, 3, 5, 7, 9 \right\}$, and the remaining qubits are $B=\left\{ 2, 4, 6, 8 \right\}$. In addition, we plot the upper bound $\mathcal{A}^\tau$ on $\mathcal{L}^\tau$ from Theorem~\ref{Th:n-tangle-upper-bound}. (b) We use the square-root concentratable entanglement as the seed measure for the same setup as in (a). Here, we also plot the correlation function lower bound from Theorem~\ref{Th:corr-func-lower-bound}and the upper bound from Theorem~\ref{Th:CE-upper-bound}.} 
    \label{fig:spinhalf-sym}
\end{figure}
We will study the effect of the phase transition at $J = h$ on the localizable entanglement using the $n$-tangle and concentratable entanglement as seed measures. For this purpose, we numerically diagonalize the relevant Hamiltonian and compute the LME of the ground state via the particle swarm optimization method described in Section~\ref{sec:numerics}.
In Fig.~\ref{fig:spinhalf-sym}, before the phase transition ($J<h$), the ground state is almost separable and the localizable entanglement $\mathcal{L}^\tau$ is low or close to zero. Once the $J>h$ regime is reached, the ground state starts to approach a GHZ-state and $\mathcal{L}^\tau$ approaches $1$ as $J$ becomes the dominant term. Note that in all regimes, $\mathcal{L}^\tau$ almost completely saturates its upper bound of $\mathcal{A}^\tau$ for the ground state of the Hamiltonian. When defined by the square-root concentratable entanglement, the localizable entanglement $\Lsqrtce$ increases as $J$ increases. In this case, the localizable entanglement $\Lsqrtce$ saturates its upper bound from Theorem~\ref{Th:CE-upper-bound} when the system undergoes a phase transition ($J>h$).
\begin{figure}[h!]
    \includegraphics[scale=0.9]{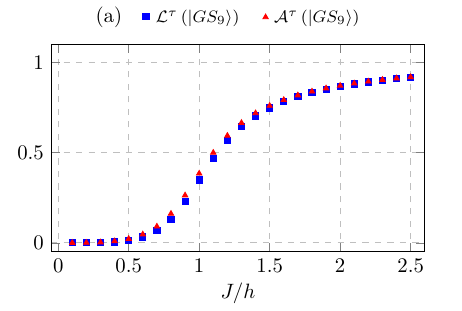}
    \includegraphics[scale=0.9]{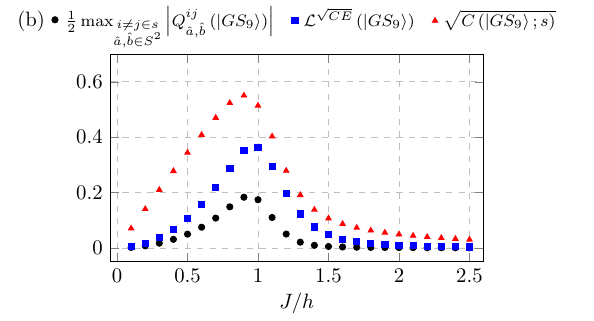}
    \caption{(a) Probing the phase transition using LME with the $N_B$-tangle as the seed measure. We calculate the transverse field Ising model with $N=9$ in the presence of a small symmetry breaking field. The measured out systems are $A=\left\{ 1, 3, 5, 7, 9 \right\}$, and the remaining qubits are $B=\left\{ 2, 4, 6, 8 \right\}$. (b) We use the concentratable entanglement as the seed measure for the same setup as in (a). }
    \label{fig:spinhalf}
\end{figure}

In realistic setups, the Hamiltonian will acquire symmetry breaking terms. To simulate this, we also explore the Transverse Field Ising Model in the presence of an additional small symmetry breaking term $-h_{x}\sum_{i}^{N}\sigma_{i}^{x}$ along the $x$-direction in Fig. \ref{fig:spinhalf}. When defined by the $n$-tangle, the localizable entanglement $\Lt$ increases as $J/h$ increases in the $J<h$ regime. After the phase transition occurs at $J=h$, the ground state will have a favored direction ($x$) and the ground state will be $\left| + \right>^{\otimes N}$. As $J$ becomes the dominant term, the state becomes more separable, so the localizable entanglement $\mathcal{L}^\tau$ approaches $0$. The behavior is similar when defined by the square-root concentratable entanglement. In fact, the localizable entanglement $\mathcal{L}^{\sqrt{CE}}$ increases with $J/h$ until the phase transition occurs and rapidly decreases as $J$ becomes the dominant term.

We anticipate that our observations can be useful to further study critical behavior in spin systems. Since a divergence in correlation length is a diagnostic of a quantum phase transition, and since $L^{\sqrt{CE}}$ is lower bounded by the two-point spin correlation lengths, we infer that $L^{\sqrt{CE}}$ can serve as a useful tool for identifying quantum phase transitions. This generalizes to the multipartite scenario the viewpoint of Verstraete et al.~\cite{Verstraete-LE}, who make the same case for the bipartite LE defined by the two-point concurrence. Moreover, our numerical results demonstrate that our bounds on the LME show the same qualitative behavior as the LME, with the bound on $\mathcal{L}^\tau$ appearing to be particularly tight. The latter implies that local measurements are near-optimal for the task of localizing entanglement in TFIM ground states. More precisely, it is possible to localize nearly as much $n$-tangle using just local measurements as it is with general global measurements. This indicates that it is possible to use the much simpler bounds to study critical phenomena in place of the LME, whose calculation involves a formidable optimization. 

\section{Discussion and outlook}\label{sec:discussion}

In this paper, we study formulations of the localizable multipartite entanglement (LME) and multipartite entanglement of assistance (MEA) as metrics of quantifying the amount of multipartite entanglement that can be localized onto a subsystem by measuring out its complement. These figures of merit are inspired by protocols that extract useful entangled states from a resource state via measuring and discarding subsystems. We derive bounds on these quantities as explicit functions of reduced density matrices and numerically analyze their tightness. This enables the theoretical study of these quantities without performing tedious optimizations and opens the door to experimental detection, as it is more viable to probe reduced density matrices than full state vectors.  

We study several applications to illustrate the utility of our bounds. First, we work out a simple matrix equation criterion for deciding whether transformations between graph states under measurements and general local unitaries are possible. While this polynomial-time criterion is not conclusive for all transformations, it is significant because deciding whether an \textit{arbitrary} graph state may be extracted from another is an NP-complete problem. We also point out the advantages over other entanglement-theoretic approaches of using the concentratable entanglement (CE) to study graph state transformations. As another application of our bounds, we characterize the typical amount of entanglement that can be localized with measurements in high-dimensional systems. We find that even though reduced states display near-minimal entanglement for large system sizes, they typically have a convex decomposition of pure states that assume near-maximal $n$-tangle and CE values. One interpretation of this result is that reduced density matrices alone are generally insufficient for quantifying correlations between members of subsystems in high-dimensional spaces; measurements are needed to \textit{reveal} quantum correlations. 
Finally, we showcase preliminary numerical results indicating how our bounds may be used to detect critical behavior in ground states of the Ising model.
We hope to inspire further study of our bounds in such spin systems to analyze their critical behavior. 

Due to its broad nature, we believe that our work may find applications in a variety of fields. Indeed, the general framework of targeting desired entangled states via partial measurements is a feature of many active research areas, including measurement-based quantum computation, state preparation, and quantum networks. For future work, it would be interesting and relevant to develop ways of experimentally measuring the LME and MEA directly. It would also be worth benchmarking the efficacy of our matrix equation criterion in deducing practical graph state transformations. The matrix equation should be unsolvable for most graphs in order for it to be most helpful in paring down the set of candidate resource graphs for a given target graph. To this end, we aim to investigate the solvability of this equation over important families of graphs. On the technical side, it would be interesting to extend the LME framework to higher-rank projective measurements, or even more generally POVMs with suitably-defined post-measurement states.
A challenge in this generalization is to deal with possibly mixed post-measurement states, which would have to be analyzed using mixed-state entanglement measures that are typically hard to evaluate.
Another interesting open problem is to analytically derive a concentration inequality for $\mathcal{L}^\tau$ and find its limiting behavior. Such a result could be interpreted as a statement about the generic maximal probability with which one can extract GHZ states (or any other $\tau = 1$ states). Our work shows that the values of $\mathcal{L}^\tau$ concentrate around its mean, though this value remains unknown analytically. Finally, a more general goal would be to strengthen our analytical results regarding local measurements and pursue more conclusions about the localizability of specifically GHZ-type entanglement. 

\paragraph*{Code Availability. } The code used to produce the numerical computations of this work is available at~\cite{Github2025}.

\paragraph*{Acknowledgements.} This work was supported by NSF grant no.~2137953.
We would like to thank Edwin Barnes, Jacob Beckey and Eric Chitambar for helpful discussions.

\bibliographystyle{quantum}
\bibliography{references}

\pagebreak

\appendix

\section{Proof of Theorem \ref{Th:n-tangle-upper-bound}}
\label{app:thm-1}

That the localizable entanglement of $|\Psi\rangle$ is bounded above by its multipartite entanglement of assistance follows immediately from the definition of these quantities. The expression for the MEA of $|\Psi\rangle$, as defined by the concurrence (equivalently, the 2-tangle), in terms of $\Psi_{B}$ has been proved for when $|\Psi\rangle$ and $A$ are respectively taken to be a 4-qubit state and a two-qubit subsystem in a somewhat different context~\cite{Lausten2003}. 
Our theorem follows from an almost unmodified version of this proof. 
For the sake of clarity, we spell out in full detail the proof of~\cite{Lausten2003} and its supporting results applied to our multi-qubit hypothesis.

We start by proving a couple of simple lemmas. In what follows, we will retain our notational convention that for an $M$-qubit state $|v\rangle$, the ket $|\tilde{v}\rangle$ refers to the Wooter's tilde of $|v\rangle$, or $\sigma_y^{\otimes M}|v^\ast\rangle$. Similarly, $\tilde{\rho} = \sigma_y^{\otimes M}\rho^*\sigma_y^{\otimes M}$ for an $M$-qubit density matrix $\rho$. Lemma \ref{lemma:rho-decomp-lemma} below has previously been observed by~\cite{Wooters1998EOF, Schrodinger1936, Hadjisavvas1981, Hughston1993}, among others.

\begin{lemma}\label{lemma:rho-decomp-lemma}
    Let $\rho$ be an arbitrary density operator of rank $m$ over some multi-qubit Hilbert space. Let $\rho = \sum_{i=1}^m \lambda_i |v_i\rangle\langle v_i|$ be its spectral decomposition. Then for any $m\times m$ unitary matrix $U$, there exists a set of vectors $\{|x_i\rangle\}_{i=1}^m$ satisfying
    \begin{equation}
        \langle x_i|\Tilde{x}_j\rangle = (U\tau U^T)_{ij},
    \end{equation}
    where $\tau$ is the $m\times m$ matrix with elements $\tau_{ij} = \sqrt{\lambda_i \lambda_j} \langle v_i|\Tilde{v}_j\rangle$,
    and 
    \begin{equation}
        \rho = \sum_{i=1}^m |x_i\rangle\langle x_i|.
    \end{equation}
	That is, the vectors $\{|x_i\rangle\}_{i=1}^m$ form a decomposition of $\rho$.
\end{lemma}

\begin{proof}
Define $|x_i\rangle \coloneqq \sum_{j=1}^m \sqrt{\lambda_j} U_{ij}^\ast |v_j\rangle$ for $i = 1, \dots, m$. Then
\begin{equation}
    \sum_{i=1}^m |x_i\rangle\langle x_i| = \sum_{i=1}^m \sum_{j,k=1}^m \sqrt{\lambda_j\lambda_k} U_{ij}^\ast U_{ik} |v_j\rangle\langle v_k| = \sum_{j=1}^m \lambda_j |v_j\rangle\langle v_j| = \rho,
\end{equation}
where in the second equality we used the orthonormality of the columns of $U$. 

Next, observe that 
\begin{align}
    \langle x_i|\Tilde{x}_j\rangle &= \sum_{k,l=1}^m \sqrt{\lambda_k\lambda_l}U_{ik} U_{jl} \langle v_k|\tilde{v}_l\rangle \\ &= 
    \sum_{k,l=1}^m U_{ik} \left(\sqrt{\lambda_k\lambda_l} \langle v_k|\tilde{v}_l\rangle\right)(U^T)_{lj} \\ &= 
    \sum_{k,l=1}^m U_{ik} \tau_{kl} (U^T)_{lj} \\ &=
    (U\tau U^T)_{ij}.
\end{align}
This completes the proof.
\end{proof}

The following Lemma is taken directly from~\cite[Cor.~4.4.4]{Horn1990}.
\begin{lemma}\label{lemma:Takagi} \textnormal{(Takagi's factorization)} 
    Let $A$ be an $m\times m$ complex symmetric matrix. Then there exists an $m\times m$ unitary matrix and a real nonnegative diagonal matrix $\Sigma = \textnormal{diag}(\sigma_1, \dots, \sigma_m)$ such that $A = U\Sigma U^T$. The diagonal entries of $\Sigma$ are the nonnegative square roots of the eigenvalues of $AA^\ast$.
\end{lemma}

Note that for an $M$-qubit density operator $\rho$ of rank $m$, the operators $\sqrt{\rho}$ and $\tilde{\rho}$ are also of rank $m$, implying that $\sqrt{\rho}\tilde{\rho}\sqrt{\rho}$ and, hence, $\sqrt{\sqrt{\rho}\tilde{\rho}\sqrt{\rho}}$ have rank at most $m$. As a result, the multiplicity of zero as an eigenvalue of the operator $\sqrt{\sqrt{\rho}\tilde{\rho}\sqrt{\rho}}$ is at least $2^M - m$. With this in mind, we are ready to prove the next lemma, which is a modified version of a result of~\cite{Wooters1998EOF}.

\begin{lemma}\label{lemma:verstraete-basis}
    Let $\rho$ be a density operator of rank $m$ over $M$ qubits, where $M$ is even. Take $l_1, \dots, l_m$ to be nonnegative eigenvalues of $\sqrt{\sqrt{\rho} \tilde{\rho} \sqrt{\rho}}$ such that the other $2^M - m$ eigenvalues of $\sqrt{\sqrt{\rho} \tilde{\rho} \sqrt{\rho}}$ are all zero.\footnote{Note that we allow for the possibility that some of the eigenvalues $l_1, \dots, l_m$ may be zero here. We are only demanding that zero occurs as an eigenvalue with a multiplicity of \textit{at least} $2^M - m$, with all possible nonzero eigenvalues being labeled by a subset of $\{l_1, \dots, l_m\}$.} Then there exists a set of vectors $\{|x_i\rangle\}_{i=1}^m$ in this $M$-qubit Hilbert space such that
    \begin{equation}
        \langle x_i|\tilde{x}_j\rangle = l_i \delta_{ij} \qquad \text{and} \qquad \rho = \sum_{i=1}^m |x_i\rangle\langle x_i|.
    \end{equation}
\end{lemma}

\begin{proof}
Let $\rho = \sum_{i=1}^m \lambda_i |v_i\rangle\langle v_i|$ be the spectral decomposition of $\rho$. Define the $m\times m$ matrix $\tau$ so that its elements are $\tau_{ij} = \sqrt{\lambda_i\lambda_j}\langle v_i|\tilde{v}_j\rangle$. Noting that $\sigma_y^\ast = -\sigma_y$, we have 
\begin{equation}
    \langle v_i|\tilde{v}_j\rangle = 
    \langle v_i|\sigma_y^{\otimes M}|v_j^\ast\rangle =
    (\langle v_j^\ast| \sigma_y^{\otimes M}| v_i\rangle)^\ast = 
    \langle v_j|(\sigma_y^\ast)^{\otimes M})|v_i^\ast\rangle =
    \langle v_j|\sigma_y^{\otimes M}|v_i^\ast\rangle = \langle v_j|\tilde{v}_i\rangle,
\end{equation}
since $M$ is even. Hence, the matrix $\tau$ is symmetric. 

It follows from Lemma \ref{lemma:Takagi} that there exists an $m\times m$ unitary matrix $V$ such that $\tau = V\Sigma V^T$, where $\Sigma = \text{diag}(s_1, \dots, s_m)$ and the $s_i$'s are real and nonnegative. Let $U = V^\dagger$. Then $\Sigma = U\tau U^T$. By Lemma \ref{lemma:rho-decomp-lemma}, there exists a set of vectors $\{|x_i\rangle\}_{i=1}^m$ such that $\rho = \sum_{i=1}^m |x_i\rangle\langle x_i|$ and
\begin{equation}\label{eq:xi-xj-overlap}
    \langle x_i|\Tilde{x}_j\rangle = (U\tau U^T)_{ij} = s_i \delta_{ij}.
\end{equation} Furthermore, $s_1^2, \dots, s_m^2$ are the eigenvalues of $\tau \tau^\ast$. 

Now observe that for $1\leq i,j \leq m$,
\begin{equation}
    (\tau \tau^\ast)_{ij} = \sum_{k=1}^m \tau_{ik}\tau^\ast_{kj} =  
    \sum_{k=1}^m \lambda_k \sqrt{\lambda_i\lambda_j} \langle v_i|\tilde{v}_k\rangle\langle\tilde{v}_k|v_j\rangle  = \sqrt{\lambda_i\lambda_j}\langle v_i|\tilde{\rho}|v_j\rangle = \langle v_i|\sqrt{\rho}\tilde{\rho}\sqrt{\rho}|v_j\rangle. 
\end{equation}
Therefore, the matrix representation of $\sqrt{\rho}\tilde{\rho}\sqrt{\rho}$ with respect to an orthonormal extension of the orthonormal set $\{|v_i\rangle\}_{i=1}^m$ to a basis of the entire $M$-qubit Hilbert space has $\tau \tau^\ast$ as a diagonal block with all other elements equal to zero. Thus, $s_1^2, \dots, s_m^2$ are eigenvalues of $\sqrt{\rho} \tilde{\rho} \sqrt{\rho}$ and all other eigenvalues of this operator are zero. Since $\sqrt{\rho}\tilde{\rho}\sqrt{\rho}$ is Hermitian, this implies again that $s_1, \dots, s_m$ are nonnegative eigenvalues of $\sqrt{\sqrt{\rho} \tilde{\rho} \sqrt{\rho}}$ and all other eigenvalues are zero. The claim of the lemma then follows from \eqref{eq:xi-xj-overlap}.
\end{proof}

We are now ready to prove Theorem \ref{Th:n-tangle-upper-bound}. The proof follows along the lines of~\cite{Lausten2003}.

\begin{proof}[Proof of Theorem~\ref{Th:n-tangle-upper-bound}] 
Let $\{(p_i, |\phi_i\rangle)\}_{i=1}^{d_A}$ be an arbitrary ensemble over the $B$ systems produced by a (not-necessarily local) projective measurement on the $A$ systems of $|\Psi\rangle$.
In that case,
\begin{align}
    \sum_{i=1}^{d_A} p_i \tau_n(|\phi_i\rangle) 
    &=\sum_{i=1}^{d_A} p_i F(|\phi_i\rangle\langle \phi_i|, |\Tilde{\phi_i}\rangle\langle\Tilde{\phi_i}|) \\
    &\leq F\bigg(\sum_{i=1}^{d_A} p_i |\phi_i\rangle\langle \phi_i|, \sum_{i=1}^{d_A} p_i |\Tilde{\phi_i}\rangle\langle\Tilde{\phi_i}|\bigg) \\
    &=\label{eq:MEA-n-tangle} F(\Psi_{B}, \Tilde{\Psi}_{B}),
\end{align}
The first equality follows from the definition \eqref{eq:n-tangle-def} of the $n$-tangle. The inequality follows from the joint concavity of fidelity. The second equality follows from the definition of the partial trace together with the observation that, given a measurement basis $\{|v_i\rangle\}_{i=1}^{d_A}$ that achieves the post-measurement ensemble chosen above, each post-measurement state may be written as $|\phi_i\rangle = (\langle v_i|_A \otimes I_{B})|\Psi\rangle/ \sqrt{p_i}$.

By Lemma \ref{lemma:verstraete-basis}, there exists a set of vectors $\{|x'_i\rangle\}_{i=1}^{\rk(\Psi_B)}$ with the properties that $\Psi_{B} = \sum_{i=1}^{\rk(\Psi_B)} |x'_i\rangle\langle x'_i|$ and $\langle x'_i|\tilde{x}'_i\rangle = l_i \delta_{ij}$, where $(l_i)_{i=1}^{\rk(\Psi_B)}$ denotes a collection of not-necessarily distinct nonnegative eigenvalues of $\sqrt{\sqrt{\Psi_{B}}\tilde{\Psi}_{B}\sqrt{\Psi_{B}}}$ such that all other eigenvalues are zero.\footnote{We emphasize that due to the fact that we enumerate the vectors $|x_i\rangle$ by indices ranging up to only $\rk(\Psi_B)$, as opposed to $\min(d_A, d_B)$, there is no need to assume that $\Psi_B$ is full rank.} Normalizing these vectors, we write $|x_i\rangle \coloneqq |x'_i\rangle/\sqrt{q_i}$, where $q_i \coloneqq \langle x'_i|x'_i\rangle$. Then for an arbitrary (not-necessarily product) basis $\{|e_i\rangle_A\}_{i=1}^{d_A}$ of $\mathcal{H}_A$, we have that
\begin{equation}
    |\Phi\rangle \coloneqq \sum_{i=1}^{\rk(\Psi_B)} \sqrt{q_i} |e_i\rangle_A \otimes |x_i\rangle_B   
\end{equation}
is a purification of $\Psi_{B}$. (Since $\rk(\Psi_B) \leq d_A$, we are able to write the above linear combination without issue.) 
Since $|\Psi\rangle$ is also a purification of $\Psi_{B}$, there exists a unitary $V_A$ such that $|\Phi\rangle = (I_B \otimes V_A)|\Psi\rangle$.
The ensemble $\{(q_i, |x_i\rangle)\}_{i=1}^{{\rk(\Psi_B)}}$ may be obtained by measuring out the $A$ systems of $|\Phi\rangle$ in the basis $\{|e_i\rangle\}_{i=1}^{d_A}$, so we obtain the same ensemble by measuring out the $A$ systems of $|\Psi\rangle$ in the basis $\{V_A^\dagger |e_i\rangle\}_{i=1}^{d_A}$. The average entanglement for the ensemble $\{(q_i, |x_i\rangle)\}_{i=1}^{\rk(\Psi_B)}$ is 
\begin{align}
    \sum_{i=1}^{\rk(\Psi_B)} q_i \tau_n(|x_i\rangle) &= \sum_{i=1}^{\rk(\Psi_B)} q_i |\langle x_i|\tilde{x}_i\rangle| \\
    &= \sum_{i=1}^{\rk(\Psi_B)} |\langle x'_i|\tilde{x}'_i\rangle| \\ 
    &= \sum_{i=1}^{\rk(\Psi_B)} l_i \\
    &= \tr\bigg(\sqrt{\sqrt{\Psi_{B}}\tilde{\Psi}_{B}\sqrt{\Psi_{B}}}\bigg) = F(\Psi_{B}, \tilde{\Psi}_{B}).
\end{align}
Combining this result with \eqref{eq:MEA-n-tangle}, it follows that $\mathcal{A}^\tau(|\Psi\rangle) = F(\Psi_{B}, \tilde{\Psi}_{B})$.
\end{proof}

\section{Proof of Theorems~\ref{Th:cgme-upper-bound} and~\ref{Th:CE-upper-bound}}\label{app:ub-proofs}
Let $\mathcal{D}(\mathcal{H}_B)$ denote the space of density operators over $\mathcal{H}_B$. 

\begin{proof}[Proof of Theorem \ref{Th:cgme-upper-bound}]
Note that the function from $\mathcal{D}(\mathcal{H}_B)$ to the reals defined by $\rho \mapsto \tr(\rho^2)$ is convex (see, e.g., \cite[Thm.~2.10]{carlen2010trace}), so that the function $\rho \mapsto 1- \tr(\rho^2)$ is concave. Also, the function $t\mapsto \sqrt{2t}$ is concave and non-decreasing. Therefore, the function $\rho \mapsto \sqrt{2(1 - \tr(\rho^2)}$ is concave as well. Consider now a post-measurement ensemble $\{(p_i, |\varphi_i\rangle)\}_i$ produced by a (not-necessarily local) measurement on the $A$ systems of $|\Psi\rangle$. Then 
\begin{align}
    \sum_i p_i C_{\text{GME}}(|\varphi_i\rangle) &=
    \sum_i p_i \min_{\varnothing \subsetneq \gamma \subsetneq B} \sqrt{2(1- \tr((\varphi_i)_{\gamma}^2)}\\ &\leq 
    \min_{\varnothing \subsetneq \gamma \subsetneq B} 
    \sum_i p_i \sqrt{2(1- \tr((\varphi_i)_{\gamma}^2)}\\
    &\leq 
    \min_{\varnothing \subsetneq \gamma \subsetneq B} 
    \sqrt{2\bigg(1- \tr\bigg[\bigg(\sum_i p_i(\varphi_i)_\gamma\bigg)^2\bigg]\bigg)} \\
    &= \min_{\varnothing \subsetneq \gamma \subsetneq B} 
    \sqrt{2(1- \tr(\Psi_{\gamma}^2))} 
\end{align} 
The bound on the multipartite entanglement of assistance and, hence, the localizable entanglement $\Lgme$ follows. 
\end{proof}

\begin{proof}[Proof of Theorem~\ref{Th:CE-upper-bound}] 
The proof is similar to the one of Theorem~\ref{Th:cgme-upper-bound}. Recall that for any $\gamma \subseteq B$, the function from $\mathcal{D}(\mathcal{H}_B)$ to the reals defined by $\rho\mapsto \tr(\rho^2)$ is convex. Consider a post-measurement ensemble $\{(p_i, |\varphi_i\rangle)\}_i$ produced by a (not-necessarily local) measurement on the $A$ systems of $|\Psi\rangle$. It then follows that for any $\varnothing \subsetneq s \subseteq B$,
\begin{align}
    \sum_i p_i \CE{|\varphi_i\rangle}{s} &=
    \sum_i p_i\bigg(1 - \frac{1}{2^{|s|}}\sum_{\gamma\subseteq s} \tr((\varphi_i)_\gamma^2)\bigg) \\
    &= 1 - \frac{1}{2^{|s|}}\sum_{\gamma\subseteq s} \sum_i p_i \tr((\varphi_i)_\gamma^2) \\
    &\leq 1 - \frac{1}{2^{|s|}}\sum_{\gamma\subseteq s} \tr\bigg(\bigg[\sum_i p_i (\varphi_i)_\gamma\bigg]^2\bigg) \\
    &= 1 - \frac{1}{2^{|s|}}\sum_{\gamma\subseteq s} \tr(\Psi_{\gamma}^2) \\
    &= \CE{|\Psi\rangle}{s},
\end{align}
from which the bounds of~\eqref{eq:LCE-bound} follow.
\end{proof}

We note that our proof of Theorem~\ref{Th:CE-upper-bound} is essentially an alternate proof of the monotonicity of the CE under local von Neumann measurements and discarding of subsystems. Our approach has the advantage that it succinctly avoids the subtleties that come with LOCC protocols that transform an initial state into states belonging to a Hilbert space of smaller dimension addressed in Section~\ref{sec:bounds}. 

\section{Proof of Theorem \ref{Th:corr-func-lower-bound}}\label{app:LCE-corr-func-LB}

In this section, we will prove a lower bound on $\Lce$ and $\Lsqrtce$ in terms of spin-correlation functions $Q_{\vec{a}, \vec{b}}^{ij}$. Recall that for any two qubits $i,j \in B$ and unit vectors $\vec{a}, \vec{b} \in \mathbb{R}^3$, we define \begin{equation}
    Q_{\vec{a}, \vec{b}}^{ij}(|\Psi\rangle) \coloneqq \langle \Psi|(\vec{a}\cdot \vec{\sigma}^{(i)}) \otimes (\vec{b} \cdot \vec{\sigma}^{(j)})|\Psi\rangle - \langle \Psi|(\vec{a} \cdot \vec{\sigma}^{(i)})|\Psi\rangle\langle \Psi|(\vec{b} \cdot\vec{\sigma}^{(j)})|\Psi\rangle.
\end{equation}

\begin{proof}[Proof of Theorem \ref{Th:corr-func-lower-bound}]
Let $E$ denote the two-qubit concurrence~\cite{Wootters1997}. Choose two different qubits $l,m\in s$, and let $|\varphi\rangle_{B_lB_m}$ be a pure state in the joint Hilbert space of $B_l$ and $B_m$.
For the corresponding marginal states, we then have
\begin{equation}\label{eq:2-qubit-concur-relation}
   \CE{|\varphi\rangle_{B_lB_m}}{\{l, m\}} = 1 - \frac{1}{2^2}\left(1 + \tr(\varphi_{B_l}^2) + \tr(\varphi_{B_m}^2) + \tr(\varphi^2)\right) = \frac{1}{2} \left(1 - \tr(\varphi_{B_l}^2) \right) = \frac{1}{4} E(|\varphi\rangle)^2.
\end{equation}
The last equality may be seen by writing the subsystem purities $\tr(\varphi_{B_l}^2)$, $\tr(\varphi_{B_m}^2)$, and the two-qubit concurrence $E(|\varphi\rangle)$ in terms of the Schmidt coefficients of $|\varphi\rangle$ (see, e.g., \cite{Wootters1997, Rungta2001}).

Let now  $\lbrace(p_{i,j}, |\varphi_{i,j} \rangle_{B_lB_m}) \rbrace_{i,j}$ be a post-measurement ensemble produced by a local projective measurement on all qubits of $|\Psi\rangle$ apart from qubits $B_l$ and $B_m$ in $B$ that maximizes the average concurrence $\sum_{i,j} p_{i,j} E(|\varphi_{i,j}\rangle)$.
Here, $E$ is again the two-qubit concurrence.
We view this measurement as consisting of two steps: first measuring the $A$-systems of $|\Psi\rangle$ with some projective measurement determined by a basis $\lbrace |i\rangle\langle i|_A\rbrace_i$ of the $A$-qubits, and then measuring out the remaining qubits in $B\setminus B_lB_m\equiv (B_lB_m)^c$ with some projective measurement determined by a basis $\lbrace |j\rangle\langle j|_{(B_lB_m)^c}\rbrace_j$. 
With these choices, we have
\begin{align}
	|\varphi_{i,j}\rangle_{B_lB_m} &= \frac{1}{\sqrt{p_{i,j}}} \left(I_{B_lB_m} \otimes \langle i|_A \otimes \langle j|_{(B_lB_m)^c}\right) |\Psi\rangle\\
	p_{i,j} &= \langle \Psi| \left(I_{B_lB_m} \otimes |i\rangle\langle i|_A \otimes |j\rangle\langle j|_{(B_lB_m)^c}\right) |\Psi\rangle,
	\intertext{and the intermediate state ensemble $\lbrace (p_i, |\omega_i\rangle_B)\rbrace_i$ after measuring just the qubits in $A$ is given by}
	|\omega_i\rangle_B &= \frac{1}{\sqrt{p_i}} \left( I_B\otimes \langle i|_A \right) |\Psi\rangle\\
	p_i &= \langle \Psi| \left(I_B \otimes |i\rangle\langle i|_A \right) |\Psi\rangle.
\end{align}
The two ensembles are related via measuring the states $|\omega_i\rangle_B$ with respect to the basis $\lbrace |j\rangle\langle j|_{(B_lB_m)^c}\rbrace_j$:
\begin{align}
	|\varphi_{i,j}\rangle_{B_lB_m} &= \frac{1}{\sqrt{p_{j|i}}} \left(I_{B_lB_m}\otimes \langle j|_{(B_lB_m)^c} \right) |\omega_i\rangle_B\\
	p_{j|i} &= \frac{p_{i,j}}{p_i}.
\end{align}
Here, $p_{j|i}$ is the conditional probability of obtaining state $|\varphi_{i,j}\rangle_{B_lB_m}$ from measuring the state $|\omega_i\rangle_B$.

Let $\hat{a}, \hat{b}$ be arbitrary unit vectors in $\mathbb{R}^3$. Our goal is to relate the two-point correlation function $Q_{\hat{a},\hat{b}}^{lm}$ to the two-qubit concurrence $E$, switch to the concentratable  entanglement $\CE{|\varphi_{i,j}\rangle}{\lbrace l,m\rbrace}$ via \eqref{eq:2-qubit-concur-relation}, and then relate it to the concentratable  entanglement of the ensemble states in $(p_i, |\omega_i\rangle_B)_i$ using the LOCC monotonicity property of the concentratable  entanglement.
This is achieved using the following chain of inequalities:
\begin{align}
	\left|Q_{\hat{a},\hat{b}}^{lm}(|\Psi\rangle)\right| &\leq \sumi_{i,j} p_{i,j} E(|\varphi_{i,j}\rangle_{B_lB_m}) \label{eq:verstraete}\\
	&= 2 \sumi_{i,j} p_{i,j} \sqrt{ \CE{|\varphi_{i,j}\rangle_{B_lB_m}}{\lbrace l,m\rbrace} } \label{eq:concurrence-ce}\\
	&= 2 \sumi_i p_i \sumi_j p_{j|i}\, \sqrt{ \CE{ |\varphi_{i,j}\rangle_{B_lB_m} \otimes |j\rangle_{(B_lB_m)^c} }{\lbrace l,m\rbrace}} \label{eq:split-conditional} \\
	&\leq 2 \sumi_i p_i\, \sqrt{ \CE{|\omega_i\rangle_B}{\lbrace l,m\rbrace}} \label{eq:locc-monotonicity}\\
	&\leq 2 \sumi_i p_i\, \sqrt{ \CE{|\omega_i\rangle_B}{s} } \label{eq:subset-monotonicity}\\
	&\leq 2 L^{\sqrt{\rm CE}} (|\Psi\rangle) \label{eq:localizable}.
\end{align}
The inequality \eqref{eq:verstraete} was proved in \cite{Verstraete-LE, Popp-LE}, and the equality \eqref{eq:concurrence-ce} uses \eqref{eq:2-qubit-concur-relation}.
In the equality \eqref{eq:split-conditional} we substituted $p_{i,j} = p_{j|i}p_i$ and used the fact that the concentratable entanglement $\CE{\psi}{s}$ is invariant under tensoring with a state not supported on the subsystems labeled by $s$.
The crucial inequality \eqref{eq:locc-monotonicity} uses the fact that the concentratable entanglement and its square root are entanglement measures and hence monotonic under LOCC \cite{Beckey2021}.
Inequality \eqref{eq:subset-monotonicity} follows from the fact that $\lbrace l,m\rbrace\subseteq s$ and \cite[Thm.~1.3]{Beckey2021}.
Finally, optimizing over all local projective measurements on $B$ gives \eqref{eq:localizable}.

Since $l, m \in s$ and $\hat{a}, \hat{b} \in S^2$ were arbitrary, the left-hand side of~\eqref{eq:verstraete} maximized over $l \neq m \in s$ and $\hat{a}, \hat{b} \in S^2$ also lower bounds the right-hand side of~\eqref{eq:localizable}, giving us part of~\eqref{eq:corr-func-lb}. To obtain~\eqref{eq:corr-func-lb} in its entirety, it suffices to note that the concavity of the square root function yields $\sqrt{\mathcal{L}^{CE}} \geq \mathcal{L}^{\sqrt{CE}}$. We complete the proof by noting that Popp \textit{et al.}~\cite{Popp-LE} showed that $\max_{\hat{a}, \hat{b} \in S^2} \left|Q_{\vec{a}, \vec{b}}^{ij}(|\Psi\rangle)\right|$ for arbitrary $i,j$ is given by the maximum singular value of the matrix $\mathbf{Q}^{ij}(|\Psi\rangle)$, whose elements are $(\mathbf{Q}^{ij}(|\Psi\rangle))_{pq} = Q_{\hat{e}_p, \hat{e}_q}^{ij}(|\Psi\rangle)$.
\end{proof}

\section{Proof of Theorem~\ref{Th:LE-continuity-bound} and Lemma~\ref{lemma:3-continuity-bounds}}\label{app:continuity-results}

Let us recall that we are considering a pure state entanglement measure $E$ defined over $\mathcal{H}_B$ that satisfies the properties
\begin{align}
    |E(|\psi\rangle) - E(|\psi'\rangle)| &\leq f(\|\psi - \psi' \|_1) \label{eq:modified-lipschitz}\qquad\text{for } |\psi\rangle, |\psi'\rangle \in \mathcal{H}_B
\intertext{and}
    E(|\psi\rangle) &\leq f(\|\psi\|_1) \qquad\text{for }  |\psi\rangle \in \mathcal{H}_B \label{ineq:entanglement-bound}
\end{align}
for some concave, monotonically increasing function $f\colon \mathbb{R}_{\geq 0} \to \mathbb{R}_{\geq 0}$. For all normalized states $|\Psi\rangle \in \mathcal{H}_A\otimes \mathcal{H}_B$ and orthonormal bases $\beta = \{|i\rangle\}_{i=1}^{d_A}$ of $\mathcal{H}_A$, define 
\begin{equation}\label{eq:avg-entanglement-def}
    \overline{E}_\beta(|\Psi\rangle) \equiv \sum_{i=1}^{d_A} \langle \Psi|(|i\rangle\langle i|_A \otimes I_B)|\Psi\rangle
    ~E\left(\frac{(\langle i|_A \otimes I_B)|\Psi\rangle}{\sqrt{\langle \Psi|(|i\rangle\langle i|_A \otimes I_B)|\Psi\rangle}}\right).
\end{equation}
That is, $\overline{E}_\beta(|\Psi\rangle)$ is the average post-measurement entanglement produced by the Von Neumann measurement defined by basis $\beta$ on the $A$ systems. If the probability $\langle \Psi|(|i\rangle\langle i|_A \otimes I_B)|\Psi\rangle$ of the $i$-th outcome is zero, we interpret 
\begin{equation}
    \frac{(\langle i|_A \otimes I_B)|\Psi\rangle}{\sqrt{\langle \Psi|(|i\rangle\langle i|_A \otimes I_B)|\Psi\rangle}}
\end{equation}
to be the zero vector and take the value of $E$ at zero to be zero. Our next lemma allows us to locally bound the variations in the average post-measurement entanglement $\overline{E}_\beta$. 

\begin{lemma}\label{lemma:variation-bounds}
Let $|\Psi\rangle, |\Psi'\rangle \in \mathcal{H}_A\otimes \mathcal{H}_B$ be arbitrary normalized states and fix an orthonormal basis $\beta = \{|i\rangle\}_{i=1}^{d_A}$ of $\mathcal{H}_A$. Then
\begin{equation}\label{eq:var-dist-bound-2}
    |\overline{E}_\beta(|\Psi\rangle) - \overline{E}_\beta(|\Psi'\rangle)| \leq f(2\|\Psi - \Psi'\|_1) + \|\Psi - \Psi'\|_1.
\end{equation}
\end{lemma}
\begin{proof}
    For $1\leq i \leq d_A$, define 
    \begin{align}
        p_i &\equiv\label{eq:pi-def} \langle \Psi|(|i\rangle\langle i|_A \otimes I_B)|\Psi\rangle \\
        p_i' &\equiv\label{eq:pi'-def} \langle \Psi'|(|i\rangle\langle i|_A \otimes I_B)|\Psi'\rangle, \\
        |\varphi_i\rangle &\equiv\label{eq:psi-i-def} \frac{(\langle i|_A \otimes I_B)|\Psi\rangle}{\sqrt{p_i}}, \\
        |\varphi_i'\rangle &\equiv\label{eq:psi-i'-def} \frac{(\langle i|_A \otimes I_B)|\Psi'\rangle}{\sqrt{p_i'}}.
    \end{align}
    Observe the following:
    \begin{align}
        |\overline{E}_\beta(|\Psi\rangle) - \overline{E}_\beta(|\Psi'\rangle)| &=
        \left| \sum_i \left[p_i E(|\varphi_i\rangle) - p_i' E(|\varphi_i'\rangle)\right]\right| \\ &\leq
        \sum_i \left| p_i E(|\varphi_i\rangle) - p_i' E(|\varphi_i'\rangle) \right| \\ &= 
        \sum_i \left| p_i E(|\varphi_i\rangle) - p_i E(|\varphi_i'\rangle) + p_i E(|\varphi_i'\rangle) - p_i' E(|\varphi_i'\rangle) \right| \\
        &\leq \underbrace{\sum_i p_i |E(|\varphi_i\rangle)- E(|\varphi_i'\rangle)|}_{\rm (I)} + \underbrace{\sum_i E(|\varphi_i'\rangle) |p_i - p_i'|}_{\rm (II)}.
    \end{align}
    The quantity (II) may be bounded as follows. 
    \begin{align}
        {\rm (II)} &\equiv \sum_i E(|\varphi_i'\rangle)|p_i - p_i'|  \\
        &\leq\label{ineq:TV-dist} \sum_i |p_i - p_i'| \\
        &=\label{eq:TV-bound-write-probs} \sum_i \left|\tr[(|i\rangle\langle i|_A \otimes I_B)\Psi] - \tr[(|i\rangle\langle i|_A \otimes I_B)\Psi']\right| \\
        &=\label{eq:TV-bound-CPTP-map} \left \| \sum_i \tr[(|i\rangle\langle i|_A \otimes I_B)\Psi] |i\rangle\langle i|_A \otimes \frac{I_B}{d_B}  -  \sum_i \tr[(|i\rangle\langle i|_A \otimes I_B)\Psi'] |i\rangle\langle i|_A \otimes \frac{I_B}{d_B} \right\|_1 \\
        &\leq\label{ineq:TV-bound-monotonicity} \|\Psi - \Psi' \|_1.
    \end{align}
    Inequality~\eqref{ineq:TV-dist} follows from the fact that the values of $E$ lie in $[0,1]$. Equation~\eqref{eq:TV-bound-CPTP-map} follows from the fact that the trace norm of an operator is equal to the sum of its singular values. Inequality~\eqref{ineq:TV-bound-monotonicity} may be seen by noting that the expression in line~\eqref{eq:TV-bound-CPTP-map} is the trace distance between the states produced by the CPTP map from $\mathcal{D}(\mathcal{H}_A \otimes \mathcal{H}_B)$ to $\mathcal{D}(\mathcal{H}_A \otimes \mathcal{H}_B)$ defined by
    \begin{equation}
        \rho_{AB} \mapsto \sum_i \tr[(|i\rangle\langle i|_A \otimes I_B)\rho_{AB}]|i\rangle\langle i|_A \otimes \frac{I_B}{d_B}
    \end{equation}
    acting on $\Psi$ and $\Psi'$ and then using the monotonicity of the trace distance under CPTP maps.     

    Let us now bound term (I). We have
    \begin{align}
        {\rm (I)} 
        &\equiv \sum_i p_i |E(|\varphi_i\rangle) - E(|\varphi_i'\rangle)| \\
        &\leq\label{ineq:Lip-bound-conds-12} \sum_i p_i f(\|\varphi_i - \varphi_i'\|_1) \\
        &\leq\label{ineq:Lip-bound-convave} f\left(\sum_i p_i \|\varphi_i - \varphi_i'\|_1\right) \\
        &=\label{eq:Lip-bound-bring-in-trace-dist} f\left( \left \|\sum_i p_i |i\rangle\langle i|_A \otimes (\varphi_i - \varphi_i') \right \|_1 \right) \\
        &= f\left(\left \|\sum_i |i\rangle\langle i|_A \otimes (p_i\varphi_i - p_i'\varphi_i' + p_i'\varphi_i' - p_i \varphi_i') \right \|_1\right) \\
        &\leq\label{ineq:triangle-f-mon} f\left(\left \|\sum_i |i\rangle\langle i|_A \otimes (p_i \varphi_i - p_i'\varphi_i')\right\|_1 +
        \left \|\sum_i (p_i' - p_i)|i\rangle\langle i|_A \otimes \varphi_i' \right \|_1\right) \\
        &=\label{eq:Lip-bound-use-psi-i-def} f\Bigg(\left \|\sum_i \left[(|i\rangle\langle i|_A \otimes I_B)\Psi(|i\rangle\langle i|_A \otimes I_B) -  (|i\rangle\langle i|_A \otimes I_B)\Psi'(|i\rangle\langle i|_A \otimes I_B) \right] \right \|_1 \\
        &\ \ \ \ \ \ \ \ \ \ \ \ + \sum_i |p_i' - p_i| \Bigg) \notag \\
        &\leq\label{ineq:Lip-bound-monotonicity} f\left(\|\Psi - \Psi'\|_1 + \|\Psi - \Psi'\|_1 \right)\\
        &= f(2\|\Psi - \Psi'\|_1)
    \end{align}

Inequality~\eqref{ineq:Lip-bound-conds-12} follows from our conditions imposed on $E$ in~\eqref{eq:modified-lipschitz} and~\eqref{ineq:entanglement-bound}. Note that~\eqref{ineq:entanglement-bound} plays an important role in this step because it ensures that in the case that either $|\varphi_i\rangle$ or $|\varphi_i'\rangle$ is the zero vector,~\eqref{ineq:Lip-bound-conds-12} still holds. Inequality~\eqref{ineq:Lip-bound-convave} follows from the concavity of $f$. Equation~\eqref{eq:Lip-bound-bring-in-trace-dist} may be seen from the fact that the trace norm of an operator is equal to the sum of its singular values. Inequality~\eqref{ineq:triangle-f-mon} is due to the triangle inequality and the fact that $f$ is monotonically increasing. Equality~\eqref{eq:Lip-bound-use-psi-i-def} follows from using the definitions~\eqref{eq:pi-def}-\eqref{eq:psi-i'-def} to expand the first term and again recalling that the trace norm of an operator is equal to the sum of its singular values to simplify the second term. Inequality~\eqref{ineq:Lip-bound-monotonicity} is obtained from a combination of using the monotonicity of the trace distance under CPTP maps to bound the first term and recalling our work from~\eqref{ineq:TV-dist}-\eqref{ineq:TV-bound-monotonicity} to bound the second term, together with the fact that $f$ is monotonically increasing.

The statement of the lemma now follows from our bounds on (I) and (II).
\end{proof}

The continuity bounds of Theorem~\ref{Th:LE-continuity-bound} then follow simply from our continuity bound on the average entanglement: 
\begin{proof}[Proof of Theorem~\ref{Th:LE-continuity-bound}] 
    Fix two quantum states $|\Psi\rangle, |\Psi'\rangle \in \mathcal{H}_A\otimes\mathcal{H}_B$, and assume without loss of generality that $L^E(|\Psi'\rangle) \leq L^E(|\Psi\rangle)$.
    Let furthermore $\beta$ be an orthonormal local measurement basis of $\mathcal{H}_A$ such that $L^E(|\Psi\rangle) = \overline{E}_\beta(|\Psi\rangle)$.
    Then by definition, $L^E(|\Psi'\rangle)\geq \overline{E}_\beta(|\Psi'\rangle)$, and it follows that 
    \begin{equation}
        |L^E(|\Psi\rangle) - L^E(|\Psi'\rangle)| 
        = L^E(|\Psi\rangle) - L^E(|\Psi'\rangle)
        \leq \overline{E}_\beta(|\Psi\rangle) - \overline{E}_\beta(|\Psi'\rangle).
    \end{equation}
    
    The conclusion of the theorem for the localizable entanglement $L^E$ then follows from Lemma~\ref{lemma:variation-bounds}. The continuity bound for the multipartite entanglement of assistance $\mathcal{A}^E$ follows from a nearly identical argument.
\end{proof}

We now prove the continuity bounds on the $N_B$-tangle, the GME-concurrence, and the concentratable entanglement of Lemma~\ref{lemma:3-continuity-bounds} one by one. For any vector $|v\rangle$ in a Hilbert space, let $\||v\rangle \|_2 \coloneqq \sqrt{\langle v|v\rangle}$, and for any linear operator $L$ on a Hilbert space let $\|L\|_2 \coloneqq \sqrt{\tr(L^\dagger L)}$ be the Hilbert-Schmidt norm of $L$. 

\begin{proof}[Proof of~\eqref{ineq:n-tangle-continuity-bound}]
    Fix states $|\psi\rangle, |\psi'\rangle \in \mathcal{H}_B$. The bound~\eqref{ineq:n-tangle-continuity-bound} follows from the following chain of inequalities, in which we may assume without loss of generality that $\langle \psi|\psi'\rangle = |\langle \psi|\psi'\rangle|$ because $\tau_{N_B}(|\psi\rangle) = \tau_{N_B}(e^{i\theta}|\psi\rangle)$ for all real $\theta$:
    \begin{align}
        |\tau_{N_B}(|\psi\rangle) - \tau_{N_B}(|\psi'\rangle)| &= \left||\langle \psi|\tilde{\psi}\rangle| - |\langle \psi'|\tilde{\psi'}\rangle|\right| \\
        &\leq |\langle \psi|\tilde{\psi}\rangle - \langle \psi'|\tilde{\psi}'\rangle | \\
        &= |\langle \psi|\tilde{\psi}\rangle - \langle \psi'|\tilde{\psi}\rangle + \langle \psi'|\tilde{\psi}\rangle - \langle \psi'|\tilde{\psi}'\rangle| \\
        &\leq |(\langle \psi| - \langle \psi'|)|\tilde{\psi}\rangle| + |\langle \psi'|(|\tilde{\psi}\rangle - |\tilde{\psi}'\rangle)| \\
        &\leq\label{ineq:CS} \| |\psi\rangle - |\psi'\rangle \|_2 + \||\tilde{\psi}\rangle - |\tilde{\psi}'\rangle \|_2 \\
        &= 2\| |\psi\rangle - |\psi'\rangle \|_2 \\
        &= 2 \sqrt{2 - 2|\langle \psi|\psi'\rangle|} \\
        &\leq 2\sqrt{2 - 2|\langle \psi|\psi'\rangle|^2} \\
        &= \sqrt{2}\|\psi - \psi'\|_1.
    \end{align}
    Inequality~\eqref{ineq:CS} follows from the Cauchy-Schwarz inequality. The last equality is the Fuchs-van-de-Graaf inequality for fidelity and trace distances, which holds as an equality for the pure states $\psi$ and $\psi'$ \cite[eq.(9.173)]{Wilde2016}.
\end{proof}

The following proof of our continuity bound on the GME-concurrence modifies the proofs of the Lipschitz continuity of the concentratable entanglement given in~\cite{schatzki2021entangled, Beckey2021}.
\begin{proof}[Proof of~\eqref{ineq:gme-concurrence-continuity-bound}]
    Fix $|\psi\rangle, |\psi'\rangle \in \mathcal{H}_B$. Observe that for any $\varnothing \subsetneq \gamma \subsetneq [N]$, we have
    \begin{align}
        \left|\sqrt{2(1 - \tr[\psi_\gamma^2])} - \sqrt{2(1 - \tr[(\psi')_\gamma^2])}\right| &\leq 
        \sqrt{\left|2(1 - \tr[\psi_\gamma^2]) - 2(1 - \tr[(\psi')_\gamma^2])\right|} \label{eq:first-gme-cont-proof-step}\\ &=
        \sqrt{2\left|\tr[(\psi')_\gamma^2] - \tr[\psi_\gamma^2]\right|} \label{eq:purity-deviation} \\ &= 
        \sqrt{2\left|\tr[(\psi'_\gamma+ \psi_\gamma)(\psi'_\gamma - \psi_\gamma)]\right|} \\ &\leq\label{ineq:CS-CGME}
        \sqrt{2\| \psi'_\gamma + \psi_\gamma \|_2 \|\psi'_\gamma - \psi_\gamma\|_2} \\ &\leq 
        \sqrt{2(\| \psi'_\gamma \|_2 + \|\psi_\gamma \|_2)\|\psi'_\gamma - \psi_\gamma\|_2} \\ &\leq 
        2\sqrt{\|\psi'_\gamma - \psi_\gamma\|_2}\\ &\leq\label{ineq:one-norm-two-norm-rel} 
        2^{3/4}\sqrt{\|\psi'_\gamma - \psi_\gamma\|_1} \\ &\leq\label{ineq:monotonicity-cgme}
        2^{3/4}\sqrt{\|\psi' - \psi\|_1}
    \end{align}

    Inequality~\eqref{ineq:CS-CGME} follows from the Cauchy-Schwarz inequality applied to the Hilbert-Schmidt inner product. Inequality~\eqref{ineq:one-norm-two-norm-rel} follows from a slight tightening of the usual relation between the unnormalized trace distance and Hilbert-Schmidt distance for mixed states, $\|\rho - \sigma\|_2 \leq \frac{1}{\sqrt{2}} \|\rho - \sigma\|_1$~\cite[eq.(6)]{Coles2019}. Finally~\eqref{ineq:monotonicity-cgme} follows from the monotonicity of the trace distance under CPTP maps. 

    Assume now without loss of generality that $C_{\gme}(|\psi\rangle)\geq C_{\gme}(|\psi'\rangle)$, and let $\varnothing \subsetneq \gamma \subsetneq [N]$ be such that $C_{\gme}(|\psi'\rangle) = \sqrt{2(1 - \tr [(\psi')_\gamma^2])}$.
    Then $C_{\gme}(|\psi\rangle) \leq \sqrt{2(1 - \tr [\psi_\gamma^2])}$ by definition, and hence
    \begin{align}
        |C_{\gme}(|\psi\rangle) - C_{\gme}(|\psi'\rangle)| 
        &= C_{\gme}(|\psi\rangle) - C_{\gme}(|\psi'\rangle)\\  
        &\leq \sqrt{2(1 - \tr[\psi_\gamma^2])} - \sqrt{2(1 - \tr [(\psi')_\gamma^2])} \\
        &\leq 2^{3/4} \sqrt{\|\psi - \psi'\|_1}.
    \end{align}
    The last inequality follows from arguments above.
\end{proof}

Finally, for the proof of~\eqref{ineq:ce-continuity-bound}, see~\cite{schatzki2021entangled}. Furthermore,~\cite{Beckey2021} also gives a similar continuity bound for the concentratable entanglement, although they show that it is $2$-Lipschitz whereas~\cite{schatzki2021entangled} shows that it is $\sqrt{2}$-Lipschitz.

\section{Concentration results}\label{app:concentration-results}

A key result we will use throughout this appendix is the well-known Levy's Lemma. For the following discussion, we define $\mathbb{S}^k \coloneqq \{(x_1, \dots, x_{k+1}) \in \mathbb{R}^{k+1}: x_1^2 + \dots + x_{k+1}^2 = 1\}$ for all integers $k \geq 0$. 
\begin{lemma}[Levy's Lemma]\label{lemma:levy}
    Let $f\colon \mathbb{S}^{2d-1} \to \mathbb{R}$ be a function such that $|f(|\Phi\rangle) - f(|\Phi'\rangle)| \leq K\||\Phi\rangle - |\Phi'\rangle\|_2$ for all $|\Phi\rangle, |\Phi'\rangle \in \mathbb{S}^{2d - 1}$. Then for any $\varepsilon > 0$, we have
    \begin{equation}
        \Pr\left( \left| f(|\Phi\rangle) - \EV_{|\Phi'\rangle \sim \Haar(d)}[f(|\Phi'\rangle)] \right| \geq \varepsilon \right) \leq 2\exp\left( - \frac{2d \varepsilon^2}{9\pi^3 K^2}\right)
    \end{equation}
\end{lemma}

First, we will lower bound the expectation value of $\mathcal{A}^\tau(|\Psi\rangle)$ over Haar-random states $|\Psi\rangle \in \mathcal{H}_A \otimes \mathcal{H}_B$. 
\begin{lemma}\label{lemma:EoA-tau-LB}
    Assume $N_B$ is even. We have 
    \begin{equation}\label{ineq:haar-avg-EoA-LB}
        \EV_{|\Psi\rangle\sim \Haar(d)} \left[\mathcal{A}^\tau(|\Psi\rangle)\right] \geq 1- \sqrt{\frac{d_B}{2d_A}}. 
    \end{equation}
\end{lemma}
\begin{proof}
    First, note that Theorem~\ref{Th:n-tangle-upper-bound} implies that for all $|\Psi\rangle \in \mathcal{H}_A \otimes \mathcal{H}_B$, we have
    \begin{equation}\label{ineq:EoA-LB}
        \mathcal{A}^\tau(|\Psi\rangle) = F(\Psi_B, \tilde{\Psi}_B) \geq 1 - \frac{1}{2} \| \Psi_B - \tilde{\Psi}_B\|_1 \geq 
        1 - \frac{\sqrt{d_B}}{2} \| \Psi_B - \tilde{\Psi}_B \|_2, 
    \end{equation}
    where the first inequality is the Fuchs-van de Graaf inequality. Now observe that 
    \begin{align}
        \EV_{|\Psi\rangle \sim \Haar(d)} \left[ \| \Psi_B - \tilde{\Psi}_B \|_2 \right] &= 
        \EV_{|\Psi\rangle \sim \Haar(d)} 
        \left[ \sqrt{2\tr(\Psi_B^2) - 2\tr(\Psi_B \tilde{\Psi}_B)} \right] \\ 
        &\leq\label{ineq:avg-spin-flip-dist-bound}      \sqrt{2\left(\EV_{|\Psi\rangle\sim\Haar(d)}\left[\tr(\Psi_B^2)\right] - \EV_{|\Psi\rangle\sim\Haar(d)}\left[\tr(\Psi_B\tilde{\Psi}_B) \right]\right)}.
    \end{align}
    The first term in~\eqref{ineq:avg-spin-flip-dist-bound} is simply the expectation value of the purity of $\Psi_B$, which equals (see, e.g.,~\cite{Mele2024})
    \begin{equation}
        \EV_{|\Psi\rangle\sim\Haar(d)}\left[\tr(\Psi_B^2)\right] =
        \frac{d_A + d_B}{d_Ad_B + 1}.
    \end{equation}

    Consider the Hilbert space $\mathcal{H}_B \otimes \mathcal{H}_B$ and let $\mathbb{F}_{B_1, B_2}$ denote the swap operator between the two copies of $\mathcal{H}_B$. Let also $\Phi_{B_1B_2}^+$ denote the \textit{maximally entangled state} across $\mathcal{H}_B \otimes \mathcal{H}_B$, or 
    \begin{equation}\label{eq:maximally-entangled-def}
        \Phi_{B_1 B_2}^+ \coloneqq  \frac{1}{d_B} \sum_{i,j=1}^{d_B} |ii\rangle\langle jj|,
    \end{equation}
    where $\{|i\rangle \}_{i=1}^{d_A}$ is an orthonormal basis of $\mathcal{H}_B$. 
    Let also $d=d_Ad_B$.
    Then the second term of~\eqref{ineq:avg-spin-flip-dist-bound} is given by
    \begin{align}                   
        &\EV_{|\Psi\rangle\sim\Haar(d)}\left[\tr(\Psi_B\tilde{\Psi}_B) \right]\notag\\
        &= \EV_{|\Psi\rangle\sim\Haar(d)}\left[\tr\left(\mathbb{F}_{B_1, B_2} \left[\Psi_{B_1} \otimes \tilde{\Psi}_{B_2}\right]\right)\right] \\
        &= \EV_{|\Psi\rangle\sim\Haar(d)}\left[\tr\left(\mathbb{F}_{B_1, B_2} 
        (I_{B_1} \otimes \sigma_y^{\otimes N_B})\tr_{A_1A_2}\left[\Psi \otimes \Psi^\ast \right](I_{B_1} \otimes \sigma_y^{\otimes N_B})
        \right)\right] \\
        &= \tr\left[\mathbb{F}_{B_1, B_2} 
        (I_{B_1} \otimes \sigma_y^{\otimes N_B})\tr_{A_1A_2}\left(\EV_{|\Psi\rangle\sim\Haar(d)} \left[\Psi \otimes \Psi^\ast\right] \right) (I_{B_1} \otimes \sigma_y^{\otimes N_B})
        \right] \\
        &=\label{eq:haar-conjugate-avg} \tr\Bigg[\mathbb{F}_{B_1, B_2} 
        (I_{B_1} \otimes \sigma_y^{\otimes N_B})\tr_{A_1 A_2}\left(\frac{I_{A_1A_2}\otimes I_{B_1B_2} + d \Phi_{A_1 A_2}^+ \otimes \Phi_{B_1B_2}^+}{d(d+1)} \right) \notag\\
        &\qquad\qquad\times (I_{B_1} \otimes \sigma_y^{\otimes N_B})
        \Bigg]\\
        &= \tr\left[\mathbb{F}_{B_1, B_2} 
        (I_{B_1} \otimes \sigma_y^{\otimes N_B}) \left(\frac{d_A^2 I_{B_1 B_2} + d \Phi_{B_1 B_2}^+ }{d(d+1)} \right) 
        (I_{B_1} \otimes \sigma_y^{\otimes N_B}) \right] \\
        &= \frac{d_A^2 \tr[\mathbb{F}_{B_1, B_2}] + d \tr[\mathbb{F}_{B_1,B_2} (I_{B_1} \otimes \sigma_y^{\otimes N_B}) \Phi_{B_1 B_2}^+ (I_{B_1} \otimes \sigma_y^{\otimes N_B})]}{d(d+1)} \\
        &=\label{eq:maximally-entangled-comp} \frac{d_A^2 d_B + d}{d(d+1)} \\
        &= \frac{d_A + 1}{d_Ad_B + 1}.
    \end{align}
    Equation~\eqref{eq:haar-conjugate-avg} follows from standard techniques involved in the computation of Haar averages (see, e.g.,~\cite[Ex.~49]{Mele2024}). 
    Equality~\eqref{eq:maximally-entangled-comp} follows from a direct computation using~\eqref{eq:maximally-entangled-def} and the fact that $N_B$ is even. Referring back to~\eqref{ineq:avg-spin-flip-dist-bound}, it then follows that 
    \begin{equation}
        \EV_{|\Psi\rangle\sim\Haar(d)}\left[ \| \Psi_B - \tilde{\Psi}_B \|_2 \right] 
        \leq \sqrt{\frac{2d_B - 2}{d_Ad_B + 1}}.
    \end{equation}
    Using~\eqref{ineq:EoA-LB}, we conclude that
    \begin{equation}
        \EV_{|\Psi\rangle \sim \Haar(d)}\left[{\mathcal{A}^\tau(|\Psi\rangle)}\right] \geq 1 - \sqrt{\frac{d_B^2 - d_B}{2(d_Ad_B + 1)}}  \geq 1 - \sqrt{\frac{d_B}{2d_A}},
    \end{equation}   
	which proves the claim.
\end{proof}

As before, we define for any quantum state $|\Psi\rangle \in \mathcal{H}_A \otimes \mathcal{H}_B$ and normalized vector $|v\rangle \in \mathcal{H}_A$ the following:
\begin{align}
    p_v(|\Psi\rangle) &\coloneqq \langle \Psi|(|v\rangle\langle v|_A \otimes I_B)|\Psi\rangle \\
    |\varphi_v(\Psi)\rangle &\coloneqq \frac{(\langle v|_A \otimes I_B)|\Psi\rangle}{\sqrt{p_v}}.
\end{align}
For any function $E\colon \mathcal{H}_B \to \mathbb{R}$ and arbitrary orthonormal basis $\beta \subset \mathcal{H}_A$, let $\overline{E}_\beta$ denote the average value of $E$ across the post-measurement ensemble produced by measuring subsystem $A$ with respect to $\beta$:
\begin{equation}
    \overline{E}_\beta(|\Psi\rangle) = \sum_{|v\rangle \in \beta} p_v(|\Psi\rangle) E(|\varphi_v(\Psi)\rangle).
\end{equation}

In our proofs we use the following result which is implicitly proved in Lemma 4 from~\cite{Cotler2023}:

\begin{lemma}\label{lemma:independence}
    Let $|v\rangle \in \mathcal{H}_A$ be a normalized vector. If $|\Psi\rangle \sim \Haar(d)$, then for any functions $F\colon [0,1] \to \mathbb{R}$ and $G\colon \mathcal{H}_B \to \mathbb{R}$ for which the relevant expressions are defined, we have
    \begin{equation}        \EV_{|\Psi\rangle\sim\Haar(d)}\left[F(p_v(|\Psi\rangle)) G(|\varphi_v(\Psi)\rangle)\right] = \EV_{|\Psi\rangle\sim\Haar(d)}\left[F(p_v(|\Psi\rangle))\right]\EV_{|\varphi \rangle\sim\Haar(d_B)}\left[G(|\varphi\rangle)\right]
    \end{equation}
\end{lemma}

Lemma~\ref{lemma:independence} allows us to conveniently compute expectation values of $\overline{E}_\beta$ for an arbitrary orthonormal basis $\beta$ and function $E\colon \mathcal{H}_B \to \mathbb{R}$.
\begin{lemma}\label{lemma:avg-E-calc}
    Fix an orthonormal basis $\beta = \{|i\rangle\}_{i=1}^{d_A}$ of $\mathcal{H}_A$. Then for any function $E\colon \mathcal{H}_B \to \mathbb{R}$, we have
    \begin{equation}
        \EV_{|\Psi\rangle\sim\Haar(d)}\left[ \overline{E}_\beta(|\Psi\rangle)\right] = 
        \EV_{|\varphi\rangle \sim \Haar(d_B)}\left[E(|\varphi\rangle)\right]
    \end{equation}
\end{lemma}
\begin{proof}
    Observe that 
    \begin{align}
        \EV_{|\Psi\rangle \sim \Haar(d)} \left[\overline{E}_\beta(|\Psi\rangle)\right] 
        &=         
        \EV_{|\Psi\rangle \sim \Haar(d)} \left[\sum_{i=1}^{d_A} p_i(|\Psi\rangle) E(|\varphi_i(\Psi)\rangle) \right] \\
        &= \sum_{i=1}^{d_A}        \EV_{|\Psi\rangle\sim\Haar(d)}\left[p_i(|\Psi\rangle) E(|\varphi_i(\Psi)\rangle) \right] \\
        &= \sum_{i=1}^{d_A} \EV_{|\Psi\rangle\sim\Haar(d)}\left[p_i(|\Psi\rangle)\right] \EV_{|\varphi \rangle\sim\Haar(d_B)} \left[ E(|\varphi\rangle) \right] \\
        &= \sum_{i=1}^{d_A} \frac{1}{d_A} \EV_{|\varphi \rangle\sim\Haar(d_B)} \left[ E(|\varphi\rangle) \right]\\
        &= \EV_{|\varphi \rangle\sim\Haar(d_B)} \left[ E(|\varphi\rangle) \right],
    \end{align}
    where we used Lemma~\ref{lemma:independence} in the third equality.
\end{proof}
We now use Lemma~\ref{lemma:avg-E-calc} to bound or compute the average post-measurement values of the $N_B$-tangle, (squared) GME-concurrence, and concentratable entanglement for a fixed, non-necessarily-optimal measurement basis. We start by bounding the Haar average of the $N_B$-tangle across $\mathcal{H}_B$.

\begin{lemma}\label{lemma:avg-ntangle-bound}
    The following upper bound holds:
    \begin{equation}
        \EV_{|\psi\rangle \sim \Haar(d_B)}[\tau_{N_B}(|\psi\rangle)] \leq \sqrt{\frac{2}{d_B + 1}}. 
    \end{equation}
\end{lemma}
\begin{proof}
    The statement is trivially true when $N_B$ is odd, since the $N_B$-tangle is generically zero in that case.
    Assume therefore that $N_B$ is even. We have
    \begin{align}
        \EV_{|\psi\rangle \sim \Haar(d_B)}[\tau_{N_B}(|\psi\rangle)] &=
        \EV_{|\psi\rangle \sim \Haar(d_B)}[|\langle \psi|\sigma_y^{\otimes N_B}|\psi^\ast\rangle|] \\
        &= \EV_{|\psi\rangle \sim \Haar(d_B)}
        \left[\sqrt{\langle \psi|\sigma_y^{\otimes N_B}|\psi^\ast\rangle
        \langle \psi^\ast|\sigma_y^{\otimes N_B}|\psi\rangle}\right] \\
        &\leq\label{ineq:haar-avg-ntangle-concavity} \sqrt{\EV_{|\psi\rangle\sim \Haar(d_B)} 
        \left[\langle \psi|\sigma_y^{\otimes N_B}|\psi^\ast\rangle
        \langle \psi^\ast|\sigma_y^{\otimes N_B}|\psi\rangle\right]} \\
        &=\label{eq:swap-trick-ntangle} \sqrt{\EV_{|\psi\rangle\sim \Haar(d_B)}  
        \left[\tr\left(\mathbb{F}_{B_1, B_2}\left[\sigma_y^{\otimes N_B}|\psi\rangle\langle \psi| \otimes \sigma_y^{\otimes N_B}|\psi^\ast\rangle\langle \psi^\ast|\right]\right)\right]} \\
        &= \sqrt{\tr\left[\mathbb{F}_{B_1, B_2}\left((\sigma_y^{\otimes N_B} \otimes \sigma_y^{\otimes N_B}) \EV_{|\psi\rangle\sim \Haar(d_B)}\left[ |\psi\rangle\langle\psi|\otimes |\psi^\ast\rangle\langle \psi^\ast| \right] \right) \right]} \\
        &=\label{eq:isotropic-state-ntangle} \sqrt{\tr\left(\mathbb{F}_{B_1, B_2} \left[(\sigma_y^{\otimes N_B} \otimes \sigma_y^{\otimes N_B})\left(\frac{I_{B_1}\otimes I_{B_2} + d_B \Phi_{B_1 B_2}^+ }{d_B(d_B+1)} \right) \right] \right)} \\
        &= \sqrt{\frac{\tr \left[\mathbb{F}_{B_1, B_2} (\sigma_y^{\otimes N_B} \otimes \sigma_y^{\otimes N_B} )\right] + d_B\tr\left[ \mathbb{F}_{B_1B_2} (\sigma_y^{\otimes N_B} \otimes \sigma_y^{\otimes N_B}) \Phi_{B_1 B_2}^+ \right]}{d_B(d_B+1)}} \\
        &=\label{eq:ntangle-simplify} \sqrt{\frac{\tr[(\sigma_y^2)^{\otimes N_B}] + d_B\tr[(\sigma_y^{\otimes N_B} \otimes \sigma_y^{\otimes N_B})\Phi_{B_1B_2}^+]}{d_B(d_B + 1)}} \\
        &=\label{eq:ntangle-final-simplify} \sqrt{\frac{d_B + d_B}{d_B(d_B + 1)}} \\
        &= \sqrt{\frac{2}{d_B + 1}}.
    \end{align}
Inequality~\eqref{ineq:haar-avg-ntangle-concavity} follows from the concavity of the square root function. Equation~\eqref{eq:swap-trick-ntangle} is an application of the swap trick. Equality~\eqref{eq:isotropic-state-ntangle} follows from standard techniques used in computations involving Haar integration, see e.g.~\cite[Ex.~49]{Mele2024}. Equation~\eqref{eq:ntangle-simplify} follows from applying the swap trick to the first term and simplifying the second term via the fact that $[\mathbb{F}_{B_1, B_2}, U_B \otimes U_B]=0$ for any unitary $U_B$ over $\mathcal{H}_B$. Finally,~\eqref{eq:ntangle-final-simplify} is obtained by noting that $\sigma_y^2 = I$ and by explicitly computing the second term using~\eqref{eq:maximally-entangled-def} and the fact that $N_B$ is even. 
\end{proof}

Lemma~\ref{lemma:avg-E-calc} in conjunction with Lemma~\ref{lemma:avg-ntangle-bound} allows us to easily compute the following bound on the expectation value of the average post-measurement $N_B$-tangle for Haar-random states. 
\begin{lemma}\label{lemma:avg-avg-ntangle}
    Fix an arbitrary orthonormal basis $\beta = \{|i\rangle\}_{i=1}^{d_A}$ of $\mathcal{H}_A$. Then
    \begin{equation}
        \EV_{|\Psi\rangle \sim \Haar(d)} \left[\overline{\tau}_\beta(|\Psi\rangle)\right] \leq \sqrt{\frac{2}{d_B + 1}}.
    \end{equation}
\end{lemma}

We now turn to the GME-concurrence and bound the Haar average of its squared concurrence across $\mathcal{H}_B$. 
\begin{lemma}\label{lemma:gme-concurrence-mean-val}
    The following lower bound holds:
    \begin{equation}
        \EV_{|\psi\rangle \sim \Haar(d_B)}\left[C^2_{\rm GME}(|\psi\rangle)\right] \geq 1- 2\min_{\varnothing \subsetneq \gamma \subsetneq B} \frac{d_\gamma + d_{B\setminus \gamma}}{d_B + 1} - 2d_B^{-1/4} - 4d_B\exp\left(-\frac{\sqrt{d_B}}{72\pi^3} \right) .
    \end{equation}
\end{lemma}
\begin{proof}
    By referring to the proof of~\eqref{ineq:gme-concurrence-continuity-bound}, particularly from~\eqref{eq:purity-deviation} onward, we see that for any subsystem $\varnothing\subsetneq\gamma \subsetneq B$ and any $|\psi\rangle, |\varphi\rangle \in \mathcal{H}_B$, we have 
    \begin{equation}
        \left|\tr\left[\psi_\gamma^2\right] - \tr\left[\varphi_\gamma^2\right]\right|\leq 2 \|\psi - \varphi\|_1 \leq 4 \||\psi\rangle-|\varphi\rangle\|_2,
    \end{equation}
    where we used the norm inequality $\||\psi\rangle\langle\psi|-|\varphi\rangle\langle\varphi|\|_1\leq 2\| |\psi\rangle - |\varphi\rangle\|_2$ in the last inequality. This norm inequality may be derived using the fact that $1-a^2 \leq 2 - 2a$ for all $a\in (0,1]$, so that for any states $|\psi\rangle, |\varphi\rangle \in \mathcal{H}_B$ we have
    \begin{align}
        \| |\psi\rangle \langle \psi| - |\varphi\rangle\langle \varphi|\|_1 
        &= 2\sqrt{1 - |\langle \psi|\varphi\rangle|^2} \\
        &\leq 2\sqrt{2 - 2|\langle \psi|\varphi\rangle|} \\
        &\leq \sqrt{2 - 2\mathrm{Re}\langle \psi|\varphi\rangle} \\
        &= 2\||\psi\rangle - |\varphi\rangle\|_2.
    \end{align}
    Using Levy's Lemma~\ref{lemma:levy} and the fact that $\EV\left[\tr\left[\psi_\gamma^2\right]\right] = (d_\gamma + d_{B\setminus \gamma})/(d_B + 1)$~\cite{Mele2024}, we arrive at
    \begin{equation}
        \Pr_{|\psi\rangle \sim \Haar(d_B)} \left( \left| \tr\left[\psi_\gamma^2\right] - \frac{d_\gamma + d_{B\setminus\gamma}}{d_B + 1} \right| \geq d_B^{-1/4} \right) \leq 2\exp\left( - \frac{\sqrt{d_B}}{72\pi^3} \right).
    \end{equation}
    Applying a union bound, this becomes
    \begin{align}
        \Pr_{|\psi\rangle\sim \Haar(d_B)} &\left(\max_{\varnothing\subsetneq \gamma \subsetneq B} \tr\left[\Psi_\gamma^2\right] \geq  \max_{\varnothing \subsetneq \gamma' \subsetneq B} \frac{d_{\gamma'} + d_{B\setminus\gamma'}}{d_B + 1} +  d_B^{-1/4} \right)\\
        &\leq \sum_{\varnothing\subsetneq\gamma\subsetneq B} 
        \Pr_{|\psi\rangle \sim \Haar(d_B)} \left( \tr\left[\psi_\gamma^2\right] \geq  \max_{\varnothing \subsetneq \gamma' \subsetneq B} \frac{d_{\gamma'} + d_{B\setminus\gamma'}}{d_B + 1}  + d_B^{-1/4} \right)  \\
        &\leq \sum_{\varnothing \subsetneq \gamma \subsetneq B} 
        \Pr_{|\psi\rangle \sim \Haar(d_B)} \left( \tr\left[\psi_\gamma^2\right] \geq  \frac{d_{\gamma} + d_{B\setminus\gamma}}{d_B + 1}  + d_B^{-1/4} \right) \\
        &\leq \sum_{\varnothing  \subsetneq \gamma\subsetneq B} \Pr_{|\psi\rangle \sim \Haar(d_B)} \left( \left| \tr\left[\psi_\gamma^2\right] - \frac{d_\gamma + d_{B\setminus\gamma}}{d_B + 1} \right| \geq d_B^{-1/4} \right)  \\
        &\leq 2d_B \exp\left( - \frac{\sqrt{d_B}}{72\pi^3} \right). 
    \end{align}
    Naturally, $\max_\gamma \tr[\psi_\gamma^2]$ may lie either above or below $\max_\gamma (d_\gamma + d_{B\setminus \gamma})/(d_B + 1) + d_B^{-1/4}$. Using this idea and the fact that both $\max_\gamma \tr[\psi_\gamma^2]$ and the probability that $\max_\gamma \tr\left[\psi_\gamma^2\right]$ lies below this threshold value are upper bounded by one, we get
    \begin{align}
        \EV_{|\psi\rangle \sim \Haar(d_B)} &\left[\max_{\varnothing \subsetneq \gamma \subsetneq B}\left[\tr\left[\psi_\gamma^2\right] \right] \right] \notag\\
        &\leq \Pr_{|\psi \rangle \sim \Haar(d_B)}\left( \max_{\varnothing \subsetneq \gamma \subsetneq B} \left[\tr\left[\psi_\gamma^2 \right]\right] \geq \max_{\varnothing \subsetneq \gamma' \subsetneq B} \frac{d_{\gamma'} + d_{B\setminus\gamma'}}{d_B + 1}  + d_B^{-1/4} \right) \notag\\
        &\qquad {}+ \max_{\varnothing \subsetneq \gamma' \subsetneq B} \frac{d_{\gamma'} + d_{B\setminus\gamma'}}{d_B + 1}  + d_B^{-1/4} \\
        &\leq \max_{\varnothing \subsetneq \gamma\subsetneq B} \frac{d_{\gamma} + d_{B\setminus\gamma}}{d_B + 1}  + d_B^{-1/4} + 2d_B\exp\left(-\frac{\sqrt{d_B}}{72\pi^3} \right) 
    \end{align}
    Finally, we compute the expectation value of the squared GME-concurrence to be
    \begin{align}
        \EV_{|\psi\rangle \sim \Haar(d_B)} \left[C^2_{\rm GME}(|\psi\rangle)\right] &= \EV_{|\psi\rangle \sim \Haar(d_B)}\left[\min_{\varnothing \subsetneq\gamma\subsetneq B} 2\left(1- \tr\left[\psi_\gamma^2\right]\right)\right] \\
        &= 2-2\EV_{|\psi \rangle \sim \Haar(d_B)}\left[\max_{\varnothing\subsetneq \gamma\subsetneq B} \tr\left[\psi_\gamma^2\right]\right] \\
        &\geq 2 - 2\max_{\varnothing \subsetneq \gamma\subsetneq B} \frac{d_{\gamma} + d_{B\setminus\gamma}}{d_B + 1}  - 2d_B^{-1/4} - 4d_B\exp\left(-\frac{\sqrt{d_B}}{72\pi^3} \right),
    \end{align}
    which concludes the proof.
\end{proof}

For arbitrary measurement basis $\beta\subset \mathcal{H}_A$, let us define the quantity $\overline{C^2_{\rm GME, \beta}}(|\Psi\rangle)$ to be the average value of the squared GME-concurrence over the ensemble produced by a measurement on $A$ with respect to $\beta$ of the system in state $|\Psi\rangle$. The expectation value of $\overline{C^2_{\rm GME, \beta}}(|\Psi\rangle)$ averaged over Haar-random $|\Psi\rangle$ is then lower bounded by the Lemma below, which follows directly from Lemmas~\ref{lemma:avg-E-calc} and~\ref{lemma:gme-concurrence-mean-val}. 
\begin{lemma}
    The following holds:
    \begin{equation}
        \EV_{|\Psi\rangle \sim \Haar(d)} \left[\overline{C^2_{\rm GME, \beta}}(|\Psi\rangle)\right] \geq 2 - 2\max_{\varnothing \subsetneq \gamma\subsetneq B} \frac{d_{\gamma} + d_{B\setminus\gamma}}{d_B + 1}  - 2d_B^{-1/4} - 4d_B\exp\left(-\frac{\sqrt{d_B}}{72\pi^3} \right).
    \end{equation}
\end{lemma}

We note that it was shown in~\cite{Schatzki2024} that for any nonempty $s \subseteq B$, we have
\begin{equation}
    \EV_{|\psi\rangle\sim\Haar(d_B)}[C(|\psi\rangle; s)] 
    = 1 - \frac{3^{|s|}(2^{N_B - |s|} + 1)}{2^{|s|}(2^{N_B} + 1)}.
\end{equation}
Using this fact and Lemma~\ref{lemma:avg-E-calc}, we obtain the average post-measurement concentratable entanglement for Haar-random states:
\begin{lemma}\label{lemma:avg-LCE-calc}
    For arbitrary measurement basis $\beta \subset \mathcal{H}_A$ and nonempty subset $s\subseteq B$, the following holds:
    \begin{equation}
        \EV_{|\Psi\rangle \sim \Haar(d)} \left[\overline{C}_\beta(|\Psi\rangle; s)\right] = 1 - \frac{3^{|s|}(2^{N_B - |s|} + 1)}{2^{|s|}(2^{N_B} + 1)}.
    \end{equation}
\end{lemma}

We now have all ingredients in place to prove the main results of Sec.~\ref{sec:concentration}:
\begin{proof}[Proof of Theorem~\ref{Th:EoA-concentration}]
By Corollary \ref{cor:LE-local-bounds}, we have
\begin{align}
    \left|\mathcal{A}^\tau(|\Psi\rangle) - \mathcal{A}^\tau(|\Psi'\rangle) \right| 
    &\leq (2\sqrt{2}+1)\|\Psi - \Psi' \|_1 \\
    &= (4\sqrt{2} + 2)\| |\Psi\rangle - |\Psi'\rangle \|_2,
\end{align}
where in the last line we use the norm inequality $\||\Psi\rangle \langle \Psi| - |\Psi'\rangle\langle \Psi'|\|_1 \leq \||\Psi\rangle - |\Psi'\rangle\|_2$ for states $|\Psi\rangle, |\Psi'\rangle \in \mathcal{H}_A \otimes \mathcal{H}_B$.\footnote{Note that Lipschitz continuity of $\mathcal{A}^\tau(|\Psi\rangle) = F(\Psi_B, \tilde{\Psi}_B)$ with respect to the Euclidean norm on $\mathbb{S}^{2d - 1}$ can also be shown using properties of the quantum fidelity.}

Putting the above work together with Lemma~\ref{lemma:EoA-tau-LB} and applying Levy's Lemma \ref{lemma:levy}, we then get that
\begin{align}
    \Pr\left(\mathcal{A}^\tau(|\Psi\rangle) \leq 1 - \sqrt{\frac{2d_B}{d_A}} - \varepsilon\right) &\leq 
    \Pr\left(\mathcal{A}^\tau(|\Psi\rangle) \leq 
    \EV_{|\Phi'\rangle \sim \Haar(d)}[\mathcal{A}^\tau(|\Phi'\rangle)] - \varepsilon\right) \\
    &\leq 
    \Pr\left(\left|\mathcal{A}^\tau(|\Psi\rangle) - \EV_{|\Psi'\rangle \sim \Haar(d)}[\mathcal{A}^\tau(|\Psi'\rangle)]\right| \geq \varepsilon\right) \\
    &\leq 2\exp\left( - \frac{2d\varepsilon^2}{9\pi^3 (4\sqrt{2} + 2)^2}\right),
\end{align}
which finishes the proof.
\end{proof}

\begin{proof}[Proof of Theorem~\ref{Th:EoA-concentration-2}]
    By referring to the work in the proof of Lemma~\ref{lemma:3-continuity-bounds}, particularly from line~\eqref{eq:first-gme-cont-proof-step} onward, we see that 
    \begin{equation}
        |C^2_{\rm GME}(|\psi\rangle) - C^2_{\rm GME}(|\varphi\rangle)| \leq 2^{3/2} \|\psi - \varphi\|_1 \qquad \text{for }|\psi\rangle, |\varphi\rangle \in \mathcal{H}_B.
    \end{equation}
    For arbitrary state $|\Psi\rangle \in \mathcal{H}_A \otimes \mathcal{H}_B$ and measurement basis $\mathcal{\beta}$, let $\overline{C^2_{\rm GME, \beta}}(|\Psi\rangle)$ denote the average post-measurement value of the squared GME-concurrence. By using Lemma~\ref{lemma:variation-bounds} with the function $f: \mathbb{R}_{\geq 0} \to \mathbb{R}_{\geq 0} $ defined by $x\mapsto 2^{3/2} x$, it follows that for fixed measurement basis $\beta$, the function $\overline{C^2_{\rm GME, \beta}}(\cdot)$ has a Lipschitz constant of at most $2^{5/2} + 1$ with respect to the operator 1-norm.\footnote{We note that $C^2_{\rm GME}$ is a genuine entanglement monotone for pure states because it is the minimum of bipartite linear subsystem entropies, and the minimum of a bipartite entanglement measure across all bipartitions is a legitimate multipartite entanglement measure~\cite{Horodecki2009}.} Therefore, by Levy's Lemma~\ref{lemma:levy} and our lower bound on the expectation value of $\overline{C^2_{\rm GME, \beta}}(|\Psi\rangle)$ from Lemma~\ref{lemma:gme-concurrence-mean-val}, we have for arbitrary fixed measurement basis $\beta$ that
    \begin{multline}
        \Pr_{|\Psi\rangle \sim \Haar(d)} \left(\overline{C^2_{\rm GME, \beta}}(|\Psi\rangle) \leq 1- 2\min_{\varnothing \subsetneq \varnothing \subsetneq B} \frac{d_\gamma + d_{B\setminus \gamma}}{d_B + 1} - 2d_B^{-1/4} - 4d_B\exp\left(-\frac{d_B}{72\pi^3} \right) - \varepsilon \right) \\
        \leq 2\exp\left( - \frac{2d\varepsilon^2}{9\pi^3(2^{7/2}+2)^2} \right). 
    \end{multline}
    Since $C_{\rm GME}$ is upper bounded by one, we have $\overline{C^2_{\rm GME, \beta}}(|\Psi\rangle) \leq \overline{C_{\rm GME, \beta}}(|\Psi\rangle)$, from which the concentration inequality for $\overline{C_{\rm GME, \beta}}(|\Psi\rangle)$ in \eqref{eq:Cgme-concentration} of Theorem~\ref{Th:EoA-concentration-2} follows.

    For the (average post-measurement) concentratable entanglement $\overline{C}_\beta$, we use Lemmas~\ref{lemma:3-continuity-bounds} and~\ref{lemma:variation-bounds} to show the Lipschitz continuity of $\overline{C}_\beta$.
    The proof of the concentration inequality for $\overline{C}_\beta$  in \eqref{eq:CE-concentration} of Theorem~\ref{Th:EoA-concentration-2} then follows similar lines of argument as the previous one by invoking Lemma~\ref{lemma:avg-LCE-calc} to bound the average of $\overline{C}_\beta$. 
\end{proof}

Finally, to see that our statement from the main text regarding equation~\eqref{eq:tau-beta-concentration} holds, one may apply Lemmas~\ref{lemma:3-continuity-bounds} and~\ref{lemma:variation-bounds} regarding the continuity of $\overline{\tau}_\beta$, the expectation value of $\overline{\tau}_\beta$ from Lemma~\ref{lemma:avg-avg-ntangle}, and Levy's lemma to derive a concentration inequality on $\overline{\tau}_\beta$.

\section{Proof of Lemma~\ref{lemma:fidelity-graph-lemma}}\label{app:unweighted-graph-results}

Throughout this appendix, we let the symbol $\sigma_i^{(b)}$ for $b\in B$ denote the Pauli operator $\sigma_i$ acting only on the system associated with $b$. 
For any graph $G$ whose vertex set is $B$, let 
\begin{equation}
    K_b \coloneqq \sigma_x^{(b)}\prod_{b'\in N_b} \sigma_z^{(b')}, 
\end{equation}
where $N_b$ denotes the set of all vertices in $B$ connected to $b$ via an edge of the graph $G$. The graph state $|G\rangle$ is a stabilizer state and the operators $\{K_b: b\in B\}$ form an independent set of stabilizer generators for $|G\rangle$ \cite{Hein2006}:
\begin{equation}
    K_b|G\rangle = |G\rangle.
\end{equation}
From this, one may see that
\begin{equation}\label{eq:generator-identity}
    K_b\sigma_z^{\mathbf{z}}|G\rangle = (-1)^{z_b}\sigma_z^{\mathbf{z}}|G\rangle,
\end{equation}
where $\mathbf{z}\in \mathbb{F}_2^B$ and $z_b$ is the component of $\mathbf{z}$ corresponding to vertex $b$. This relation implies that
\begin{equation}
    \langle G|\sigma_z^{\mathbf{z}}|G\rangle = \delta_{\mathbf{z}, \mathbf{0}}.
\end{equation}
Hence, the set of vectors $\{\sigma_z^{\mathbf{z}}|G\rangle\}_{\mathbf{z}\in \mathbb{F}_2^B}$ forms an orthonormal basis for $\mathcal{H}_B$, which is called the \textit{graph state basis}. 

Another useful fact is that the graph state $|G\rangle$ of Lemma~\ref{lemma:fidelity-graph-lemma} has a reduced state $G_B$ of the form
\begin{equation}\label{eq:reduced-state}
    G_B = \frac{1}{2^{N_A}} \sum_{\mathbf{z} \in \mathbb{F}_2^A} \sigma_z^{\mathbf{\Gamma}_{BA}\mathbf{z}}|G-A\rangle\langle G-A|\sigma_z^{\mathbf{\Gamma}_{BA}\mathbf{z}}.
\end{equation}
As before, the matrix multiplication $\mathbf{\Gamma}_{BA} \mathbf{z}$ is understood to be modulo 2 here. Furthermore, $G_B$ is an unnormalized projector, meaning that $G_B^2 = KG_B$ for some real constant $K$~\cite[Prop.~8]{Hein2006}. This then implies that $G_B$ is maximally mixed on its support. 

Next, we state a general lemma for graph states. Our proof is inspired by the proof of Theorem III.2 in~\cite{schatzki2023}. 
\begin{lemma}\label{lemma:graph-state-sigma-y}
    For any $\mathbf{z} \in \mathbb{F}_2^B$, we have
    \begin{equation}
        \sigma_y^{\otimes N_B}\sigma_z^{\mathbf{z}}|G-A\rangle = C(\mathbf{z}, G-A) \sigma_z^{\mathbf{z} + \mathbf{D}}|G-A\rangle, 
    \end{equation}
    where $C(\mathbf{z}, G - A)$ is phase factor that is equal to either $i^{N_B}$ or $-i^{N_B}$ depending on the vector $\mathbf{z}$ and the subgraph $G - A$. 
\end{lemma}

\begin{proof}
    Take $\mathbf{z}\in \mathbb{F}_2^B$ and $b\in B$ to be arbitrary. We have 
    \begin{align}
        \sigma_y^{(b)}\sigma_z^{\mathbf{z}}|G-A\rangle &= 
        -i \sigma_z^{(b)} \sigma_x^{(b)} \sigma_z^{\mathbf{z}} |G-A\rangle \\
        &= -i\sigma_z^{(b)} \left( \prod_{b' \in N_b} \sigma_z^{(b')} \right)\left( \prod_{b'' \in N_b} \sigma_z^{(b'')} \right) \sigma_x^{(b)}  \sigma_z^{\mathbf{z}} |G-A\rangle \\
        &= -i\sigma_z^{(b)} \left( \prod_{b' \in N_b} \sigma_z^{(b')} \right) \sigma_x^{(b)} \left(\prod_{b'' \in N_b} \sigma_z^{(b'')} \right) \sigma_z^{\mathbf{z}} |G-A\rangle \\
        &= -i\sigma_z^{(b)} \left( \prod_{b' \in N_b} \sigma_z^{(b')} \right) K_b  \sigma_z^{\mathbf{z}} |G-A\rangle \\
        &= -i(-1)^{z_b} \sigma_z^{(b)} \left( \prod_{b' \in N_b} \sigma_z^{(b')} \right) \sigma_z^{\mathbf{z}} |G-A\rangle,
    \end{align}
    where the last line follows from~\eqref{eq:generator-identity}.
    Therefore,
    \begin{equation}\label{eq:yprod-state}
        \sigma_y^{\otimes N_B}\sigma_z^{\mathbf{\Gamma}_{BA} \mathbf{z}}|G-A\rangle = C(\mathbf{z}, G-A) \underbrace{\prod_{b\in B} \sigma_z^{(b)} \left(\prod_{b'\in N_b} \sigma_z^{(b')} \right)}_{P}\sigma_z^{\mathbf{\Gamma}_{BA}\mathbf{z}}|G-A\rangle, 
    \end{equation}
    where $C(\mathbf{z}, G-A)$ is a phase factor either equal to $i^{N_B}$ of $-i^{N_B}$. Note that if vertex $b\in B$ has odd degree, then the product $P$ in~\eqref{eq:yprod-state} contains an even number of $\sigma_z^{(b)}$ operators. On the other hand, if $b\in B$ has even degree, then the product $P$ contains an odd number of $\sigma_z^{(b)}$ operators. Recall now that $\mathbf{D}_b = 1$ if the degree of vertex $b\in B$ is even in the subgraph $G-A$ and zero otherwise. From this, we conclude that
    \begin{equation}
        \sigma_y^{\otimes N_B} \sigma_z^{\mathbf{\Gamma}_{BA} \mathbf{z}} |G- A\rangle =  
        C(\mathbf{z}, G-A) \sigma_z^{\mathbf{\Gamma}_{BA}\mathbf{z} + \mathbf{D}}|G-A\rangle, 
    \end{equation}
    with the expression $\mathbf{\Gamma}_{BA}\mathbf{z} + \mathbf{D}$ evaluated in $\mathbb{F}_2^B$.
\end{proof} 

We are now ready to prove Lemma~\ref{lemma:fidelity-graph-lemma}.

\begin{proof}[Proof of Lemma~\ref{lemma:fidelity-graph-lemma}]
    Equation~\eqref{eq:reduced-state} and the orthogonality of the graph state basis imply that the supports of the states $G_B$ and $\tilde{G}_B$ are given by
    \begin{align}
        \supp(G_B) &=\label{eq:supp-psi-b} \Span(\{\sigma_z^{\mathbf{\Gamma}_{BA}\mathbf{z}}|G-A\rangle\}_{\mathbf{z}\in \mathbb{F}_2^A})\\
        \supp(\tilde{G}_B) &= \Span(\{\sigma_y^{\otimes N_B}\sigma_z^{\mathbf{\Gamma}_{BA}\mathbf{z}}|G-A\rangle\}_{\mathbf{z}\in \mathbb{F}_2^A}),
    \end{align}
    where to obtain the latter we used the fact that graph states have real components in the computational basis. Using Lemma~\ref{lemma:graph-state-sigma-y}, it follows that
    \begin{align}
        \supp(\tilde{G}_B) &=\label{eq:supp-psi-tilde-b} \Span(\{\sigma_z^{\mathbf{\Gamma}_{BA}\mathbf{z} + \mathbf{D}}|G-A\rangle\}_{\mathbf{z}\in \mathbb{F}_2^A}).
    \end{align}

    Now suppose that $\mathbf{\Gamma}_{BA} \mathbf{x} = \mathbf{D}$ does not have a solution $\mathbf{x} \in \mathbb{F}_2^A$. Take $\mathbf{z}, \mathbf{z}' \in \mathbb{F}_2^A$ to be arbitrary. By hypothesis, we must have that $\mathbf{\Gamma}_{BA}(\mathbf{z} + \mathbf{z}') \neq \mathbf{D}$. Recalling that addition is understood to be modulo 2, it follows that $\mathbf{\Gamma}_{BA}\mathbf{z} + \mathbf{D} + \mathbf{\Gamma}_{BA}\mathbf{z'} \neq \mathbf{0}$. Hence, 
    \begin{equation}
        \langle G-A|\sigma_z^{\mathbf{\Gamma}_{BA}\mathbf{z} + \mathbf{D} + \mathbf{\Gamma}_{BA}\mathbf{z'}}|G-A\rangle = 0. 
    \end{equation}
    In view of~\eqref{eq:supp-psi-b} and~\eqref{eq:supp-psi-tilde-b}, this implies that $G_B$ and $\tilde{G}_B$ have orthogonal support, and hence $F(G_B, \tilde{G}_B) = 0$.

    On the other hand, suppose that there does exist $\mathbf{x} \in \mathbb{F}_2^A$ such that $\mathbf{\Gamma}_{BA}\mathbf{x} = \mathbf{D}$. We then have
    \begin{equation}
        \{\sigma_z^{\mathbf{\Gamma}_{BA}\mathbf{z}}|G-A\rangle\}_{\mathbf{z}\in \mathbb{F}_2^A} 
        =
        \{\sigma_z^{\mathbf{\Gamma}_{BA}(\mathbf{z}+\mathbf{x})}|G-A\rangle\}_{\mathbf{z}\in \mathbb{F}_2^A} 
        =
        \{\sigma_z^{\mathbf{\Gamma}_{BA}\mathbf{z} + \mathbf{D}}|G-A\rangle\}_{\mathbf{z}\in \mathbb{F}_2^A},
    \end{equation}
    where the first equality follows from the fact that $\mathbf{z}\mapsto \mathbf{z} + \mathbf{x}$ is a bijection from $\mathbb{F}_2^A$ to itself. 
    Therefore, by equations~\eqref{eq:supp-psi-b} and~\eqref{eq:supp-psi-tilde-b}, we have
    $\supp(G_B) = \supp(\tilde{G}_B)$. The fact that $G_B$ is maximally mixed on its support implies the same for $\tilde{G}_B$. But since the two states have the same support, it follows that $G_B = \tilde{G}_B$. Hence, $F(G_B, \tilde{G}_B) = 1$. 
\end{proof}

\section{Remarks on weighted graph states}\label{app:weighted-graph-results}

In this section, we given an expression and a bound for $\|G_{\varphi} - G_{\chi}\|_1$ that help clarify the behavior of the maximum probability with which one can extract a GHZ state from a uniformly weighted graph state $|G_{\theta}\rangle$ for $\theta \in [0,2\pi)$.

For a fixed graph $G = (V, E)$, suppose that we are given two uniformly weighted $N$-qubit graph states $|G_\varphi\rangle, |G_{\chi}\rangle$ with edge weights $\varphi$ and $\chi$, respectively. 
Let $d = 2^N$. As discussed in~\cite[Section 9]{Hein2006}, we have for $\theta \in \{\varphi, \chi\}$ that
\begin{equation}
    |G_\theta\rangle = \frac{1}{\sqrt{d}} \sum_{\mathbf{z}\in \{0,1\}^N} e^{i \theta \mathbf{z}^T \mathbf{\Gamma} \mathbf{z}/2} |\mathbf{z}\rangle,
\end{equation}
where $\mathbf{\Gamma}$ is the adjacency matrix of the underlying (unweighted) graph. Note that in this instance, the matrix multiplication $\mathbf{z}^T\mathbf{\Gamma}\mathbf{z}$ is understood to be over the complex field $\mathbb{C}$ rather than $\mathbb{F}_2^{A}$, as in Section~\ref{sec:graph-states}.
It follows that
\begin{equation}\label{eq:wgs-trace-dist-exact}
    \|G_\varphi - G_{\chi}\|_1 = 2\sqrt{1 - |\langle G_{\chi}|G_{\varphi}\rangle|^2} = 2\sqrt{1 - \frac{1}{d^2}\left|\sum\nolimits_{\mathbf{z}\in \{0,1\}^N} \exp(i(\varphi - \chi)\mathbf{z}^T \mathbf{\Gamma} \mathbf{z}/2)\right|^2}.
\end{equation}
Hence, $\|G_{\varphi} - G_{\chi}\|_1$ may be computed from a sum of $2^N$ terms involving simple computations on an $N\times N$ adjacency matrix in place of more complicated computations involving $2^N \times 2^N$ density matrices. 

We may approximate $\| G_\varphi - G_{\chi} \|_1$ in the interval $|\varphi - \chi| < \min_{\mathbf{z} \in \mathbb{F}_2^N} (\mathbf{z}^T\mathbf{\Gamma}\mathbf{z}/2)^{-1}$ by simply using the power series expansion of the exponential functions in~\eqref{eq:wgs-trace-dist-exact}:
\begin{align}
    \left |\sum_{\mathbf{z} \in \{0,1\}^N} \exp\left(i(\varphi - \chi) \mathbf{z}^T \mathbf{\Gamma} \mathbf{z}/2 \right) \right| 
    &= \left |\sum_{\mathbf{z}\in \{0,1\}^N} \left(1 + \frac{i}{2}(\varphi - \chi)\mathbf{z}^T\mathbf{\Gamma} \mathbf{z} + \varepsilon_{\mathbf{z}}\right) \right|  \\
    &\geq \left |\sum_{\mathbf{z}\in \{0,1\}^N} 1\right| - \left|\sum_{\mathbf{z}\in\{0,1\}^N} \frac{1}{2}|\varphi - \chi|\mathbf{z}^T\mathbf{\Gamma} \mathbf{z}\right| - \sum_{\mathbf{z}\in\{0,1\}^N} |\varepsilon_{\mathbf{z}}| \\
    &= d - \frac{1}{2}|\varphi - \chi|\sum_{\mathbf{z}\in \{0,1\}^N} \mathbf{z}^T\mathbf{\Gamma} \mathbf{z} - \sum_{\mathbf{z}\in\{0,1\}^N} |\varepsilon_{\mathbf{z}}| \\
    &\geq\label{ineq:edge-bound} d - d|E||\varphi - \chi| - \sum_{\mathbf{z}\in \{0,1\}^N} |\varepsilon_{\mathbf{z}}| \\
    &\geq\label{ineq:Lagrange} d - d|E||\varphi - \chi| - |\varphi - \chi|^2\sum_{\mathbf{z}\in \{0,1\}^N} \left( \frac{\mathbf{z}^T\mathbf{\Gamma} \mathbf{z}}{2} \right)^2 \\
    &\geq\label{ineq:final-trunc-error-bound} d - d|E||\varphi - \chi| - d|E|^2|\varphi - \chi|^2,
\end{align}
where $\varepsilon_{\mathbf{z}}$ is the truncation error for our expansion of $\exp(i(\varphi - \chi)\mathbf{z}^T\mathbf{\Gamma} \mathbf{z})$. A power expansion of the exponential function together with a Cauchy estimate show that $|e^z - (1+z)| \leq |z|^2$ for all $|z| < 1$, implying $|\varepsilon_{\mathbf{z}}| \leq |\varphi- \chi|^2 (\mathbf{z}^T\mathbf{\Gamma}\mathbf{z}/2)^2$ whenever $|\varphi- \chi|^2 (\mathbf{z}^T\mathbf{\Gamma}\mathbf{z}/2)^2 < 1$. Inequalities~\eqref{ineq:edge-bound} and~\eqref{ineq:final-trunc-error-bound} follow from the fact that $z^T \mathbf{\Gamma} z \leq \sum_{i,j} \mathbf{\Gamma}_{ij} = 2|E|$ by the definition of the graph adjacency matrix, while inequality~\eqref{ineq:Lagrange} follows from our estimate $|\varepsilon_{\mathbf{z}}| \leq |\varphi- \chi|^2 (\mathbf{z}^T\mathbf{\Gamma}\mathbf{z}/2)^2$ on the interval $|\varphi- \chi|<  (\mathbf{z}^T\mathbf{\Gamma}\mathbf{z}/2)^{-1}$.

Hence, on the interval $|\varphi - \chi| < \min_{\mathbf{z} \in \mathbb{F}_2^N} (\mathbf{z}^T \mathbf{\Gamma} \mathbf{z}/2)^{-1}$,
\begin{align}
    \|G_\varphi - G_\chi \|_1 
    &= 2\sqrt{1 - \frac{1}{d^2}\left |\sum_{z \in \{0,1\}^N} \exp\left(i(\varphi - \chi) \mathbf{z}^T \mathbf{\Gamma} \mathbf{z} \right) \right|^2} \\
    &\leq 2\sqrt{1 - \left(1 - |E||\varphi - \chi| - |E|^2|\varphi - \chi|^2\right)^2} \\
    &= 2\sqrt{2\left(|E||\varphi - \chi| + |E|^2|\varphi - \chi|^2\right) - \left(|E||\varphi - \chi| + |E|^2|\varphi - \chi|^2\right)^2} \\
    &= 2\sqrt{2|E||\varphi - \chi|}\sqrt{1 -  |E|^2|\varphi - \chi|^2 -|E|^3|\varphi - \chi|^3} \\
    &= 2\sqrt{2|E||\varphi - \chi|}\left(1 - \mathcal{O}\left(|E|^2|\varphi - \chi|^2\right)\right). \label{eq:trace-dist-estimate}
\end{align}

As shown in the main text, the maximum probability $p_{\tau = 1}(\varphi)$ of extracting any $\tau = 1$ state from a uniformly weighted graph state $|G_{\varphi}\rangle$ with underlying graph $G$ is upper bounded in terms of the trace distance $\| |G_{\varphi}\rangle \langle G_{\varphi}| - |G\rangle\langle G| \|_1$ via 
\begin{equation}\label{eq:p-t-1-bound}
    p_{\tau = 1}(\varphi) \leq (2\sqrt{2} + 1)\|G_\varphi - G\|_1
\end{equation}
whenever the matrix equation of Theorem~\ref{Th:main-graph-thm} has no solution (see~\eqref{eq:p-ghz-bound}). By substituting~\eqref{eq:trace-dist-estimate} into~\eqref{eq:p-t-1-bound}, we then obtain an upper bound on the decay of $p_{\tau = 1}(\varphi)$ to zero whenever the noise parameter $\varphi$ satisfies $|\varphi - \pi/2| < \min_{\mathbf{z} \in \mathbb{F}_2^N} (\mathbf{z}^T \mathbf{\Gamma} \mathbf{z}/2)^{-1}$.

\end{document}